\documentclass[11pt]{article}
\usepackage[utf8]{inputenc}
\usepackage{amsmath,amsfonts,amssymb}
\usepackage{amsthm}
\usepackage{latexsym}
\usepackage[noadjust]{cite}
\usepackage{graphicx}
\usepackage{xcolor}
\usepackage{mathtools}
\usepackage[margin=1in]{geometry}
\usepackage{thm-restate}
\usepackage{enumitem}
\usepackage[ruled,vlined,linesnumbered]{algorithm2e}
\usepackage{tikz}
\usepackage{subcaption}
\usepackage[framemethod=TikZ]{mdframed}

\definecolor{ForestGreen}{rgb}{0.1333,0.5451,0.1333}
\definecolor{DarkRed}{rgb}{0.8,0,0}
\definecolor{Red}{rgb}{1,0,0}
\usepackage[linktocpage=true,
pagebackref=true,colorlinks,
linkcolor=DarkRed,citecolor=ForestGreen,
bookmarks,bookmarksopen,bookmarksnumbered]
{hyperref}

\usepackage[capitalise]{cleveref}

\newtheorem{theorem}{Theorem}[section]

\newtheorem{lemma}[theorem]{Lemma}
\newtheorem{observation}[theorem]{Observation}

\newtheorem{claim}[theorem]{Claim}
\theoremstyle{definition}

\newtheorem{remark}[theorem]{Remark}

\newtheorem{assumption}[theorem]{Assumption}
\newtheorem{invariant}[theorem]{Invariant}

\newcommand{\shortcut}{\mathtt{Shortcut}}
\newcommand{\scin}{\mathtt{ShortcutIn}}
\newcommand{\scout}{\mathtt{ShortcutOut}}
\newcommand{\ein}{\mathtt{AddEdgeIn}}
\newcommand{\eout}{\mathtt{AddEdgeOut}}
\newcommand{\findpathin}{\mathtt{FindPathIn}}
\newcommand{\findpathout}{\mathtt{FindPathOut}}
\newcommand{\estimate}{\mathtt{Estimate}}

\newcommand{\Nin}{N^{\mathrm{in}}}
\newcommand{\Nout}{N^{\mathrm{out}}}

\newcommand{\Vin}{V^{\mathrm{in}}}
\newcommand{\Vout}{V^{\mathrm{out}}}
\newcommand{\Fin}{F^{\mathrm{in}}}
\newcommand{\Fout}{F^{\mathrm{out}}}
\newcommand{\Uin}{U^{\mathrm{in}}}
\newcommand{\Uout}{U^{\mathrm{out}}}
\newcommand{\Ain}{A^{\mathrm{in}}}
\newcommand{\Aout}{A^{\mathrm{out}}}
\newcommand{\Bin}{B^{\mathrm{in}}}
\newcommand{\Bout}{B^{\mathrm{out}}}
\newcommand{\Tin}{T^{\mathrm{in}}}
\newcommand{\Tout}{T^{\mathrm{out}}}
\newcommand{\Din}{D^{\mathrm{in}}}
\newcommand{\Dout}{D^{\mathrm{out}}}
\newcommand{\Sin}{S^{\mathrm{in}}}
\newcommand{\Sout}{S^{\mathrm{out}}}

\newcommand{\pin}{p^{\mathrm{in}}}
\newcommand{\pout}{p^{\mathrm{out}}}
\newcommand{\pn}{p^{\mathrm{N}}}
\newcommand{\win}{w^{\mathrm{in}}}
\newcommand{\wout}{w^{\mathrm{out}}}

\newenvironment{ourbox}{\begin{mdframed}[hidealllines=false,innerleftmargin=10pt,backgroundcolor=white!10,innertopmargin=5pt,innerbottommargin=5pt,roundcorner=10pt]}{\end{mdframed}}

\usetikzlibrary{arrows.meta,calc,decorations.markings}
\tikzset{
  midarrow/.style={
    postaction={decorate},
    decoration={
      markings,
      mark=at position 0.5 with {\arrow{Stealth}}
    }
  },
  dot/.style={circle,fill=black,inner sep=0.8pt},
  bpath/.style={black,thin,-{Stealth[length=2mm]}},
  rseg/.style={red,thick},
  rthin/.style={red,thin,-{Stealth[length=2mm]}},
  rcurve/.style={violet,thick,-{Stealth[length=3mm]}},
  dashedlvl/.style={blue!80,dashed,thin},
}
\newcommand{\segmentdots}[4]{
  \foreach \i in {1,...,#3}{
    \node[#4] at ($ (#1)!{\i/(#3+1)}!(#2) $) {};
  }
}

\title{An $n^{2+o(1)}$ Time Algorithm for Single-Source Negative Weight Shortest Paths}
\author{Sanjeev Khanna\thanks{Courant Institute, Warren Weaver Hall, New York University, 251 Mercer Street, New York, NY 10012. Email: {\tt sanjeev.khanna@nyu.edu}.} \and Junkai Song\thanks{Courant Institute, Warren Weaver Hall, New York University, 251 Mercer Street, New York, NY 10012. Email: {\tt junkaisong@nyu.edu}. }}
\date{\today}

\begin{document}

\maketitle

\begin{abstract}
We present a randomized algorithm for the single-source shortest path (SSSP) problem on directed graphs with arbitrary real-valued edge weights that runs in $n^{2+o(1)}$ time with high probability. This result yields the first almost linear-time algorithm for the problem on dense graphs ($m = \Theta(n^2)$) and improves upon the best previously known bounds for moderately dense graphs ($m = \omega(n^{1.306})$).

Our approach builds on the ``hop-reduction via shortcutting'' framework introduced by Li, Li, Rao, and Zhang (2025), which iteratively augments the graph with shortcut edges to reduce the negative hop count of shortest paths. The central computational bottleneck in prior work is the cost of explicitly constructing these shortcuts in dense regions. We overcome this by introducing a new compression technique using auxiliary Steiner vertices. Specifically, we construct these vertices to represent large neighborhoods compactly in a structured manner, allowing us to efficiently generate and propagate shortcuts while strictly controlling the growth of vertex degrees and graph size.
\end{abstract}

\section{Introduction}
We study the \emph{single-source shortest paths} (SSSP) problem in directed graphs with real-valued edge weights. The input is a directed graph $G=(V,E)$ with edge weights $w:E\to \mathbb{R}$ and a source vertex $s\in V$. The goal is to compute the distances $d_G(s,v)$ for all $v\in V$, or report the existence of a negative cycle.

Shortest path computation is among the most fundamental problems in graph algorithms. When all edge weights are non-negative, Dijkstra's algorithm solves SSSP in
$O(m+n\log n)$ time~\cite{Dij59,FT87}. A sequence of recent breakthroughs has further improved the running time to $O(m\sqrt{\log n\log \log n})$ for undirected graphs, due to Duan, Mao, Shu, and Yin~\cite{DMSY23}, and for directed graphs, due to Duan {\em et al.}~\cite{DMM+25,DMSY26}. For integer (possibly negative) weights, near-linear running times have been achieved by Bernstein, Nanongkai, and Wulff-Nilsen~\cite{BNW22}, which was further improved by Bringmann, Cassis, and Fischer~\cite{BCF23}, and by Li, Mowry, and Rao~\cite{LMR25}. At the same time, when the weights are polynomially bounded integers, almost-linear-time algorithms for the minimum-cost flow problem, a generalization of single-source shortest paths, were given by Chen {\em et al.}~\cite{CKL+25} and by van den Brand {\em et al.}~\cite{vdBCK+24}.

In contrast, for \emph{real-valued} weights with negative edges, the classic Bellman--Ford dynamic programming algorithm~\cite{Shi55,For56,Bel58,Moo59}, which runs in $O(mn)$ time, remained the fastest algorithm for decades. This barrier was broken by Fineman~\cite{Fin24}, who gave a randomized $\widetilde{O}(mn^{8/9})$ time algorithm. Subsequent work by Huang, Jin, and Quanrud improved the running time to $\widetilde{O}(mn^{4/5})$~\cite{HJQ25b}, and later to $\widetilde{O}(mn^{3/4}+m^{4/5}n)$~\cite{HJQ25a}, further advancing Fineman’s framework.

Very recently, Li, Li, Rao, and Zhang~\cite{LLRZ25} introduced a new shortcutting framework, which augments the graph with shortcut edges that reduce the number of negative edges appearing along shortest paths. Using this approach, they obtained a $\widetilde{O}(n^{2.5})$ time algorithm for real-weighted SSSP, improving prior bounds in the dense regime. In parallel, Quanrud and Tajkhorshid~\cite{QT25} developed a randomized sparsification technique that further advances Fineman’s framework. When combined with the recursive betweenness reduction of~\cite{LLRZ25}, their method achieves $O(mn^{0.69562})$ time for $m \ge n^{1.0274}$ and $O((mn)^{0.85})$ time for $m\leq n^{1.0274}$.

\paragraph{Our result.}
Our main result is an almost-quadratic-time algorithm for real-weighted single-source shortest paths. Building on the shortcutting framework of~\cite{LLRZ25}, we introduce a new shortcutting approach that uses auxiliary Steiner vertices. These auxiliary vertices compactly represent large vertex neighborhoods, allowing shortcut edges to be generated and propagated efficiently. Our algorithm maintains several structural invariants involving these Steiner vertices that tightly control the growth in the number of vertices and edges in the graph.

\begin{theorem}
\label{thm:main}
There is a randomized algorithm that computes single-source shortest paths in real-weighted graphs in $n^{2+o(1)}$ time with high probability.
\end{theorem}

In the dense regime, \cref{thm:main} establishes an almost-optimal running time, matching the $\Omega(n^2)$ input size barrier up to subpolynomial factors. Moreover, our algorithm improves upon the best known bounds of~\cite{LLRZ25,QT25} for moderately dense graphs, specifically when $m=\omega(n^{1.306})$.

\paragraph{Organization.}
In \cref{sec:prelim}, we introduce notation and basic definitions.
\cref{sec:tool} presents key algorithmic ingredients used by our approach.
Finally, \cref{sec:alg} presents our main result, namely, the $n^{2+o(1)}$-time algorithm for SSSP.
We begin with an overview of the ideas underlying the algorithm, and then develop our shortcutting procedure and the maintenance of the structural invariants in detail.
We conclude this section with the correctness proof and running time analysis.

\paragraph{Note.} We very recently learned that a similar result has been independently obtained by George Z. Li, Jason Li, and Junkai Zhang. Their work appears on arXiv simultaneously with ours.

\section{Preliminaries}
\label{sec:prelim}

\paragraph{Notations.} All logarithms in this paper are to base 2. Consider a directed graph $G=(V,E)$. For every vertex $v\in V$, the vertices in the set $\Nin(v)=\{u\in V:(u,v)\in E\}$ are called the \emph{in-neighbors} of $v$, and the vertices in the set $\Nout(v)=\{u\in V:(v,u)\in E\}$ are called the \emph{out-neighbors} of $v$. For any pair of vertices $u,v\in V$, let $d(u,v)$ denote the shortest path distance from $u$ to $v$. Given a path $P$ in $G$ and vertices $u,v$ lying on $P$ with $v$ appearing after $u$, we denote by $P[u,v]$ the subpath of $P$ from $u$ to $v$.

\paragraph{Hop distance.} The number of \emph{hops} in a path is the number of negative edges it traverses. For vertices $u,v\in V$, the \emph{$h$-hop distance} $d^h(u,v)$ is defined as the minimum length of any path from $u$ to $v$ that uses at most $h$ hops. Given a source vertex $s\in V$, one can compute $d^h(s,v)$ for every $v\in V$ in time $O(h(m+n\log n))$ using a hybrid of Dijkstra and Bellman-Ford.

\paragraph{Vertex potential.} For a vertex potential function $\varphi:V\to\mathbb R$, we define $w_{\varphi}(u,v)=w(u,v)+\varphi(u)-\varphi(v)$ for every edge $(u,v)\in E$. We say $\varphi$ is \emph{valid} if $w_\varphi(e)\geq 0$ for every edge $e\in E$ with $w(e)\geq 0$. Given a graph $G$ and a vertex potential function $\varphi$, let $G_{\varphi}$ denote the graph with the same vertex and edge sets as $G$, but with edge weights given by $w_{\varphi}$.

\paragraph{Well-behaved graph.} We call a weighted directed graph $G$ \emph{well-behaved} if for every negative edge $(r,\bar r)$ in $G$, $r$ has out-degree 1 and no incoming negative edges, and $\bar r$ has in-degree 1 and no outgoing negative edges. Moreover, there is an edge $(\bar r,r)$ of weight $-w(r,\bar r)$.

The following reduction given in~\cite{Fin24} converts an arbitrary directed graph $G=(V,E)$ to a well-behaved graph $G'$. For each vertex $v\in V$, we define
\[
w^{\min}_{v}=\min\{0,\min\{w(v,u):u\in V,(v,u)\in E\}\}.
\]
We construct $G'$ by making a copy $\bar v$ for every vertex $v\in V$, connected by an edge $(v,\bar v)$ of weight $w^{\min}_{v}$ and an edge $(\bar v,v)$ of weight $-w^{\min}_v$. Each original edge $(u,v)\in E$ of weight $w(u,v)$ is then replaced by an edge $(\bar u,v)$ with weight $w(u,v)-w^{\min}_u$. Observe that the only negative edges in $G'$ are the edges $(v,\bar v)$ for vertices $v\in V$ with $w^{\min}_v<0$. It is possible that for some vertices $v\in V$, we have $w(v,\bar v)=w^{\min}_v=0$. Nevertheless, $(v,\bar v)$ is still considered as a negative edge and counted as one hop in shortest path computations. It is immediate that $G'$ is well-behaved.

In the graph $G'$, we define $N=V$ as the set of \emph{negative vertices} (i.e. vertices with outgoing negative edge), and define $\bar N=\{\bar v:v\in V\}$. 
If $G$ has $n$ vertices, then $G'$ has $2n$ vertices. It is straightforward to verify that $d_G(u,v)=d_{G'}{(u,v)}$ for every $u,v\in V$. Therefore, we apply this transformation to the input graph $G$, and hereafter assume that $G$ is well-behaved.

\section{Useful Algorithmic Tools} \label{sec:tool}

In a well-behaved graph $G=(V,E)$, for any negative vertex $r\in N$, we say $r$ is $h$-hop in between vertices $(u,v)$ for some $u,v\in V$ if $d^h(u,r)+d^h(r,v)<0$. The \emph{$h$-hop betweenness} of a pair $(u,v)\in V\times V$ is defined as the number of negative vertices that are $h$-hop in between $(u,v)$.

\begin{lemma}[Lemma 11 of \cite{LLRZ25}] \label{lem:betweenness-reduction}
Let $G$ be a well-behaved graph with $n$ vertices and $\eta$ negative vertices. For any parameter $b\geq 1$ and positive integer $h$, there is a randomized algorithm that returns either a negative cycle, or valid vertex potential $\varphi$, such that with high probability, all pairs $(u,v)\in V\times V$ have $h$-hop betweenness at most $\eta/b$ in graph $G_{\varphi}$. The algorithm makes one call to negative-weight single-source shortest paths on a graph with $O(hn)$ vertices and $O(b\log n)$ negative vertices.
\end{lemma}

The following lemma follows from a careful implementation of the Dijkstra’s algorithm.

\begin{lemma}[Lemma 8 of \cite{LLRZ25}] \label{lem:compute-U}
Let $G$ be a well-behaved graph with $n$ vertices. Suppose for every negative vertex $r\in N$, we are given a value $\Delta_r$. Then there is an algorithm that, for every negative vertex $r\in N$, computes two sets
\[
\Uin_r=\{v\in V:d^0(v,r)<\Delta_r\}, \quad \Uout_r=\{v\in V:d^1(r,v)<-\Delta_r\}.
\]
The algorithm also returns values of $d^0(v,r)$ for all $v\in \Uin_r$ and $d^1(r,v)$ for all $v\in \Uout_r$. The running time is $O((|\Uin_r|+|\Uout_r|)^2+(|\Uin_r|+|\Uout_r|)\log n)$.
\end{lemma}

We also establish the following lemma to compute suitable values of $\Delta_r$ for each $r \in N$, ensuring that the sets $\Vin_r$ and $\Vout_r$ (defined in the lemma) have bounded size. In contrast to the $0$-hop and $1$-hop distances considered in \cite{LLRZ25}, we extend the result to arbitrary $h$-hop distances. A key challenge is that even if $\Vin_r$ and $\Vout_r$ is small, explicitly computing $h$-hop shortest paths from every $r\in N$ would require $O(\eta \cdot hm)$ time. Instead, we use a randomized sampling approach to estimate the thresholds such that the balls around each $r \in N$ fall into $\log n$ different scales, and then carefully select an appropriate radius.

\begin{lemma} \label{lem:compute-delta}
Let $G$ be a well-behaved graph with $n$ vertices and $\eta$ negative vertices. For every negative vertex $r\in N$ and a given value $\Delta_r$, define
\[
\Vin_r=\{v\in V:d^h(v,r)<\Delta_r\} ,\quad \Vout_r=\{v\in V:d^h(r,v)<-\Delta_r\}.
\]
After applying the betweenness reduction of \cref{lem:betweenness-reduction} with parameters $b$ and $h$, there is an algorithm that computes $\Delta_r$ for every negative vertex $r\in N$ such that
\[
\sum_{r\in N}\left(|\Vin_r|+|\Vout_r|\right)^2 = O\left( \frac{\eta n^2}{b}\right), \quad \sum_{r\in N}\left(|\Vin_r|+|\Vout_r|\right)\le O\left(\frac{\eta n}{\sqrt{b}}\right)
\]
holds. The running time of the algorithm is $\tilde O(hn^2)$, independent of parameter $b$.
\end{lemma}

\subsection{Proof of \cref{lem:compute-delta}}
For any pair of vertex $u,v\in V$, we define their \emph{extended distance} $\hat d(u,v)$ as the triple $(d^h(u,v),u,v)$ where the first entry is the $h$-hop distance, and the second and third entries are the index of the source and sink vertex respectively which serve as a tie-breaker. The extended distances are compared lexicographically. For an extended distance triple $D$, we use $d(D)$ to denote its first entry. For a subset of vertices $S\subseteq V$, we define $\hat d(S,v)=\min_{u\in S} \hat d(u,v)$ and $\hat d(v,S)=\min_{u\in S} \hat d(v,u)$. Given any $S\subseteq V$, both $\hat d(S,v)$ and $\hat d(v,S)$ for every $v\in V$ can be computed in $O(h(m+n\log n))$ time using a hybrid of Dijkstra and Bellman-Ford.

\begin{algorithm}
\caption{$\estimate(p)$}
\label{alg:estimate}
\For{$O(\log n)$ iterations}{
Sample a random set $S\subseteq V$ by including each vertex $v\in V$ independently with probability $p$.\\
Compute $\hat d(S,v)$ and $\hat d(v,S)$ for every vertex $v\in V$.\\
}
For every negative vertex $r\in N$, define $\Din_r$ (resp. $\Dout_r$) as the median of $\hat d(S,r)$ (resp. $\hat d(r,S)$) computed over the $O(\log n)$ experiments.
\end{algorithm}

\paragraph{An estimation algorithm.} Consider the estimation algorithm given in \cref{alg:estimate}. For any negative vertex $r\in N$ and extended distance $D$, define
\[
\Bin_r(D)=\{v\in V:\hat d(v,r)\leq D\}, \quad \Bout_r(D)=\{v\in V:\hat d(r,v)\leq D\}.
\]
We prove the following claim.

\begin{claim} \label{clm:estimate}
For any negative vertex $r\in N$, let $\Din_r$ and $\Dout_r$ be the values obtained from $\estimate(p)$ for some $4/n\leq p\leq 1/2$. Then
\[
\frac{1}{8p}\leq \left|\Bin_r(\Din_r)\right| \leq \frac{4}{p}, \quad \frac{1}{8p}\leq \left|\Bout_r(\Dout_r)\right| \leq \frac{4}{p}
\]
with high probability.
\end{claim}
\begin{proof}
Let
\[
L_p=\left\lfloor\frac{1}{8p}\right\rfloor, R_p=\left\lfloor\frac{4}{p}\right\rfloor \geq \frac{2}{p}.
\]
Reorder the vertices in $V$ so that $\hat d(v_1,r)<\hat d(v_2,r)<\dots<\hat d(v_n,r)$. We say a random subset $S\subseteq V$ is \emph{good} if
\[
\frac{1}{8p}\leq \left|\Bin_r(\hat d(S,r))\right| \leq \frac{4}{p}.
\]
We first show that a single random experiment produces a good set $S$ with constant probability.

The probability that none of $v_1,\dots,v_{L_p}$ is included in the set $S$ is
\[
(1-p)^{L_p} \geq \exp(-2pL_p) \geq \exp(-1/4) \geq 0.77,
\]
where we use the inequality $1-x \ge e^{-2x}$ for $x \in [0,\tfrac12]$. On the other hand, the probability that none of $v_1,\dots,v_{R_p}$ is included in $S$ is
\[
(1-p)^{R_p} \leq \exp(-pR_p) \leq \exp(-2)\leq0.15.
\]
Therefore, with probability at least $0.77-0.15\geq 0.62$, we have $\hat d(S,r)=\min_{u\in S} \hat d(u,r)=\hat d(v_i,r)$ for some $L_p <i\leq R_p$. This implies that $S$ is a good set.

We repeat this experiment $O(\log n)$ times and take the median of the resulting values of $\hat d(S,r)$.
By a standard Chernoff bound, with high probability at least half of the iterations are good.
Hence the median value $\Din_r$ is also good.
Therefore, with high probability, we have
\[
\frac{1}{8p}\leq \left|\Bin_r(\Din_r)\right| \leq \frac{4}{p}.
\]

By an analogous argument for outgoing distances, with high probability,

\[
\frac{1}{8p}\leq \left|\Bout_r(\Dout_r)\right| \leq \frac{4}{p}
\]
Taking a union bound over all negative vertices completes the proof.
\end{proof}

\paragraph{Obtain extended distances.} Let $L=\lfloor \log n\rfloor-2$. For every $\ell \in [L]$, we invoke $\estimate(2^{-\ell})$ and obtain $(\Din_r)_\ell,(\Dout_r)_{\ell}$ for every negative vertex $r\in N$. 

Fix a negative vertex $r\in N$. For notational convenience, we omit the subscript $r$ and write $\Din_1,\dots,\Din_L$ and $\Dout_1,\dots,\Dout_L$ for the corresponding values. By \cref{clm:estimate}, we see that with high probability, for every $\ell \in [L]$,
\begin{equation} \label{eq:size}
2^{\ell-3}\leq \left|\Bin_r(\Din_\ell)\right| \leq 2^{\ell+2}, \quad 2^{\ell-3}\leq \left|\Bout_r(\Dout_\ell)\right| \leq 2^{\ell+2}
\end{equation}
holds. We condition on this event from here onwards.

We first note that for every $\ell\in [L]$, we can set $\Din_{\ell}=\max_{i\leq \ell} \Din_i$ and $\Dout_{\ell}=\max_{i\leq \ell} \Dout_i$ and the properties in \cref{eq:size} still hold. Indeed, suppose $\Din_i>\Din_j$ for some $i<j$, then we have
\[
2^{j-3}\leq \left|\Bin_r(\Din_j)\right|\leq \left|\Bin_r(\Din_i)\right|\leq 2^{i+2}\leq 2^{j+2}.
\]
Hence, replacing $\Din_j$ by $\Din_i$ preserves the bound in \cref{eq:size}. The argument for the outgoing case is analogous. \emph{We can therefore assume both $\{\Din_\ell\}_{\ell\in [L]}$ and $\{\Dout_\ell\}_{\ell\in [L]}$ are non-decreasing.}

\paragraph{Determine $\Delta_r,\Tin_r,\Tout_r$.}
We now determine the value of $\Delta_r$. We will choose $\Delta_r$ using the discretized radii $\{\Din_\ell\},\{\Dout_\ell\}$.
As a first step, we will define auxiliary sets $\Tin_r,\Tout_r$ that satisfy the following properties:

\begin{enumerate}[label=P\arabic*.]
    \item $\Vin_r\subseteq \Tin_r$ and $\Vout_r\subseteq \Tout_r$.
    \item For every $u\in \Tin_r$ and $v\in \Tout_r$, $d^h(u,r)+d^h(r,v)<0$.
    \item $|\Tin_r|=\Theta(|\Tout_r|)$.
\end{enumerate}

These properties will enable a charging argument using betweenness reduction.
We first show a useful claim.
\begin{claim} \label{clm:contain}
For any $\ell\in [L]$, we have:
\begin{itemize}
    \item If $\Delta_r\leq d(\Din_\ell)$, then $\Vin_r\subseteq \Bin_r(\Din_\ell)$.
    \item If $\Delta_r>d(\Din_\ell)$, then $\Vin_r\supseteq \Bin_r(\Din_\ell)$.
    \item If $-\Delta_r\leq d(\Dout_\ell)$, then $\Vout_r\subseteq \Bout_r(\Dout_\ell)$.
    \item If $-\Delta_r> d(\Dout_\ell)$, then $\Vout_r\supseteq \Bout_r(\Dout_\ell)$.
\end{itemize}
\end{claim}
\begin{proof}
We show the first two items; the proofs of the third and fourth items are analogous.
If $\Delta_r\leq d(\Din_\ell)$, then for every $v\in \Vin_r$, we have $d^h(v,r)<\Delta_r\leq d(\Din_\ell)$. Thus, $\hat d(v,r)<\Din_\ell$ and $v$ must be in $\Bin_r(\Din_\ell)$. On the other hand, if $\Delta_r>d(\Din_\ell)$, then for every $v\in \Bin_r(\Din_\ell)$, we have $d^h(v,r)\leq d(\Din_\ell)<\Delta_r$. Thus, $v$ must be in $\Vin_r$.
\end{proof}

Let $\ell^\star$ be the smallest index such that $d(\Din_{\ell^\star})+d(\Dout_{\ell^\star})\ge 0$. 
We note that the choice of the index $\ell^\star$ is made to control the
sizes of the sets $\Tin_r$ and $\Tout_r$ via the betweenness-reduction guarantee
of \cref{lem:betweenness-reduction}, as captured by Property~P2.
In particular, the definition of $\ell^\star$ ensures that these sets can be chosen
to be sufficiently large to contain $\Vin_r$ and $\Vout_r$, yet can be bounded in size using \cref{lem:betweenness-reduction}. 
We now consider four possible cases that may arise:
\begin{itemize}
    \item No such an index exists; so $d(\Din_L)+d(\Dout_L)<0$. Let $\Delta_r=d(\Din_L),\Tin_r=\Bin_r(\Din_L),\Tout_r=\Vout_r$. Then by \cref{clm:contain}, we have $\Vin_r\subseteq \Bin_r(\Din_L)=\Tin_r$ and $\Tout_r=\Vout_r\supseteq \Bout_r(\Dout_L)$ since $-\Delta_r=-d(\Din_L)>d(\Dout_L)$. It is clear that Property P1 holds. For P2, we have $d^h(u,r)\leq d(\Din_L)$ for every $v\in\Tin_r$ and $d^h(r,v)<-\Delta_r=-d(\Din_L)$ for every $v\in \Tout_r$. Therefore, $d^h(u,r)+d^h(r,v)<0$ always holds. For P3, since $|\Tin_r|=|\Bin_r(\Din_L)|=\Theta(2^L)=\Theta(n)$, and $|\Tout_r|\geq |\Bout_r(\Dout_L)|=\Omega(2^L)=\Omega(n)$, we see that $|\Tin_r|=\Theta(|\Tout_r|)$.

    \item $\ell^\star=1$: Let $\Delta_r=d(\Din_1)$, $\Tin_r=\Vin_r,\Tout_r=\Vout_r$. 
    Properties P1 and P2 hold by definition of $\Vin_r$ and $\Vout_r$.
    For P3, by \cref{eq:size} we have $|\Bin_r(\Din_1)|=O(1)$ and $|\Bout_r(\Dout_1)|=O(1)$. By \cref{clm:contain}, $\Vin_r\subseteq \Bin_r(\Din_1)$ and $\Vout_r\subseteq \Bout_r(\Dout_1)$ since $-\Delta_r=-d(\Din_1)\leq d(\Dout_1)$. Hence $|\Tin_r|=\Theta(|\Tout_r|)=O(1)$.

    \item   $\ell^\star>1$ and $d(\Din_{\ell^\star-1})<d(\Din_{\ell^\star})$:
    \begin{itemize}
        \item If $-d(\Dout_{\ell^\star-1})\leq d(\Din_{\ell^\star})$: Let $\Delta_r=-d(\Dout_{\ell^\star-1})$, $\Tin_r=\Vin_r$, and $\Tout_r=\Bout_r(\Dout_{\ell^\star-1})$. Then by \cref{clm:contain}, we have $\Vout_r\subseteq \Bout_r(\Dout_{\ell^\star-1})=\Tout_r$. On the other hand, by the minimality of $\ell^\star$ we have $d(\Din_{\ell^\star-1})+d(\Dout_{\ell^\star-1})<0$ and hence $d(\Din_{\ell^\star-1})< -d(\Dout_{\ell^\star-1})=\Delta_r\leq d(\Din_{\ell^\star})$, so \cref{clm:contain} implies that $\Bin_r(\Din_{\ell^\star-1})\subseteq \Vin_r\subseteq \Bin_r(\Din_{\ell^\star})$.
        Property P1 is immediate.
        For P2, we have $d^h(u,r)<\Delta_r=-d(\Dout_{\ell^\star-1})$ for every $u\in\Tin_r$ and $d^h(r,v)\le d(\Dout_{\ell^\star-1})$ for every $v\in\Tout_r$, hence $d^h(u,r)+d^h(r,v)<0$.
        For P3, $|\Tout_r|=|\Bout_r(\Dout_{\ell^\star-1})|=\Theta(2^{\ell^\star})$ by \cref{eq:size}, and the sandwich above gives $|\Tin_r|=\Theta(2^{\ell^\star})$ as well, so $|\Tin_r|=\Theta(|\Tout_r|)$.
        \item If $-d(\Dout_{\ell^\star-1})> d(\Din_{\ell^\star})$: Let $\Delta_r=d(\Din_{\ell^\star})$, $\Tin_r=\Vin_r$, and $\Tout_r=\Vout_r$. Since $d(\Din_{\ell^\star-1})<\Delta_r=d(\Din_{\ell^\star})$, \cref{clm:contain} implies that $\Bin_r(\Din_{\ell^\star-1})\subseteq \Vin_r\subseteq \Bin_r(\Din_{\ell^\star})$. On the other hand, since $d(\Dout_{\ell^\star-1})<-d(\Din_{\ell^\star})=-\Delta_r$ and $d(\Dout_{\ell^\star})\geq -d(\Din_{\ell^\star})=-\Delta_r$, \cref{clm:contain} implies that $\Bout_r(\Dout_{\ell^\star-1})\subseteq \Vout_r\subseteq \Bout_r(\Dout_{\ell^\star})$. Properties P1 and P2 are immediate. For P3, we have both $|\Tin_r|=\Theta(2^{\ell^\star})$ and $|\Tout_r|=\Theta(2^{\ell^\star})$ by the sandwiches above.
    \end{itemize}

    \item $\ell^\star>1$ and $d(\Dout_{\ell^\star-1})<d(\Dout_{\ell^\star})$:
    \begin{itemize}
        \item If $d(\Din_{\ell^\star-1})\geq -d(\Dout_{\ell^\star})$: Let $\Delta_r=d(\Din_{\ell^\star-1}),\Tin_r=\Bin_r(\Din_{\ell^\star-1}), \Tout_r=\Vout_r$. Then by \cref{clm:contain}, we have $\Vin_r\subseteq \Bin_r(\Din_{\ell^\star-1})=\Tin_r$. On the other hand, by the minimality of $\ell^\star$, we have $d(\Din_{\ell^\star-1})+d(\Dout_{\ell^\star-1})<0$, which implies
        $-d(\Dout_{\ell^\star-1})>d(\Din_{\ell^\star-1})=\Delta_r\geq -d(\Dout_{\ell^\star})$, so \cref{clm:contain} implies that $\Bout_r(\Dout_{\ell^\star-1})\subseteq \Vout_r\subseteq \Bout_r(\Dout_{\ell^\star})$. Property P1 is immediate.
        For P2, $d^h(u,r)\le d(\Din_{\ell^\star-1})=\Delta_r$ for all $u\in\Tin_r$ and $d^h(r,v)<-\Delta_r$ for all $v\in\Tout_r$, hence $d^h(u,r)+d^h(r,v)<0$.
        For P3, \cref{eq:size} gives $|\Tin_r|=|\Bin_r(\Din_{\ell^\star-1})|=\Theta(2^{\ell^\star})$ and the sandwich gives $|\Tout_r|=\Theta(2^{\ell^\star})$, so $|\Tin_r|=\Theta(|\Tout_r|)$.
        \item If $d(\Din_{\ell^\star-1})< -d(\Dout_{\ell^\star})$: Let $\Delta_r = -d(\Dout_{\ell^\star}),\Tin_r=\Vin_r,\Tout_r=\Vout_r$. Since $d(\Dout_{\ell^\star-1})< -\Delta_r=d(\Dout_{\ell^\star})$, \cref{clm:contain} implies that $\Bout_r(\Dout_{\ell^\star-1})\subseteq \Vout_r\subseteq \Bout_r(\Dout_{\ell^\star})$. On the other hand, since $d(\Din_{\ell^\star-1})<-d(\Dout_{\ell^\star})=\Delta_r$ and $d(\Din_{\ell^\star})\geq -d(\Dout_{\ell^\star})=\Delta_r$, \cref{clm:contain} implies that $\Bin_r(\Din_{\ell^\star-1})\subseteq \Vin_r\subseteq \Bin_r(\Din_{\ell^\star})$. Properties P1 and P2 are immediate. For P3, we have both $|\Tin_r|=\Theta(2^{\ell^\star})$ and $|\Tout_r|=\Theta(2^{\ell^\star})$ by the sandwiches above.
    \end{itemize}
    
\end{itemize}

We first note that the four cases above exhaustively cover all possible scenarios. The first case covers when $\ell^\star$ does not exist.
Otherwise, either $\ell^\star=1$, yielding the second case, or $\ell^\star>1$.
When $\ell^\star>1$, by minimality we have $d(\Din_{\ell^\star-1})+d(\Dout_{\ell^\star-1})<0$.
Since the sequences $\{\Din_\ell\}_{\ell\in[L]}$ and
$\{\Dout_\ell\}_{\ell\in[L]}$ are non-decreasing, at least one of the inequalities
$d(\Din_{\ell^\star-1})<d(\Din_{\ell^\star})$ or $d(\Dout_{\ell^\star-1})<d(\Dout_{\ell^\star})$
must hold; otherwise the sum would not increase at index $\ell^\star$.
If the former holds, we are in the third case; if the latter holds, we are in the fourth case.
Thus, the four cases above are exhaustive.

Hence we can always compute $\Delta_r$ together with $\Tin_r$ and $\Tout_r$ satisfying Properties P1-P3. Note that the sets $\Tin_r,\Tout_r$ are not explicitly computed; they are used solely for analysis.

\paragraph{Bounding the size.} Now we bound size of $\Tin_r$ and $\Tout_r$. By Property P2, for every $u\in \Tin_r$ and $v\in \Tout_r$, we have $d^h(u,r)+d^h(r,v)<0$. This implies that the negative vertex $r$ is in between $(u,v)$ for all $u\in \Tin_r,v\in \Tout_r$. We charge vertex $r$ to these $|\Tin_r|\cdot |\Tout_r|$ pairs.

By the betweenness reduction guarantee of \cref{lem:betweenness-reduction}, each pair $(u,v)\in V\times V$ is charged at most $\eta/b$ times. Therefore, 
\[
\sum_{r\in N} |\Tin_r|\cdot |\Tout_r|\leq \frac{\eta n^2}{b}
\]
Since $|\Tin_r|=\Theta(|\Tout_r|)$ holds for all $r\in N$ (Property P3), we have $|\Tin_r|+|\Tout_r|=O\left(\sqrt{|\Tin_r|\cdot |\Tout_r|+1}\right)$ and
\[
\sum_{r\in N} \left(|\Tin_r|+|\Tout_r|\right)^2=O\left(\sum_{r\in N} \left(\sqrt{|\Tin_r|\cdot |\Tout_r|+1}\right)^2\right) = O\left(\frac{\eta n^2}{b}\right).
\]
Since $|N|=\eta$, applying Cauchy--Schwarz yields
\[
\sum_{r\in N} \left(|\Tin_r|+|\Tout_r|\right) = O\left(\frac{\eta n}{\sqrt{b}}\right).
\]
Since $|\Vin_r|\leq |\Tin_r|$ and $|\Vout_r|\leq |\Tout_r|$ by Property P1, we obtain
\[
\sum_{r\in N} \left(|\Vin_r|+|\Vout_r|\right)^2\leq \sum_{r\in N} \left(|\Tin_r|+|\Tout_r|\right)^2= O\left(\frac{\eta n^2}{b}\right),
\]
and
\[
\sum_{r\in N} \left(|\Vin_r|+|\Vout_r|\right)\leq \sum_{r\in N} \left(|\Tin_r|+|\Tout_r|\right) = O\left(\frac{\eta n}{\sqrt{b}}\right).
\]

\paragraph{Runtime.}
Finally, we analyze the running time of the algorithm. For the algorithm $\estimate$, a single random experiment can be implemented in $O(h(m+n\log n))=\tilde O(hn^2)$ time using a hybrid of Dijkstra and Bellman-Ford. Since we perform $O(\log n)$ repetitions, the total time for one call to $\estimate$ is $\tilde O(hn^2)$. As $\estimate$ is invoked $O(\log n)$ times, the total time spent on these calls is bounded by $\tilde O(hn^2)$. After that, for each negative vertex $r\in N$, it requires only $\tilde O(1)$ time to obtain $\Delta_r$. Therefore, the overall running time of the algorithm is $\tilde O(hn^2)$, which completes the proof.

\section{The Main Algorithm} \label{sec:alg}

\begin{algorithm}
\caption{Single-Source Shortest Path}
\label{alg:sssp}
\KwIn{A directed graph $H$ and a source vertex $s\in V(H)$.}
\KwOut{$d(s,v)$ for every $v\in V(H)$.}
Convert $H$ into an equivalent well-behaved graph $G = (V,E)$.\\
Let $G_0=(V_0,E_0)$ denote the initial state of the graph $G$, and let $n_0=|V_0|$.\\
Let $N$ be the set of negative vertices in $G_0$, define $\bar N=\{\bar r:r\in N\}$, and let $\eta=|N|$. \\
Note that by the construction described in \cref{sec:prelim}, we have $N=V(H)$.\\
For every $v\in V_0$, let $v$ be a heavy vertex.\\
\If{$\eta\leq 100$} {
\Return $d^{100}(s,v)$ for every $v\in V(H)$.
}
\For{$t=\{1,2,\dots,L\}$ where $L=O(\log n_0)$} {
    $\shortcut(G,t)$.
}
Compute $2$-hop shortest paths from $s$.\\
Let $\varphi$ denote the cumulative vertex potential.\\
\Return $d^2(s,v)-\varphi(s)+\varphi(v)$ for every $v\in V(H)$.
\end{algorithm}

We start by providing an overview to our single-source shortest path algorithm, described in \cref{alg:sssp}. We first convert the input graph into an equivalent well-behaved graph.
We then run $\shortcut(G,t)$ for $L = O(\log n_0)$ iterations.
The purpose of a single iteration is to reduce, by a constant factor, the
number of negative edges (the \emph{hop count}) on shortest paths, while preserving all
shortest-path distances. In particular, any $h$-hop shortest path can be transformed into a path with at most $h-\lfloor h/3\rfloor$ hops while preserving total weight. Consequently, after $L = O(\log n_0)$ iterations, it suffices to compute $2$-hop shortest paths from the source.

As in previous works, we \emph{freeze} the set $N$ of negative vertices at the start of the
algorithm; that is, the set $N$ remains fixed throughout.
Recall that under our definition of a well-behaved graph, for every negative vertex $r$
the graph contains both a negative edge $(r,\bar r)$ and a non-negative reverse edge
$(\bar r,r)$ of weight $-w(r,\bar r)$; this symmetric representation is used throughout
the shortcutting algorithm. 

\paragraph{Recap of \cite{LLRZ25}.} We now consider a single iteration of the $\shortcut$ procedure.
Our approach follows the shortcutting framework of~\cite{LLRZ25}.
Suppose that for every negative vertex $r\in N$, we have computed a value $\Delta_r$ and the sets
\[
\Uin_r=\{v\in V:d^0(v,r)<\Delta_r\}, \quad \Uout_r=\{v\in V:d^1(r,v)<-\Delta_r\},
\]
along with $d^0(v,r)$ for every $v\in \Uin_r$ and $d^1(r,v)$ for every $v\in \Uout_r$.
Using betweenness reduction and recursion, \cite{LLRZ25} shows it is possible to obtain the following bounds:
\[
\sum_{r\in N}(|\Uin_r|+|\Uout_r|)=\tilde O(n^{1.5}), \quad \sum_{r\in N}(|\Uin_r|+|\Uout_r|)^2=\tilde O(n^2).
\]

To explain the local shortcutting goal, consider a shortest path $P$ and three successive
negative edges on $P$, say $(u,\bar u)$, $(r,\bar r)$, and $(v,\bar v)$.
The algorithm’s objective around $r$ is to ensure that, after adding shortcut edges,
the middle negative edge $(r,\bar r)$ can be eliminated from the representation of $P$, either by \emph{merging} it with
a predecessor or successor negative edge (the easy scenarios shown as Case 1 and 2 in \Cref{fig:simple}), or by
\emph{bypassing} it as in the hard scenario shown as Case 3 in \Cref{fig:hard}.

\begin{figure}
    \centering
    \begin{subfigure}[b]{0.45\textwidth}
        \centering
        \begin{tikzpicture}
\coordinate (u)  at (-3.2, 1.9);
  \coordinate (up) at (-2.6, 0.1);
  \coordinate (r)  at (-0.3, 1.2);
  \coordinate (rp) at ( 0.3,-1.0);
  \coordinate (v)  at ( 2.6, 0.5);
  \coordinate (vp) at ( 3.2,-1.8);
  \coordinate (rb)  at (-0.3, -0.2);
  \coordinate (rc)  at (-0.3, 0.4);

  \draw[blue,thin] (r) -- (rb);
  \node[blue,left] at (rc) {$\Delta_r$};

  % Dashed levels + bracket
%  \draw[dashedlvl] (-3.6,1.2) -- (3.2,1.2);
  \draw[dashedlvl] (-3.6,-0.2) -- (3.6,-0.2);
%  \sqbracket{0.85}{0.2}{1.0}{0.28}{\Delta_r}

  % Black sampled polylines (up -> r) and (rp -> v)
  \draw[bpath] (up) -- (r);
  \segmentdots{up}{r}{4}{dot}

  \draw[bpath] (rp) -- (v);
  \segmentdots{rp}{v}{4}{dot}

  % Red segments and red curve arrow
  \draw[rthin] (u) -- (up);
  \draw[rthin] (v) -- (vp);
  \draw[rthin] (r) -- (rp);

  \draw[rcurve] (u) .. controls (-2.6,0.2) and (-1.2,-0.7) .. (rp);

  % Labels
  \node[above left]  at (u)  {$u$};
  \node[left]  at (up) {$\bar u$};
  \node[above]       at (r)  {$r$};
  \node[below]       at (rp) {$\bar r$};
  %     \node[below]       at (rb)  {$\tilde r$};
  \node[above right] at (v)  {$v$};
  \node[below right] at (vp) {$\bar v$};

  % Endpoint dots
  \node[dot] at (u)  {};
  \node[dot] at (up) {};
  \node[dot] at (r)  {};
  \node[dot] at (rb)  {};
  \node[dot] at (rp) {};
  \node[dot] at (v)  {};
  \node[dot] at (vp) {};
\end{tikzpicture}
        \caption{Case 1: $\bar u\in \Uin_r$.}
        \label{fig:simple-1}
    \end{subfigure}
    ~
    \begin{subfigure}[b]{0.45\textwidth}
        \centering
        \begin{tikzpicture}
\coordinate (u)  at (-3.2, 1.9);
  \coordinate (up) at (-2.6, -0.6);
  \coordinate (r)  at (-0.3, 1.2);
  \coordinate (rp) at ( 0.3,-1.0);
  \coordinate (v)  at ( 2.9, -0.4);
  \coordinate (vp) at ( 3.3,-1.8);
  \coordinate (rb)  at (-0.3, -0.2);
  \coordinate (rc)  at (-0.3, 0.35);

  \draw[blue,thin] (r) -- (rb);
  \node[blue,left] at (rc) {$\Delta_r$};

  % Dashed levels + bracket
%  \draw[dashedlvl] (-3.6,1.2) -- (3.2,1.2);
  \draw[dashedlvl] (-3.6,-0.2) -- (3.6,-0.2);
%  \sqbracket{0.85}{0.2}{1.0}{0.28}{\Delta_r}

  % Black sampled polylines (up -> r) and (rp -> v)
  \draw[bpath] (up) -- (r);
  \segmentdots{up}{r}{4}{dot}

  \draw[bpath] (rp) -- (v);
  \segmentdots{rp}{v}{4}{dot}

  % Red segments and red curve arrow
  \draw[rthin] (u) -- (up);
  \draw[rthin] (v) -- (vp);
  \draw[rthin] (r) -- (rp);

  \draw[rcurve] (r) .. controls (1.3,0.8) and (2.3,-0.1) .. (vp);

  % Labels
  \node[above left]  at (u)  {$u$};
  \node[below left]  at (up) {$\bar u$};
  \node[above]       at (r)  {$r$};
  \node[below]       at (rp) {$\bar r$};
  %     \node[below]       at (rb)  {$\tilde r$};
  \node[right] at (v)  {$v$};
  \node[below right] at (vp) {$\bar v$};

  % Endpoint dots
  \node[dot] at (u)  {};
  \node[dot] at (up) {};
  \node[dot] at (r)  {};
  \node[dot] at (rb)  {};
  \node[dot] at (rp) {};
  \node[dot] at (v)  {};
  \node[dot] at (vp) {};
\end{tikzpicture}
        \caption{Case 2: $v\in \Uout_r$}
        \label{fig:simple-2}
    \end{subfigure}
	\caption{Illustration of two simple cases. The height of each vertex represents its cumulative distance along the path. Black edges denote non-negative edges in the graph, red edges denote negative edges, and purple edges are the shortcut edges added by the shortcut algorithm.}
    \label{fig:simple}
\end{figure}
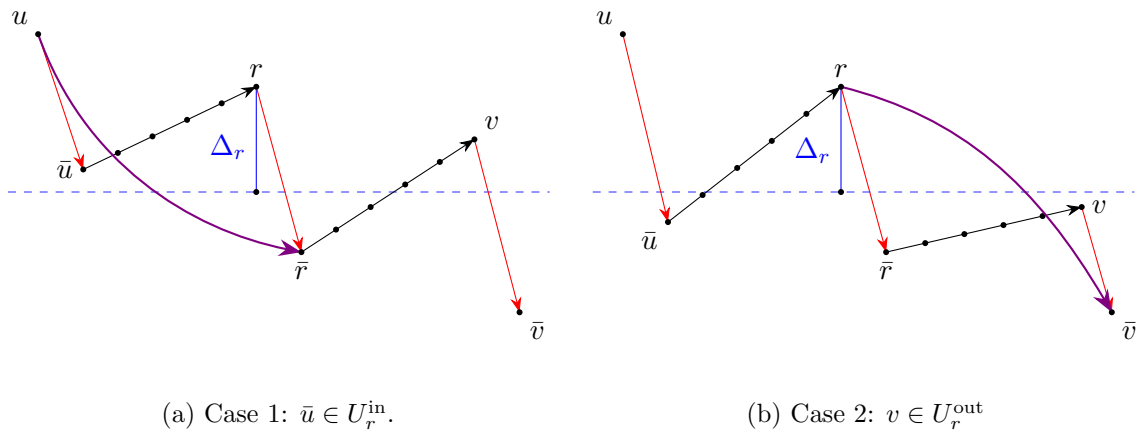

\paragraph{Merging $(r,\bar r)$ with adjacent negative edge (Case 1 and 2).} We first consider two simple cases shown in \cref{fig:simple}. When $\bar u\in \Uin_r$ (i.e. $\bar u$ is above the $\Delta_r$ line), the algorithm adds an edge $(u,\bar r)$ of weight $w(u,\bar u)+d^0(\bar u,r)+w(r,\bar r)$, reducing  the 2-hop path $u\rightsquigarrow \bar r$ into a single negative edge. Likewise, when $v\in \Uout_r$ (i.e. $v$ is below the $\Delta_r$ line), the algorithm adds an edge $(r,\bar v)$ of weight $d^1(r,v)+w(v,\bar v)$, reducing the 2-hop $r\rightsquigarrow \bar v$ path into a single negative edge. Since we assume that the graph $G$ is well-behaved, each vertex is incident to at most one negative edge. Consequently, adding these shortcut edges requires
\[
\sum_{r\in N} O\left(|\Uin_r|+|\Uout_r|\right) = \tilde O(n^{1.5})
\]
time. One issue that arises after adding shortcut edges is that the graph may no longer remain well-behaved. However, it is straightforward to restore this property.

\begin{figure}
    \centering
    \begin{subfigure}[b]{0.45\textwidth}
        \centering
        \begin{tikzpicture}
\coordinate (u)  at (-3.2, 1.9);
  \coordinate (up) at (-2.6, -0.8);
  \coordinate (r)  at (-0.3, 1.2);
  \coordinate (rp) at ( 0.3,-1.0);
  \coordinate (v)  at ( 2.6, 0.7);
  \coordinate (vp) at ( 3.2,-1.8);
  \coordinate (rb)  at (-0.3, -0.2);
  \coordinate (rc)  at (-0.3, 0.33);
  \coordinate (y)  at ( -2.14, -0.4);
  \coordinate (x)  at ( -1.68, 0);
  \coordinate (z)  at ( 1.68, 0.02);

  \draw[blue,thin] (r) -- (rb);
  \node[blue,left] at (rc) {$\Delta_r$};

  % Dashed levels + bracket
%  \draw[dashedlvl] (-3.6,1.2) -- (3.2,1.2);
  \draw[dashedlvl] (-3.6,-0.2) -- (3.6,-0.2);
%  \sqbracket{0.85}{0.2}{1.0}{0.28}{\Delta_r}

  % Black sampled polylines (up -> r) and (rp -> v)
  \draw[bpath] (up) -- (r);
  \segmentdots{up}{r}{4}{dot}

  \draw[bpath] (rp) -- (v);
  \segmentdots{rp}{v}{4}{dot}

  % Red segments and red curve arrow
  \draw[rthin] (u) -- (up);
  \draw[rthin] (v) -- (vp);
  \draw[rthin] (r) -- (rp);

%  \draw[rcurve] (r) .. controls (1.3,0.8) and (2.3,-0.1) .. (vp);

  \draw[violet,thick,-{Stealth[length=2.5mm]}] (y) -- (rb);
  \draw[violet,thick,-{Stealth[length=2.5mm]}] (rb) -- (z);

  % Labels
  \node[above left]  at (u)  {$u$};
  \node[below left]  at (up) {$\bar u$};
  \node[above]       at (r)  {$r$};
  \node[below]       at (rp) {$\bar r$};
  \node[violet,below]       at (rb)  {$\tilde r_t$};
  \node[right] at (v)  {$v$};
  \node[below right] at (vp) {$\bar v$};
  \node at (-1.4, -0.05)  {$x$};
  \node[below right] at (y)  {$y$};
  \node[right] at (z)  {$z$};

  % Endpoint dots
  \node[dot] at (u)  {};
  \node[dot] at (up) {};
  \node[dot] at (r)  {};
  \node[dot] at (rb)  {};
  \node[dot] at (rp) {};
  \node[dot] at (v)  {};
  \node[dot] at (vp) {};
\end{tikzpicture}
        \caption{Case 3: $\bar u\not\in \Uin_r$ and $v\notin \Uout_r$. \textcolor{violet}{$\tilde r_t$} is an N-Steiner vertex that we introduce.}
        \label{fig:hard}
    \end{subfigure}
    ~
    \begin{subfigure}[b]{0.45\textwidth}
        \centering
        \begin{tikzpicture}
  \coordinate (r)  at (2.5, 3);
  \coordinate (rb)  at (2.5, 0);
  \coordinate (rc)  at (2.5, 1.5);
  \coordinate (x)  at (0, 1.7);
  \coordinate (u1)  at (-3.4, -1.8);
  \coordinate (u2)  at (-2.55, -1.5);
  \coordinate (u3)  at (-1.75, -1);
  \coordinate (u4)  at (-0.8, -0.6);
  \coordinate (u5)  at (-0.2, -0.2);
  \coordinate (u6)  at (0.4, 0.3);
  \coordinate (ud)  at (1.12, 0.62);
  \coordinate (uk)  at (2, 1);
  \coordinate (xl)  at (-1.3, -1);

  \draw[blue,thin] (r) -- (rb);
  \node[blue,left] at (rc) {$\Delta_r$};
  
  \draw[black,thin,-{Stealth[length=2mm]}] (x) -- (r);
  \draw[black,thin,midarrow] (u1) -- (x);
  \draw[black,thin,midarrow] (u2) -- (x);
  \draw[black,thin,midarrow] (u3) -- (x);
  \draw[black,thin,midarrow] (u4) -- (x);
  \draw[black,thin,midarrow] (u5) -- (x);
  \draw[black,thin,midarrow] (u6) -- (x);
  \draw[black,thin,midarrow] (uk) -- (x);

  \draw[violet,thick,midarrow] (xl) -- (x);
  \draw[violet,thin,midarrow] (u3) -- (xl);
  \draw[violet,thin,midarrow] (u2) .. controls (-1.65,-1.25) .. (xl);
  \draw[violet,thin,midarrow] (u1) .. controls (-2.5,-2) and (-1.65,-1.5) .. (xl);
  \draw[violet,thick,midarrow] (xl) .. controls (-0.5,-1.1) and (1.2,-0.8).. (rb);
  \draw[violet,thick,midarrow] (u4) .. controls (0,-0.6) and (1.2,-0.4)..  (rb);
  \draw[violet,thick,midarrow] (u5) -- (rb);

  \draw[dashedlvl] (-3.6,0) -- (3.1,0);
  % Labels
  \node[above]      at (r)  {$r$};
  \node[violet,below]       at (rb)  {$\tilde r_t$};
  \node[left]      at (x)  {$x$};
  \node[below]      at (u1)  {$y_1$};
  \node[below]      at (u2)  {$y_2$};
  \node[below]      at (u3)  {$y_3$};
  \node[below]      at (u4)  {$y_4$};
  \node[below]      at (u5)  {$y_5$};
  \node[below]      at (u6)  {$y_6$};
  \node[below]      at (uk)  {$y_k$};
  \node[violet,below right] at (xl)  {$x_\ell$};

  \node[rotate=20] at (ud) {$\dots$};

  % Endpoint dots
  \node[dot] at (x)  {};
  \node[dot] at (r)  {};
  \node[dot] at (rb)  {};
  \node[dot] at (u1)  {};
  \node[dot] at (u2)  {};
  \node[dot] at (u3)  {};
  \node[dot] at (u4)  {};
  \node[dot] at (u5)  {};
  \node[dot] at (u6)  {};
  \node[dot] at (uk)  {};
  \node[dot] at (xl)  {};
\end{tikzpicture}
        \caption{Example of an in-Steiner gadget for $x$. The figure illustrates that case when $k-2^\ell=3$.}
        \label{fig:steiner}
    \end{subfigure}
    \caption{Illustration of the hard case, and an in-Steiner gadget.}
\end{figure}
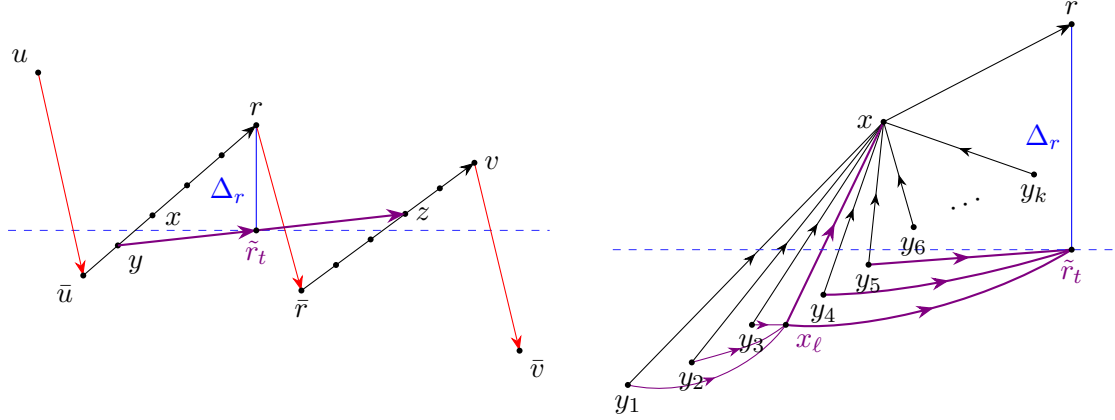

\paragraph{Bypassing $(r,\bar r)$ (Case 3).} Now we turn to the more challenging case as shown in \cref{fig:hard}, where $\bar u\not\in \Uin_r$ and $v\notin \Uout_r$. In this situation, $P$ contains a subpath of the form
\[
\bar u \rightsquigarrow r \to \bar r \rightsquigarrow v,
\]
where the two connecting segments $\bar u \rightsquigarrow r$ and $\bar r \rightsquigarrow v$
contain \emph{no negative edges} (the only negative edge in the displayed subpath is
$(r,\bar r)$).
Bypassing $(r,\bar r)$ means constructing an alternative path from $\bar u$ to $v$
with the same total weight but with one fewer negative edge.

To locate where shortcut edges must attach efficiently, the algorithm measures distances
to and from $r$ along shortest paths.
Along the incoming segment $\bar u \rightsquigarrow r$, the values $d^0(\cdot,r)$ increase
monotonically as we move away from $r$.
Given a threshold value $\Delta_r$, there is therefore a first edge $(y,x)$ on this segment
such that
\[
d^0(x,r) < \Delta_r \le d^0(y,r).
\]
We call such an edge an \emph{incoming threshold-crossing edge}.
A symmetric threshold crossing may occur on the outgoing segment $\bar r \rightsquigarrow v$
when distances are measured from $r$.
Threshold crossings are unavoidable since $\bar u\notin \Uin_r$ and $v\notin\Uout_r$, and they simply identify where
shortcut edges must attach; the choice of $\Delta_r$ is made to control the total work
(and hence running time).

\paragraph{The shortcut construction: $\scin$, $\scout$, and $\tilde r_t$.}
For each negative vertex $r$ and iteration $t$, the algorithm introduces a new auxiliary
vertex $\tilde r_t$, called an \emph{N-Steiner vertex}, which acts as a hub for shortcut
paths associated with $r$.
The construction of shortcut edges is split into two symmetric procedures:
\begin{itemize}
    \item $\scin(r,t)$ identifies vertices $u$ on incoming sides of shortest paths through $r$
    for which adding a non-negative weight edge $(y,\tilde r_t)$ preserves shortest-path distances.
    \item $\scout(r,t)$ identifies vertices $v$ on outgoing sides of shortest paths through $r$
    for which adding a non-negative weight edge $(\tilde r_t,z)$ preserves shortest-path distances.
\end{itemize}
In the hard case, the edges generated by $\scin(r,t)$ and $\scout(r,t)$ yield paths of the form
$\bar u\rightsquigarrow y \to \tilde r_t \to z\rightsquigarrow v$ (as illustrated in \cref{fig:hard}) that bypass the negative edge $(r,\bar r)$ and reduce the hop count
by one.

\paragraph{Adding shortcut edges efficiently via Steiner gadgets.} The main challenge in implementing the procedure $\scin$ efficiently is the time required to identify the vertices $y$ to which the edges $(y,\tilde r_t)$ should be added. In~\cite{LLRZ25}, this is accomplished by checking all neighbors of each vertex $x \in \Uin_r$, which may require $\Theta(|\Uin_r| \cdot n)$ time. Summing over all negative vertices, this leads to a total running time of $\Theta(n^{2.5})$

A natural attempt to improve the running time is to introduce Steiner vertices for each vertex $x$ to aggregate its incident edges efficiently, see \cref{fig:steiner}. Let $y_1,\dots,y_k$ be the in-neighbors of $x$. We first order them so that $w(y_1,x)\geq \dots\geq w(y_k,x)$. We introduce $O(\log k)$ Steiner vertices: for each $0\leq \ell \leq \lceil \log k \rceil - 1$, we create an \emph{in-Steiner vertex} $x_\ell$ whose in-parent is $x$, denoted by $\pin(x_\ell) = x$. For each such $x_\ell$, we first add an edge $(x_\ell,x)$ of weight $w(y_{k-2^\ell},x)$. Then, for every $i \in [1,k - 2^\ell]$, we add an edge $(y_i,x_\ell)$ with weight $w(y_i,x) - w(x_\ell,x)$.

With this construction, whenever $x \in \Uin_r$ for some negative vertex $r \in N$, we no longer need to examine all in-neighbors of $x$ explicitly. Instead, we perform a binary search to find the smallest index $j \in [k]$ such that $w(y_j,x)+d^0(x,r)<\Delta_r$. Let $\ell^\star$ be the smallest integer such that $2^{\ell^\star}>k-j$. For vertices $y_i$ with $i \leq k - 2^{\ell^\star}$, we connect them to $\tilde r_t$ via the Steiner vertex $x_{\ell^\star}$, while for the remaining vertices $y_i$ with $k - 2^{\ell^\star}< i < j$, we connect them directly to $\tilde r_t$ using brute force. Observe that for all $y_i$ with $i\in [j,k]$, we have $d^0(y_i,r)\leq w(y_i,x)+d^0(x,r)<\Delta_r$, and hence $y_i\in \Uin_r$. Consequently, the number of vertices handled via brute force is bounded by $2^{\ell^\star}=O(k-j+1)=O(|\Uin_r|)$. Therefore, the total time spent on this step over all negative vertices is bounded by
\[
\sum_{r\in N} O\left((|\Uin_r|+|\Uout_r|)^2\right) = \tilde O(n^2),
\]
as desired.

The straightforward implementation of this natural idea works fine for the first iteration. However, as each successive iteration creates more and more Steiner vertices, the total number of vertices grows by a multiplicative factor in each successive iteration. As a result, the graph becomes too large after a logarithmic number of iterations of this process. In the following, we describe how we structure shortcut edge insertions that allow us to control this growth, and achieve the desired runtime.

\paragraph{Ensuring efficiency.}
The main challenge, as noted above, is to control the growth in the number of vertices. Under the above construction, the number of vertices can increase by a factor of $\Theta(\log n)$ in each shortcutting iteration. Since $\Theta(\log n)$ iterations are required to reduce the hop distance to $O(1)$, this multiplicative growth would cause the total number of vertices to blow up by the end of the algorithm.

To address this issue, the algorithm distinguishes vertices that can be handled by direct adjacency scans
from those that cannot. For auxiliary Steiner vertices, the classification is one-sided, reflecting how they are used.
For example, an \emph{in-Steiner} vertex is classified based on its \emph{out-degree} (it may have
large in-degree, which can be charged to its in-parent). In particular, an in-Steiner vertex is \emph{light} if its out-degree is sufficiently small, and \emph{heavy} otherwise; and symmetrically for out-Steiner vertices. For the original vertices and N-Steiner vertices, they are always designated as heavy since their number remains bounded in course of teh algorithm.
Heavy vertices are equipped with in-Steiner and out-Steiner gadgets that aggregate adjacency at
geometrically spaced scales, as explained above, enabling threshold crossings to be processed efficiently rather
than linear scans.

Since we create auxiliary in-Steiner and out-Steiner vertices only for heavy vertices, bounding the number of heavy vertices suffices to control the overall growth in the number of vertices. However, a vertex that is initially light may become heavy as shortcut edges are added and its degree increases. If this occurs repeatedly, the number of heavy vertices could grow rapidly. To prevent this, we show that we can ``shift" a shortcut edge insertion that would be incident to a light vertex to instead be incident to a heavy vertex while achieving the same shortcutting effect. This allows a light vertex to remain low degree, thereby keeping the number of heavy vertices, and consequently the total number of vertices, under control.

\paragraph{Shortcut edge insertions: $\ein$ and $\eout$.}
To ensure necessary structural invariants, 
\emph{candidate} shortcut edges generated by procedures $\scin$ and $\scout$  (to
help eliminate a negative hop across $(r,\bar r)$),
are not inserted directly into $G$ but instead, every insertion is routed through $\ein(u,v,W)$ (for edges \emph{into} a hub such as
$\tilde r_t$) and $\eout(u,v,W)$ (for edges \emph{out of} such a hub).

The reason is that shortcutting is implemented using auxiliary Steiner gadgets for high-degree
vertices. As a result, a shortcut edge produced by $\scin/\scout$ may naturally be incident to a Steiner
vertex rather than to an original vertex.
The procedures $\ein$ and $\eout$ are responsible for making such insertions \emph{compatible}
with the Steiner gadget structure and with the invariants used in the analysis.
Concretely, Steiner vertices are not independent vertices in the shortest-path
instance: they are auxiliary constructs that summarize groups of edges incident to
an underlying \emph{parent} vertex. When a new edge is inserted, its effect may need to be propagated to the parent vertex for which the corresponding Steiner gadget was created.
This propagation can require inserting additional edges, and those edges may in turn
be incident to other auxiliary Steiner vertices.
For this reason, the procedures $\ein(\cdot)$ and $\eout(\cdot)$ perform insertions
recursively.

Finally, $\ein$ and $\eout$ enforce the key restriction on negative edges:
although the algorithm may introduce negative-weight edges, any such edge must be incident to a
negative vertex in $N$.
Any negative edges that cannot be safely inserted at the current moment (because they would violate the maintained
invariants under the current reweighting) are stored in auxiliary deferred sets and reconsidered
after the next potential update, at which point they may become safe to add.

\paragraph{Key invariants (informal).}
Throughout the execution of $\shortcut(G,t)$, the algorithm maintains the following properties,
which are formalized later:
\begin{itemize}
    \item Any negative edge inserted during shortcutting is incident to a vertex in $N$.
    \item After each iteration, the graph can be restored to be a well-behaved graph.
    \item For each negative vertex $r$, the work performed in one iteration is bounded in terms of
    $|\Vin_r|$ and $|\Vout_r|$ (and their aggregate bounds from \cref{lem:compute-delta}).
    \item Degree growth is concentrated on vertices designated heavy.
\end{itemize}

We next describe the procedure $\shortcut(G,t)$ in detail.

\subsection{The Procedures $\shortcut(G,t)$} 

During the execution of $\shortcut(G,t)$, the algorithm introduces several types of
auxiliary vertices in addition to the original vertices.

\begin{itemize}
    \item \emph{In-Steiner vertices}: auxiliary vertices created to compactly represent
    groups of incoming edges to a heavy vertex (created in Step~\ref{step:steiner-vertex-in}).
    \item \emph{Out-Steiner vertices}: auxiliary vertices created to compactly represent
    groups of outgoing edges from a heavy vertex (created in Step~\ref{step:steiner-vertex-out}).
    \item \emph{N-Steiner vertices}: vertices $\tilde r_t$ created for each negative
    vertex $r\in N$ in iteration $t$, serving as hubs for shortcut paths that bypass
    $(r,\bar r)$ (created in Step~\ref{step:tilde-r}).
\end{itemize}

Each auxiliary vertex has a designated parent (in-parent, out-parent, or N-parent),
which records the vertex for which it was created. Vertices are also assigned a \emph{level}, equal to the iteration in which they are created.
We say a vertex is a \emph{level-$t$ vertex} if it is created during the execution of $\shortcut(G,t)$ (i.e. the $t$-th iterations). In particular, every vertex in the original graph is a level-$0$ vertex.
All vertices inherited from the original graph as well as N-Steiner vertices are called \emph{regular vertices}.

To ensure both correctness and running time, the algorithm maintains two distinct
systems of adjacency lists. The lists $\{\Nin(v),\Nout(v)\}$ represent the actual adjacency structure of the graph
$G$. In contrast, the lists $\{\Ain_v,\Aout_v\}$ are auxiliary views used during shortcut
generation. They are constructed so as to (i) freeze the relevant neighborhoods while edges are
being added, and (ii) exclude edges whose presence would cause redundant or excessive
work.
The sets $\Fin_t$ and $\Fout_t$ store edges whose insertion is deferred; such edges are
reconsidered in later iteration after reweighting by vertex potentials. We also initialize the sets of deferred edges as $\Fin_0=\Fout_0=\emptyset$.

We now describe the steps of $\shortcut(G,t)$ in detail. We use $G=(V,E)$ to denote the current state of the graph $G$, and let $n=|V|$. Recall that $n_0$ denotes the number of vertices in the original graph $G_0$.

\begin{ourbox}
\begin{enumerate}
\item[] \hspace{-\leftmargin} {\bf Phase 1: Reweighting, threshold computation, and preprocessing.}

\item \label{step:betweenness-reduction}
We set the parameters $h=L+1=O(\log n_0)$, $b=\eta/\gamma$, and $\lambda=n/\sqrt{b}$, where
$\gamma=\Theta(\log n_0\cdot 2^{\sqrt{\log n_0}})$.
We then apply the algorithm of \cref{lem:betweenness-reduction} to the graph $G$. Recall that the algorithm recursively
calls negative-weight single-source shortest path on a graph $G'$ with
$O(hn)$ vertices and $O(b\log n)$ negative vertices.
Let $\varphi$ be the resulting vertex potential, and set $G\gets G_\varphi$.

\smallskip
\noindent
\emph{Explanation.}
This step reweights the graph so that distances are better structured for the
current iteration.
The betweenness reduction guarantees that, for each negative vertex $r$, the
sets of vertices close to $r$ (which will participate in shortcutting) are small
enough to handle efficiently.

\item \label{step:f-in-f-out}
Let $\Fin_t=\Fin_{t-1}$ and $\Fout_t=\Fout_{t-1}$.
\begin{enumerate}
    \item \label{step:f-in}
    For every edge $e=(u,v)\in \Fin_{t-1}$, let $w(e)=w_\varphi(e)$.
    \\If $w(e)\ge 0$, call $\ein(u,v,w(e))$.
    \item \label{step:f-out}
    For every edge $e=(u,v)\in \Fout_{t-1}$, let $w(e)=w_\varphi(e)$.
    \\If $w(e)\ge 0$, call $\eout(u,v,w(e))$.
\end{enumerate}

\smallskip
\noindent
\emph{Explanation.}
We carry forward the deferred edge sets by setting $\Fin_t=\Fin_{t-1}$ and $\Fout_t=\Fout_{t-1}$.
After applying new vertex potentials, some previously deferred edges may become non-negative and hence safe to insert. This step revisits all deferred edges and inserts exactly those that have
become admissible.

\item \label{step:pre}
We compute a threshold value $\Delta_r$ for every negative vertex $r\in N$ using
\cref{lem:compute-delta} with $h=O(\log n_0)$ as defined in Step~\ref{step:betweenness-reduction} above.
We then compute the sets $\Uin_r$ and $\Uout_r$ using \cref{lem:compute-U},
together with distances $d^0(v,r)$ for all $v\in \Uin_r$ and
$d^1(r,v)$ for all $v\in \Uout_r$.

\smallskip
\noindent
\emph{Explanation.}
The threshold $\Delta_r$ determines how far from $r$ we are willing to search
when constructing shortcuts.
The sets $\Uin_r$ and $\Uout_r$ capture the vertices that are “close enough” to
$r$ to be relevant for shortcutting, and their sizes are carefully controlled
by the choice of $\Delta_r$.

\item \label{step:delete} For every vertex $v\in V$:
\begin{itemize}
    \item Let $\Ain_v=\{u\in\Nin(v):w(u,v)\geq 0\}$ and
    $\Aout_v=\{u\in\Nout(v):w(v,u)\geq 0\}$.
    \item If $v$ is heavy, remove from $\Ain_v$ (resp.\ $\Aout_v$) any in-Steiner
    (resp.\ out-Steiner) vertex whose parent is $v$.
    \item If $v$ is an in-Steiner vertex, set $\Ain_v=\emptyset$ and remove its in-parent $\pin(v)$ from $\Aout_v$.
    \item If $v$ is an out-Steiner vertex, set $\Aout_v=\emptyset$ and remove its out-parent $\pout(v)$ from $\Ain_v$.
\end{itemize}

\smallskip
\noindent
\emph{Explanation.}
This step constructs the auxiliary adjacency lists $\Ain_v$ and $\Aout_v$ that
will be scanned by $\scin$ and $\scout$. These lists serve two purposes.

First, they \emph{freeze} the neighborhoods visible during the shortcutting phase:
only edges that were present before shortcut edges are added are included.
Edges inserted later in the same iteration are therefore not reconsidered,
which prevents repeated or cascading work.

Second, we remove edges that connect a Steiner vertex directly to its parent. This is done for technical reasons, and simplifies the analysis. We also remove all incoming edges to an in-Steiner vertex, since such vertex are intended to aggregate adjacency information of its parent. Any shortcut involving its in-neighbors can instead be handled at its parent. Keeping its incoming edges would cause the same neighborhood to be processed twice and would unnecessarily increase degrees. Analogously, we remove all outgoing edges from out-Steiner vertices.

\item[] \hspace{-\leftmargin} {\bf Phase 2: Simple shortcutting.}

\item \label{step:negative-left}
For each negative edge $(r,\bar r)$ and each $\bar u\in \Uin_r$, if $\bar u$ has an
incoming negative edge $(u,\bar u)$, add an edge $(u,\bar r)$ of weight
$w(u,\bar u)+d^0(\bar u,r)+w(r,\bar r)$.

\smallskip
\noindent
\emph{Explanation.}
This handles the first simple case (as illustrated in \cref{fig:simple-1}): two consecutive negative edges can be merged
into a single negative edge. Doing so immediately reduces the hop count without requiring Steiner vertices
or threshold crossings, as analyzed in Case 1 of \cref{lem:shortcut-main}.

\item \label{step:negative-right}
For every negative vertex $r\in N$ and every $v\in \Uout_r$, if $v$ has an outgoing
negative edge $(v,\bar v)$, add an edge $(r,\bar v)$ of weight
$d^1(r,v)+w(v,\bar v)$.

\smallskip
\noindent
\emph{Explanation.}
This is the symmetric simple case on the outgoing side (illustrated in \cref{fig:simple-2}) and is analyzed as Case 2 in \cref{lem:shortcut-main}. After this step, only the hard case remains (illustrated in \cref{fig:hard}), and all steps below are devoted to handling this case. We analyze the correctness of the following steps in Case 3 of \cref{lem:shortcut-main}.\\

\item[] \hspace{-\leftmargin} {\bf Phase 3: Shortcutting via Steiner gadgets.}

\item \label{step:steiner-vertex-in} For every heavy vertex $v\in V$, if $\Ain_v\neq \emptyset$:
        \begin{itemize}
            \item Sort $\Ain_v=\{u_1,u_2,\dots,u_k\}$ so that $w(u_1,v)\geq w(u_2,v)\geq \dots\geq w(u_k,v)$; break ties arbitrarily.
            \item Let $\tau = \lceil \log k\rceil-1$. Create $\tau+1$ Steiner vertices $v_0,v_1,\dots,v_\tau$. We refer to these vertices as \emph{in-Steiner vertices}.
            \item For every $0\leq \ell\leq \tau$:
                \begin{itemize}
                    \item Set $\pin(v_\ell)=v$, we say $\pin(v_\ell)$ is the \emph{in-parent} of $v_\ell$.
                    \item We will later decide whether $v_\ell$ is light or heavy.
                \end{itemize}
            
            \item For every $0\leq \ell\leq \tau$:
                \begin{enumerate}
                    \item \label{step:in-up} Add an edge $(v_\ell,v)$ of weight $w(u_{k-2^{\ell}},v)$.
                    \item \label{step:in-down} For every $1\leq i \leq k-2^\ell$, add an edge $(u_i,v_\ell)$ of weight $w(u_i,v)-w(v_\ell,v)$. If $u_i$ is an out-Steiner vertex, add an edge $(\pout(u_i),v_\ell)$ of weight $w(\pout(u_i),u_i)+w(u_i,v)-w(v_\ell,v)$.
                    \item \label{step:in-r} If $v\in N$, add an edge $(v,v_\ell)$ of weight $-w(v_\ell,v)$.
                    \item \label{step:in-tilde-r} If $v$ is an N-Steiner vertex, add an edge $(\pn(v),v_\ell)$ of weight $w(\pn(v),v)-w(v_\ell,v)$.
                \end{enumerate}
        \end{itemize}

\smallskip
\noindent
\emph{Explanation.}
This step builds an \emph{in-Steiner gadget} for every heavy vertex $v$ in order to
avoid scanning all of its incoming neighbors during shortcut generation.
Specifically, if $v$ has many incoming edges, then when $\scin$ later searches for an
incoming threshold-crossing edge, examining all neighbors of $v$ directly could be
too expensive. To avoid this, we group the incoming neighbors of $v$ by weight and introduce
$\tau+1=O(\log k)$ auxiliary vertices that each represent a geometrically sized
prefix of the sorted list.

After sorting $\Ain_v=\{u_1,\dots,u_k\}$ by decreasing weight, the Steiner
vertex $v_\ell$ represents the block $\{u_1,\dots,u_{k-2^\ell}\}$.
Edges are added so that for every $i\le k-2^\ell$, the two-edge path
$u_i \to v_\ell \to v$ has exactly the same total weight as the original edge $(u_i,v)$.
Thus routing through $v_\ell$ preserves all shortest-path distances. This construction aggregates many threshold crossings through Steiner vertices, so that only a small number of neighbors need to be scanned rather than all $k$ neighbors.

Additional edge is added in Steps~\ref{step:in-r} and \ref{step:in-tilde-r} when $v\in N$ or when $v$ is an N-Steiner vertex. These edges allow one-hop transitions from a negative vertex $r$ to its in-Steiner vertices, or the in-Steiner vertices of its N-Steiner vertices. They help tune the $h$-hop neighborhood of $r$ and are used in the proof of \cref{lem:plus-one-hop}.

\item \label{step:steiner-vertex-out} For every heavy vertex $v\in V$, if $\Aout_v\neq \emptyset$:
        \begin{itemize}
            \item Sort $\Aout_v=\{u_1,u_2,\dots,u_k\}$ so that $w(v,u_1)\geq w(v,u_2)\geq \dots\geq w(v,u_k)$; break ties arbitrarily.
            \item Let $\tau = \lceil \log k\rceil-1$. Create $\tau+1$ Steiner vertices $v_0,v_1,\dots,v_\tau$. We refer to these vertices as \emph{out-Steiner vertices}.
            \item For every $0\leq \ell\leq \tau$:
                \begin{itemize}
                    \item Set $\pout(v_\ell)=v$, we say $\pout(v_\ell)$ is the \emph{out-parent} of $v_\ell$.
                    \item We will later decide whether $v_\ell$ is light or heavy.
                \end{itemize}
            \item For every $0\leq \ell\leq \tau$:
                \begin{enumerate}
                    \item \label{step:out-up} Add an edge $(v,v_\ell)$ of weight $w(v,u_{k-2^{\ell}})$.
                    \item \label{step:out-down} For every $1\leq i \leq k-2^\ell$, add an edge $(v_\ell,u_i)$ of weight $w(v,u_i)-w(v,v_\ell)$. If $u_i$ is an in-Steiner vertex, add an edge $(v_\ell,\pin(u_i))$ of weight $w(v,u_i)+w(u_i,\pin(u_i))-w(v,v_\ell)$.
                    \item \label{step:out-r} If $v\in \bar N$, let $r\in N$ be the unique in-neighbor of $v$. Add an edge $(v_\ell,r)$ of weight $-w(r,v)-w(v,v_\ell)$.
                    \item \label{step:out-tilde-r} If $v$ is an N-Steiner vertex, add an edge $(v_\ell,\pn(v))$ of weight $w(v,\pn(v))-w(v,v_\ell)$.
                \end{enumerate}
        \end{itemize}
\smallskip
\noindent
\emph{Explanation.} This step builds an \emph{out-Steiner gadget} for every heavy vertex. As it is largely symmetric to Step~\ref{step:steiner-vertex-in}, we omit the details.

\item \label{step:tilde-r} For every negative vertex $r\in N$,
    \begin{enumerate}
        \item Create a new vertex $\tilde r_t$. Let $\tilde r_t$ be a heavy vertex. \item Set $\pn(\tilde r_t)=r$, we say $r$ is the \emph{N-parent} of $\tilde r_t$.
        \item \label{step:tilde-r-up} Add an edge $(\tilde r_t,r)$ of weight $\Delta_r$. Add an edge $(r,\tilde r_t)$ of weight $-\Delta_r$.
    \end{enumerate}

\smallskip
\noindent
\emph{Explanation.} This step introduces an N-Steiner vertex $\tilde r_t$ for every negative vertex $r$, which serves as a hub for shortcut paths associated with $r$. We connect $r$ and $\tilde r_t$ with bidirectional edges of weight $\pm \Delta_r$.

\item[] \hspace{-\leftmargin} {\bf Phase 4: Shortcut generation and postprocessing.}

\item \label{step:shortcut} For every negative vertex $r\in N$:
\begin{itemize}
    \item $\scin(r,t)$, $\scout(r,t)$.
    \item If $\Delta_r\leq 0$: For every $v\in \Ain_r$, $\ein(v,\tilde r_t,w(v,r)-\Delta_r)$.
    \item If $\Delta_r\geq -w(r,\bar r)$: For every $v\in \Aout_{\bar r}$, $\eout(\tilde r_t,v,w(r,\bar r)+\Delta_r+w(\bar r,v))$.
\end{itemize}

\smallskip
\noindent
\emph{Explanation.} This step invokes the procedures $\scin(r,t)$ and $\scout(r,t)$ for each negative vertex $r$ to add the shortcut edges needed to bypass the negative edge $(r,\bar r)$. The procedure $\scin(r,t)$ identifies vertices $u$ on the incoming side of shortest paths through $r$ and adds non-negative weight edges $(u,\tilde r_t)$, while $\scout(r,t)$ performs the analogous construction on the outgoing side. The implementation  of these procedures are described and explained in detail below.

We also handle the edge cases where $\Delta_r \leq 0$ or $\Delta_r \geq -w(r,\bar r)$, in which case either $\Uin_r$ or $\Uout_r$ is empty. In particular, if $\Delta_r\leq 0$, then $\Uin_r=\emptyset$, and we would like to add an edge $(r,\tilde r_t)$ with nonnegative weight $-\Delta_r$. However, to preserve the structural property that $r$ has a unique outgoing edge $(r,\bar r)$, we instead add an edge $(v,\tilde r_t)$ for every in-neighbor $v$ of $r$. The case $\Delta_r \geq -w(r,\bar r)$ is treated analogously.

\item \label{step:r'} For every negative vertex $r\in N$, let $W=w(r,\bar r)$, $\wout_r=\min\{w(r,v):v\in \Nout(r)\}$ and $\win_{r}=\min\{0,\min\{w(v,r):v\in \Nin(r)\}\}$.
    \begin{itemize}
        \item For every $v\in \Nin(r)$, let $w(v,r)=w(v,r)-\win_r$.
        \item Let $w(r,\bar r)=\win_r+\wout_r, w(\bar r,r)=-w(r,\bar r)$.
        \item For every $v\in \Nout(\bar r)$, let $w(\bar r,v)=w(\bar r,v)+W-\wout_r$.
        \item For every $v\in \Nout(r)\setminus \{\bar r\}$, add an edge $(\bar r,v)$ of weight $w(r,v)-\wout_r$ and delete the edge $(r,v)$.
    \end{itemize}

\smallskip
\noindent
\emph{Explanation.} After adding the shortcut edges (some of which may be negative), the graph $G$ may no longer be well-behaved. This step restores $G$ to a well-behaved graph. By \cref{clm:negative-incident}, every negative edge must be incident to a negative vertex. Therefore, we may first redirect outgoing edges from $r$ to instead originate from $\bar r$, and then apply appropriate vertex potential to $r$ and $\bar r$, so that $(r,\bar r)$ becomes the unique outgoing edge of $r$, and the only edge incident to $r$ that may have negative weight is $(r,\bar r)$.

\item \label{step:classify-in} For every in-Steiner $v$ created in Step~\ref{step:steiner-vertex-in}: classify $v$ as heavy if $|\Nout(v)|-1\geq\lambda$; classify $v$ as light otherwise.
\item \label{step:classify-out} For every out-Steiner $v$ created in Step~\ref{step:steiner-vertex-out}: classify $v$ as heavy if $|\Nin(v)|-1\geq\lambda$; classify $v$ as light otherwise.

\end{enumerate}
\end{ourbox}

\begin{remark} \label{rmk:r'}
In Step~\ref{step:r'}, for every negative vertex $r\in N$, we effectively replace the edge $(r,v)$ for every $v\in \Nout(r)\setminus \{\bar r\}$ by a length-two path $r\to \bar r\to v$ of the same weight. Henceforth, for simplicity, we may refer to such a construction simply as an $(r,v)$ edge, even though it may in fact correspond to a length-two path.
\end{remark}

\subsubsection*{Procedure $\ein(u,v,W)$}

The input consists vertices $u,v\in V$ and a weight $W$. The vertex $u$ is not a light in-Steiner vertex (in particular, $u$ may be an in-Steiner vertex created in the current iteration that has not yet been classified as light or heavy), and $v$ is either a negative vertex or an N-Steiner vertex. The weight $W$ can be negative; \cref{clm:negative-incident} shows that $W<0$ only if $v$ is a negative vertex.

The primary goal of this procedure is to add an edge $(u,v)$ of weight $W$ into the graph $G$. However, to preserve the invariants in \cref{lem:in-heavy-edge}, this insertion may force the insertion of further edges. In that case, we may recursively call $\ein$ to perform the necessary insertions.

\begin{ourbox}
\begin{itemize}
    \item Add an edge $(u,v)$ of weight $W$.
    \item If $u$ is an in-Steiner vertex:
        \begin{itemize}
            \item If $v\in N$, $\ein(\pin(u),v,W-w(u,\pin(u)))$.
            \item If $v$ is an N-Steiner vertex:
                \begin{itemize}
                    \item Add an edge $(\pin(u),v)$ of weight $W-w(u,\pin(u))$ to $\Fin_t$ (note that it is not added to $G$). If $W-w(u,\pin(u))\geq 0$, $\ein(\pin(u),v,W-w(u,\pin(u)))$.
                    \item $\ein(\pin(u),\pn(v),W-w(\pn(v),v)-w(u,\pin(u)))$.
                \end{itemize}
        \end{itemize}
    \item If $u$ is an out-Steiner vertex:
        \begin{itemize}
            \item $\ein(\pout(u),v,W+w(\pout(u),u))$.
        \end{itemize}
\end{itemize}
\end{ourbox}

For intuition, note that each recursive call requires $u$ to be either an in-Steiner or an out-Steiner vertex, and we recurse with $u$ replaced by $\pin(u)$ or $\pout(u)$. In this case, by \cref{obs:steiner} \ref{obs:pin} and \ref{obs:pout}, if $u$ is a level-$\ell$ vertex, $\pin(u)$ or $\pout(u)$ lies at some level $\ell'<\ell$. Therefore, the level strictly decreases in each recursive step, and the procedure terminates after at most $O(\log n)$ levels of recursion.

\subsubsection*{Procedure $\eout(u,v,W)$}
The input consists vertices $u,v\in V$ and a weight $W$. The vertex $u$ is either a negative vertex or an N-Steiner vertex, and $v$ is not a light out-Steiner vertex (in particular, $v$ may be an out-Steiner vertex created in the current iteration that has not yet been classified as light or heavy). The weight $W$ can be negative; \cref{clm:negative-incident} shows that $W<0$ only if $u$ is a negative vertex.

\begin{ourbox}
\begin{itemize}
    \item Add an edge $(u,v)$ of weight $W$.
    \item If $v$ is an in-Steiner vertex:
        \begin{itemize}
            \item $\eout(u,\pin(v),W+w(v,\pin(v)))$.
        \end{itemize}
    \item If $v$ is an out-Steiner vertex:
        \begin{itemize}
            \item If $u\in N$, $\eout(u,\pout(v),W-w(\pout(v),v))$.
            \item If $u$ is an N-Steiner vertex:
            \begin{itemize}
                \item Add an edge $(u,\pout(v))$ of weight $W-w(\pout(v),v)$ to $\Fout_t$ (note that it is not added to $G$). If $W-w(\pout(v),v)\geq 0$, $\eout(u,\pout(v),W-w(\pout(v),v))$.
                \item $\eout(\pn(u),\pout(v),W-w(u,\pn(u))-w(\pout(v),v))$.
            \end{itemize}
        \end{itemize}
\end{itemize}
\end{ourbox}

\subsubsection*{Procedure $\scin(r,t)$}
To build intuition, we first describe an inefficient brute force implementation of the procedure: For every $u\in \Nin(v)$, add an edge $(u,\tilde r_t)$ of weight $w(u,v)+d_v-\Delta_r$ if it is non-negative. This construction ensures the desired shortcut property as illustrated in \cref{fig:hard}. To obtain a more efficient procedure, we make use of the in-Steiner gadget created earlier for every heavy vertex processed (see Steps~\ref{step:scin-sort} and \ref{step:scin-steiner}). Apart from this optimization, Steps~\ref{step:scin-out}, \ref{step:scin-bf-other} and \ref{step:scin-light-out} follows the same logic as in the brute force implementation.

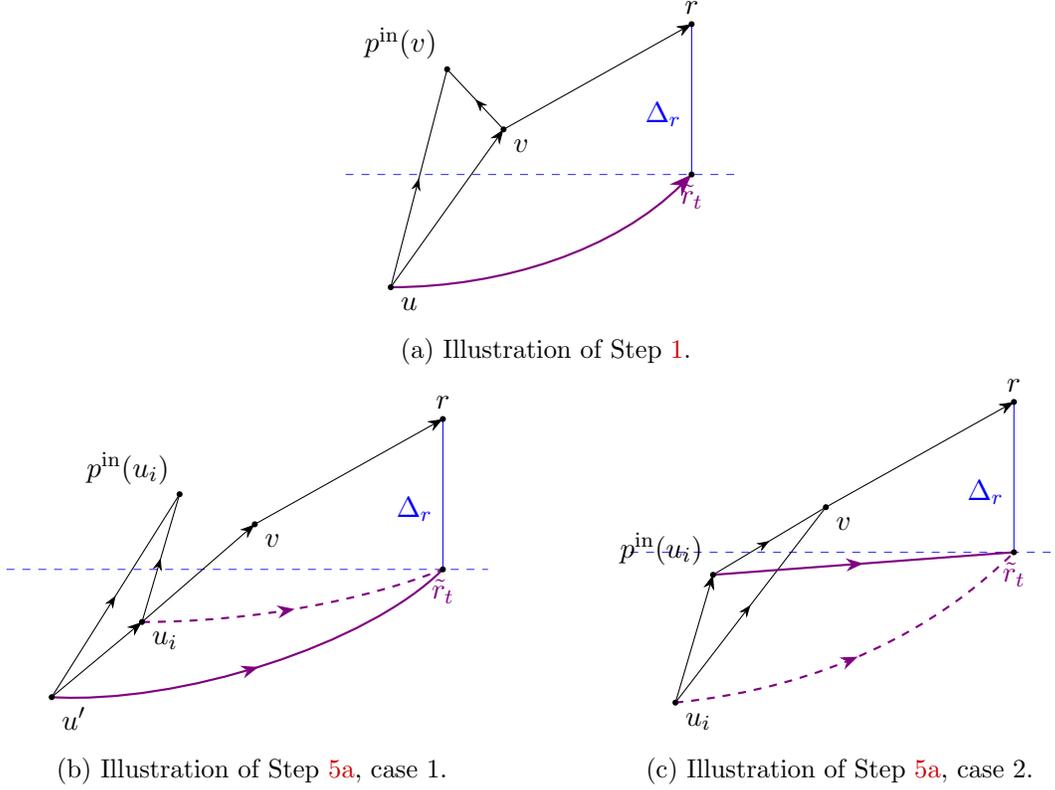
\begin{figure}[h]
    \centering
    \begin{subfigure}[b]{\textwidth}
        \centering
        \begin{tikzpicture}
  \coordinate (r)  at (2.5, 3);
  \coordinate (rb)  at (2.5, 1);
  \coordinate (rc)  at (2.5, 1.8);
  \coordinate (x)  at (0, 1.6);
  \coordinate (y)  at (-1.5, -0.5);
  \coordinate (xin)  at (-0.75, 2.4);

  \draw[blue,thin] (r) -- (rb);
  \node[blue,left] at (rc) {$\Delta_r$};

  \draw[dashedlvl] (-2.1,1) -- (3.1,1);

  \draw[black,thin,-{Stealth[length=2mm]}] (x) -- (r);
  \draw[black,thin,-{Stealth[length=2mm]}] (y) -- (x);
  \draw[black,thin,midarrow] (y) -- (xin);
  \draw[black,thin,midarrow] (x) -- (xin);

  \draw[violet,thick,-{Stealth[length=3mm]}] (y) .. controls (0,-0.5) and (1.5,0) .. (rb);
  
  % Labels
  \node[above]      at (r)  {$r$};
  \node[violet,below]       at (rb)  {$\tilde r_t$};
  \node[below right]      at (x)  {$v$};
  \node[below right]      at (y)  {$u$};
  \node[above left] at (xin)  {$p^{\mathrm{in}}(v)$};

  % Endpoint dots
  \node[dot] at (x)  {};
  \node[dot] at (r)  {};
  \node[dot] at (rb)  {};
  \node[dot] at (y)  {};
  \node[dot] at (xin)  {};
\end{tikzpicture}
        \caption{Illustration of Step~\ref{step:scin-in}.}
        \label{fig:scin-in}
    \end{subfigure}
    \\
    \begin{subfigure}[b]{0.45\textwidth}
        \centering
        \begin{tikzpicture}
  \coordinate (r)  at (2.5, 3);
  \coordinate (rb)  at (2.5, 1);
  \coordinate (rc)  at (2.5, 1.8);
  \coordinate (x)  at (0, 1.6);
  \coordinate (y)  at (-1.5, 0.3);
  \coordinate (yp)  at (-2.7, -0.7);
  \coordinate (yin)  at (-1, 2);

  \draw[blue,thin] (r) -- (rb);
  \node[blue,left] at (rc) {$\Delta_r$};

  \draw[dashedlvl] (-3.3,1) -- (3.1,1);

  \draw[black,thin,-{Stealth[length=2mm]}] (x) -- (r);
  \draw[black,thin,-{Stealth[length=2mm]}] (y) -- (x);
  \draw[black,thin,midarrow] (y) -- (yin);
  \draw[black,thin,midarrow] (yp) -- (yin);
  \draw[black,thin,-{Stealth[length=2mm]}] (yp) -- (y);

  \draw[violet,thick,dashed,midarrow] (y) .. controls (0,0.3) and (1.5,0.6) .. (rb);
  \draw[violet,thick,midarrow] (yp) .. controls (-1,-0.8) and (1.5,0) .. (rb);
  
  % Labels
  \node[above]      at (r)  {$r$};
  \node[violet,below]       at (rb)  {$\tilde r_t$};
  \node[below right]      at (x)  {$v$};
  \node[below right]      at (y)  {$u_i$};
  \node[below right]      at (yp)  {$u'$};
  \node[above left] at (yin)  {$p^{\mathrm{in}}(u_i)$};

  % Endpoint dots
  \node[dot] at (x)  {};
  \node[dot] at (r)  {};
  \node[dot] at (rb)  {};
  \node[dot] at (y)  {};
  \node[dot] at (yin)  {};
  \node[dot] at (yp)  {};
\end{tikzpicture}
        \caption{Illustration of Step~\ref{step:scin-bf-in}, case 1.}
        \label{fig:scin-bf-in-1}
    \end{subfigure}
    ~
    \begin{subfigure}[b]{0.45\textwidth}
        \centering
        \begin{tikzpicture}
  \coordinate (r)  at (2.5, 3);
  \coordinate (rb)  at (2.5, 1);
  \coordinate (rc)  at (2.5, 1.8);
  \coordinate (x)  at (0, 1.6);
  \coordinate (y)  at (-2, -1);
  \coordinate (yin)  at (-1.5, 0.7);

  \draw[blue,thin] (r) -- (rb);
  \node[blue,left] at (rc) {$\Delta_r$};

  \draw[dashedlvl] (-2.6,1) -- (3.1,1);

  \draw[black,thin,-{Stealth[length=2mm]}] (x) -- (r);
  \draw[black,thin,midarrow] (y) -- (x);
  \draw[black,thin,-{Stealth[length=2mm]}] (y) -- (yin);
  \draw[black,thin,midarrow] (yin) -- (x);

  \draw[violet,thick,dashed,midarrow] (y) .. controls (0,-0.8) and (1.5,0) .. (rb);
  \draw[violet,thick,midarrow] (yin) -- (rb);
  
  % Labels
  \node[above]      at (r)  {$r$};
  \node[violet,below]       at (rb)  {$\tilde r_t$};
  \node[below right]      at (x)  {$v$};
  \node[below right]      at (y)  {$u_i$};
  \node[above left] at (yin)  {$p^{\mathrm{in}}(u_i)$};

  % Endpoint dots
  \node[dot] at (x)  {};
  \node[dot] at (r)  {};
  \node[dot] at (rb)  {};
  \node[dot] at (y)  {};
  \node[dot] at (yin)  {};
\end{tikzpicture}
        \caption{Illustration of Step~\ref{step:scin-bf-in}, case 2.}
        \label{fig:scin-bf-in-2}
    \end{subfigure}
    \caption{Illustration of some steps in the procedure $\scin(r,t)$.}
    \label{fig:scin}
\end{figure}

Another objective is to avoid increasing the out-degree of light in-Steiner vertices, which is crucial for controlling the total number of vertices in the graph. To achieve this, we depart from the brute-force implementation in Steps~\ref{step:scin-in} and \ref{step:scin-bf-in}: instead of adding shortcut edges, they may introduce a new vertex for further consideration. These steps are illustrated in \cref{fig:scin}.

In Step~\ref{step:scin-in}, if $v$ is an in-Steiner vertex, we update $d_{\pin(v)}$ and do not process the in-neighbors of $v$ in later steps. As illustrated in \cref{fig:scin-in}, suppose that $u$ is an in-neighbor of $v$. Intuitively, we should add an edge $(u,\tilde r_t)$. We can safely not process the in-neighbors of $v$ because $\pin(v)$ is inserted into $S_{\ell'}$ and will be processed later. When $\pin(v)$ is processed, \cref{lem:in-heavy-edge}~\ref{item:in-1} guarantees that $u$ is an in-neighbor of $\pin(v)$. Therefore, the task of connecting $u$ to $\tilde r_t$ is effectively postponed until $\pin(v)$ is processed.

In Step~\ref{step:scin-bf-in}, there are two cases:
\begin{itemize}
    \item $W=w(u_i,v)+d_v-w(u_i,\pin(u_i))<\Delta_r$ (see \cref{fig:scin-bf-in-1}): In this case, the vertex $\pin(u_i)$ is ``above" the $\Delta_r$ line. Since $(v,u_i)$ crosses the $\Delta_r$ line, intuitively, we should add an edge $(u_i,\tilde r_t)$ (shown in dashed line). However, since $u_i$ is a light in-Steiner vertex, we hope to avoid increasing its out-degree. Recall that the purpose of adding the edge $(u_i,\tilde r_t)$ is to shortcut some shortest path $P$ on which this edge lies. Let $u'$ denote the predecessor of $u_i$ on the path $P$. Then by \cref{lem:in-heavy-edge}~\ref{item:in-1}, $u'$ is an in-neighbor of $\pin(u_i)$. Since $\pin(u_i)$ is inserted into $S_{\ell'}$, it will be processed later. When $\pin(u_i)$ is processed, since $u'$ is an in-neighbor of $\pin(u_i)$, we will consider adding an edge $(u',\tilde r_t)$, which achieves the same shortcutting effect as the edge $(u_i,\tilde r_t)$.
    \item $W\geq \Delta_r$ (see \cref{fig:scin-bf-in-2}): In this case, the vertex $\pin(u_i)$ is ``below" the $\Delta_r$ line. Then we can show that by \cref{lem:in-heavy-edge}~\ref{item:in-3} and \cref{obs:f}~\ref{obs:fin}, the edge $(\pin(u_i),v)$ must exist in the graph. Since $\pin(u_i)$ is an in-neighbor of $v$, the algorithm connects $\pin(u_i)$ to $\tilde r_t$, which achieves the same shortcutting effect as the edge $(u_i,\tilde r_t)$.
\end{itemize}

Moreover, in Steps~\ref{step:scin-in} and \ref{step:scin-bf-in}, every newly introduced vertex ($\pin(v)$ in Step~\ref{step:scin-in} and $\pin(u_i)$ in Step~\ref{step:scin-bf-in}) is guaranteed to lie at a strictly lower level than the currently processed vertex $v$. This is proved formally in \cref{lem:scin-dist-valid}. Hence, the procedure avoids circular insertions and repeated processing of a vertex.\\

Now we describe the procedure $\scin(r,t)$.\\

\begin{ourbox}
\begin{itemize}
    \item Initialize empty sets $S_0,S_1,\dots,S_t$. For every $v\in V$, initialize $d_v=\infty$.
    \item For every $v\in \Uin_r$: Let $d_v=d^0(v,r)$, and insert $v$ into $S_\ell$ if $v$ is a level-$\ell$ vertex.
    \item For $\ell=\{t,t-1,\dots,0\}$: For every vertex $v\in S_{\ell}$: 
        \begin{itemize}
            \item If $v$ is an in-Steiner vertex:
                \begin{enumerate}
                    \item \label{step:scin-in} Update $d_{\pin(v)}=\min\{d_{\pin(v)},d_v-w(v,\pin(v))\}$, and insert $\pin(v)$ into $S_{\ell'}$ if $\pin(v)$ is a level-$\ell'$ vertex.
                \end{enumerate}
            \item If $v$ is an out-Steiner vertex:
                \begin{enumerate}[start=2]
                    \item \label{step:scin-out} Let $W=w(\pout(v),v)+d_v-\Delta_r$. If $W\geq 0$, $\ein(\pout(v),\tilde r_t,W)$.
                \end{enumerate}
            \item If $v$ is not an in-Steiner vertex, and $v$ is heavy:
                \begin{enumerate}[start=3]
                    \item \label{step:scin-sort} Recall that we have $\Ain_v=\{u_1,\dots,u_k\}$ sorted in non-increasing order of $w(u_i,v)$. Let $j$ be the smallest index such that $w(u_j,v)+d_v-\Delta_r<0$; set $j=k$ if no such index exists.
                    \item \label{step:scin-steiner} Let $\ell=\lceil \log (k-j+1)\rceil$. If $\ell\leq \lceil \log k\rceil-1$: $\ein(v_\ell,\tilde r_t,w(v_\ell,v)+d_v-\Delta_r)$.
                    \item \label{step:scin-bf} For every $\max \{k-2^{\ell}+1,1\}\leq i\leq k$:
                        \begin{enumerate}
                            \item \label{step:scin-bf-in} If $u_i$ is a light in-Steiner vertex: Let $W=w(u_i,v)+d_v-w(u_i,\pin(u_i))$. If $W<\Delta_r$, update $d_{\pin(u_i)}=\min\{d_{\pin(u_i)},W\}$, and insert $\pin(u_i)$ into $S_{\ell'}$ if $\pin(u_i)$ is a level-$\ell'$ vertex.
                            \item \label{step:scin-bf-other} Otherwise, $u_i$ is not a light in-Steiner vertex: Let $W=w(u_i,v)+d_v-\Delta_r$. If $W\geq 0$, $\ein(u_i,\tilde r_t,W)$.
                        \end{enumerate}
                \end{enumerate}
            \item If $v$ is not an in-Steiner vertex, and $v$ is light:
                \begin{enumerate} [start=6]
                    \item \label{step:scin-light-out} For every $u\in \Ain_v$: Let $W=w(u,v)+d_v-\Delta_r$. If $W\geq 0$, $\ein(u,\tilde r_t,W)$.
                \end{enumerate}
        \end{itemize}
\end{itemize}
\end{ourbox}

\subsubsection*{Procedure $\scout(r,t)$}
This procedure adds outgoing shortcut edges from $\tilde r_t$. As it is largely symmetric to $\scin(r,t)$, we omit the explanation here.

\begin{ourbox}
\begin{itemize}
    \item Initialize empty sets $S_0,S_1,\dots,S_t$. For every $v\in V$, initialize $d_v=\infty$.
    \item For every $v\in \Uout_r$: Let $d_v=d^1(r,v)$, and insert $v$ into $S_\ell$ if $v$ is a level-$\ell$ vertex.
    
    \item For $\ell=\{t,t-1,\dots,0\}$: For every vertex $v\in S_{\ell}$:
        \begin{itemize}
            \item If $v$ is an out-Steiner vertex:
                \begin{enumerate}
                    \item \label{step:scout-out} Update $d_{\pout(v)}=\min\{d_{\pout(v)},d_v-w(\pout(v),v)\}$, and insert $\pout(v)$ into $S_{\ell'}$ if $\pout(v)$ is a level-$\ell'$ vertex.
                \end{enumerate}
            \item If $v$ is an in-Steiner vertex:
                \begin{enumerate}[start=2]
                    \item \label{step:scout-in} Let $W=d_v+w(v,\pin(v))+\Delta_r$. If $W\geq 0$, $\eout(\tilde r_t,\pin(v),W)$.
                \end{enumerate}
            \item If $v$ is not an out-Steiner vertex, and $v$ is heavy:
                \begin{enumerate}[start=3]
                    \item \label{step:scout-sort} Recall that we have $\Aout_v=\{u_1,\dots,u_k\}$ sorted in non-increasing order of $w(v,u_i)$. Let $j$ be the smallest index such that $d_v+w(v,u_j)+\Delta_r<0$; set $j=k$ if no such index exists.
                    \item \label{step:scout-steiner} Let $\ell=\lceil \log (k-j+1)\rceil$. If $\ell\leq \lceil \log k\rceil-1$: $\eout(\tilde r_t,v_\ell,d_v+w(v,v_\ell)+\Delta_r)$.
                    \item \label{step:scout-bf} For every $\max \{k-2^{\ell}+1,1\}\leq i\leq k$:
                        \begin{enumerate}
                            \item \label{step:scout-bf-out} If $u_i$ is a light out-Steiner vertex: Let $W=d_v+w(v,u_i)-w(\pout(u_i),u_i)$. If $W<-\Delta_r$, update $d_{\pout(u_i)}=\min\{d_{\pout(u_i)},W\}$, and insert $\pout(u_i)$ into $S_{\ell'}$ if $\pout(u_i)$ is a level-$\ell'$ vertex.
                            \item \label{step:scout-bf-other} Otherwise, $u_i$ is not a light out-Steiner vertex: Let $W=d_v+w(v,u_i)+\Delta_r$. If $W\geq 0$, $\eout(\tilde r_t,u_i,W)$.
                        \end{enumerate}
                \end{enumerate}
            \item If $v$ is not an out-Steiner vertex, and $v$ is light:
                \begin{enumerate}[start=6]
                    \item \label{step:scout-light-in} For every $u\in \Aout_v$: Let $W=d_v+w(v,u)+\Delta_r$. If $W\geq 0$, $\eout(\tilde r_t,u,W)$.
                \end{enumerate}
        \end{itemize}
\end{itemize}
\end{ourbox}

\subsection{Correctness and Runtime Analysis}

We now establish the correctness and runtime of various steps in the procedure $\shortcut$, laying the groundwork for bounding the overall runtime of our single-source shortest path algorithm in \Cref{sec:runtime}. The analysis is organized as follows.

\begin{itemize}
    \item In \Cref{sec:structural_invariants}, we establish \emph{local structural invariants} that characterize the parent relationships of Steiner vertices, the behavior of
    deferred edges, and the precise form of the neighborhoods of in/out-Steiner vertices.
    These properties ensure that shortcutting does not result in uncontrolled degree growth
    and that all negative edges remain incident at the original negative vertices.

    \item In \Cref{sec:scin}, we analyze the procedures $\scin$ and $\scout$ directly and
    prove bounds on the work they perform per negative vertex.

    \item In \Cref{sec:overall}, we combine these ingredients to prove several global guarantees:
    bounding the growth in the number of vertices, showing that shortcut edges do not decrease distances, and that the hop distance is reduced by a constant factor in each iteration.

    \item In \cref{sec:runtime}, we bound the total number of Steiner vertices created, and analyze the total work performed by the single-source shortest path algorithm.
\end{itemize}

\subsubsection{Local Structural Invariants}
\label{sec:structural_invariants}

\begin{observation}[Properties of in-Steiner and out-Steiner vertices] \label{obs:steiner}
\noindent
\begin{enumerate}[label=(\alph*)]
    \item \label{obs:lower-level-in} If $v$ is a level-$t$ in-Steiner vertex, then its parent $\pin(v)$ is a level-$\ell$ vertex for some $\ell<t$.
    \item \label{obs:lower-level-out} If $v$ is a level-$t$ out-Steiner vertex, then its parent $\pout(v)$ is a level-$\ell$ vertex for some $\ell<t$.
    \item \label{obs:pin} An in-Steiner vertex $v$ can never be the in-parent of another in-Steiner vertex $u$.
    \item \label{obs:pout} An out-Steiner vertex $v$ can never be the out-parent of another out-Steiner vertex $u$.
\end{enumerate}
\end{observation}

\begin{proof}
For \cref{obs:lower-level-in}, observe that whenever a level-$t$ in-Steiner vertex $v$ is created in Step~\ref{step:steiner-vertex-in}, its in-parent $\pin(v)$ must be a vertex that already exists in the graph prior to the $t$-th iteration. Therefore, $\pin(v)$ is a level-$\ell$ vertex for some $\ell < t$. The argument for \cref{obs:lower-level-out} is analogous.

For \cref{obs:pin}, observe that in Step~\ref{step:steiner-vertex-in}, for a heavy vertex $v$, we create in-Steiner vertices $u$ with $\pin(u)=v$ only when $\Ain_v\neq \emptyset$. However, if $v$ is an in-Steiner vertex, then in Step~\ref{step:delete} we set $\Ain_v=\emptyset$. Therefore, an in-Steiner vertex can never be the in-parent of any vertex. The argument for \cref{obs:pout} is analogous.
\end{proof}

In Step~\ref{step:f-in-f-out}, we distinguish between $\Fin_{t-1},\Fout_{t-1}$ and $\Fin_t,\Fout_t$ to ensure that the algorithm does not consider the edges newly added during the procedures $\ein$ and $\eout$. In other steps, for simplicity, we use $\Fin$ (resp. $\Fout$) to denote the set $\Fin_t$ (resp. $\Fout_t$) at the current iteration $t$.

\begin{observation} [Insertion of deferred edges] \label{obs:f}
\noindent
\begin{enumerate}[label=(\alph*)]
    \item \label{obs:fin} For an edge $e\in \Fin_t$, if its current weight $w(e)\geq 0$, then $e$ is added to $G$.
    \item \label{obs:fout} For an edge $e\in \Fout_t$, if its current weight $w(e)\geq 0$, then $e$ is added to $G$.
\end{enumerate}
\end{observation}
\begin{proof}
We prove \cref{obs:fin}; the proof of \cref{obs:fout} is analogous. In Step~\ref{step:f-in-f-out}, we begin by initializing $\Fin_t=\Fin_{t-1}$. For every edge $e \in \Fin_{t-1}$ with current weight $w(e) \geq 0$, it is inserted into $G$ via a call to $\ein$. Moreover, observe that the execution of $\ein$ may add new edges to $\Fin_t$. For each such newly added edge, if it has non-negative weight, the procedure recursively invokes $\ein$ and inserts it into $G$.

\end{proof}

As described in Step~\ref{step:delete}, the edges excluded from $\{\Ain_v, \Aout_v\}_{v \in V}$ are summarized in the following observation: these are precisely the parent-child edges inside the Steiner gadgets created by the algorithm.

\begin{observation} [Edges excluded from auxiliary lists $\{\Ain_v,\Aout_v\}_{v\in V}$]\label{obs:list}
\noindent
\begin{enumerate}[label=(\alph*)]
    \item \label{obs:in-list} Let $v \in V$ and $u \in \Nin(v) \setminus \Ain_v$ with $w(u,v)\geq 0$. Then one of the following holds:
    \begin{itemize}
    \item $u$ is an in-Steiner vertex with $\pin(u) = v$.
    \item $v$ is an in-Steiner vertex.
    \item $v$ is an out-Steiner vertex with $\pout(v)=u$.
    \end{itemize}
    \item \label{obs:out-list} Let $v \in V$ and $u \in \Nout(v) \setminus \Aout_v$ with $w(u,v)\geq 0$. Then one of the following holds:
    \begin{itemize}
    \item $u$ is an out-Steiner vertex with $\pout(u) = v$.
    \item $v$ is an out-Steiner vertex.
    \item $v$ is an in-Steiner vertex with $\pin(v)=u$.
    \end{itemize}
\end{enumerate}
\end{observation}
\begin{proof}
It follows immediately from the description of Step~\ref{step:delete}.
\end{proof}

\begin{claim} [Running time of $\ein$ and $\eout$] \label{clm:addedge-time}
\noindent
Each call to procedures $\ein(u,v,W)$ and $\eout(u,v,W)$ takes $O(L^2)=\tilde O(1)$.
\end{claim}
\begin{proof}
We will bound here the runtime of $\ein$, the proof of $\eout$ runtime follows analogously. Consider the procedure $\ein(u,v,W)$, where $u$ is a level-$\ell$ vertex. Observe that any recursive call to $\ein$ must satisfy one of the following: either $u$ is an in-Steiner vertex, in which case the first parameter becomes $\pin(u)$, or $u$ is an out-Steiner vertex, in which case the first parameter becomes $\pout(u)$. In either case, \cref{obs:steiner}~\ref{obs:lower-level-in} and \ref{obs:lower-level-out}, $\pin(v)$ or $\pout(v)$ is a level-$\ell'$ vertex for some $\ell'<\ell$. Therefore, in every recursive call, the vertex in the first parameter strictly decreases in level, so the recursion depth is bounded by $t\leq L$.

We now analyze the second parameter $v$, which is either a negative vertex or an N-Steiner vertex. If $v \in N$, then at most one recursive call is made and the second parameter remains unchanged. Therefore, when $v \in N$, the running time of $\ein(u,v,W)$ is bounded by $O(L)$. If instead $v$ is an N-Steiner vertex, the procedure may make one recursive call with the second parameter unchanged and another recursive call in which the second parameter becomes $\pn(v)$, which is a negative vertex. In this case, since the recursion depth is bounded by $t\leq L$, the running time of $\ein(u,v,W)$ is bounded by $O(L^2)$.

\end{proof}

Now we analyze the edges added by the algorithm. In Steps~\ref{step:f-in}, \ref{step:f-out} and \ref{step:shortcut}, the algorithm may add edges into $G$ via procedures $\ein$ and $\eout$. In addition, edges are added in Steps~\ref{step:negative-left}, \ref{step:negative-right}, \ref{step:in-up}, \ref{step:in-down}, \ref{step:in-r}, \ref{step:in-tilde-r}, \ref{step:out-up}, \ref{step:out-down}, \ref{step:out-r}, \ref{step:out-tilde-r}, \ref{step:tilde-r-up}, \ref{step:r'} of $\shortcut(G,t)$. During the execution of $\scin$, edges are added via the procedure $\ein$, which is called in Steps~\ref{step:scin-out}, \ref{step:scin-steiner}, \ref{step:scin-bf-other}, \ref{step:scin-light-out}. During the execution of $\scout$, edges are added via the procedure $\eout$, which is called in Steps~\ref{step:scout-in}, \ref{step:scout-steiner}, \ref{step:scout-bf-other}, \ref{step:scout-light-in}.

\begin{claim} \label{clm:negative-incident}
Before Step~\ref{step:r'} of $\shortcut(G,t)$, every negative edge in $G$ is incident to a negative vertex in $N$.
\end{claim}
\begin{proof}
Since we only apply valid vertex potentials, a non-negative weight edge will never become negative. Therefore, it suffices to consider the newly added edges. First consider the edges added in the following steps of $\shortcut(G,t)$,
\begin{itemize}
    \item Steps~\ref{step:negative-left}, \ref{step:negative-right}: These edges are incident to a negative vertex.
    \item Steps~\ref{step:in-up}, \ref{step:in-down}, \ref{step:out-up}, \ref{step:out-down}: These edges are non-negative.
    \item Steps~\ref{step:in-r}, \ref{step:in-tilde-r}, \ref{step:out-r}, \ref{step:out-tilde-r}: These edges are incident to a negative vertex.
    \item Step~\ref{step:tilde-r-up}: These edges are incident to a negative vertex.
\end{itemize}
Now we consider the edges added via procedures $\ein$ and $\eout$. For simplicity, we focus on $\ein$. We first observe that the initial calls to $\ein(u,v,W)$ (excluding recursive calls) always satisfy $W\geq 0$. The only step that requires justification is Step~\ref{step:scin-steiner} of $\scin$. At that step, $\Ain_v=\{u_1,\dots,u_k\}$ is sorted so that $w(u_1,v)\geq \dots\geq w(u_k,v)$, and we have computed the smallest index $j$ such that $w(u_j,v)+d_v-\Delta_r<0$ and $\ell=\lceil \log(k-j+1)\rceil$. Then we have $k-2^{\ell}\leq k-(k-j+1) = j-1$. By the minimality of $j$, $w(u_{j-1},v)+d_v-\Delta_r\geq 0$ (the case $j=1$ cannot occur since $\ein$ is called only if $\ell\leq \lceil \log k\rceil-1$). Therefore,
\[
w(v_\ell,v)+d_v-\Delta_r=w(u_{k-2^{\ell}},v)+d_v-\Delta_r \geq w(u_{j-1},v)+d_v-\Delta_r\geq 0,
\]
and the weight is non-negative.

Now we show that for any call to the procedure $\ein(u,v,W)$ with $W\geq 0$, any negative edge added by the procedure (including those added during recursions) must be incident to a negative vertex. The claim is immediate if $v$ is a negative vertex. Therefore, we assume in the following that $v$ is an N-Steiner vertex. We show that whenever the procedure makes a recursive call to $\ein(u',v',W')$, either $v'\in N$ or $W'\geq 0$ holds.

If $u$ is an in-Steiner vertex, the procedure recursively calls $\ein(\pin(u),\pn(v),W-w(\pn(v),v)-w(u,\pin(u)))$; since $\pn(v)$ is a negative vertex, the claim follows. It may also calls $\ein(\pin(u),v,W-w(u,\pin(u)))$, but this occurs only when $W-w(u,\pin(u))\geq 0$. If $u$ is an out-Steiner vertex, then the procedure calls $\ein(\pout(u),v,W+w(\pout(u),u))$. Since $W\geq 0$ and $w(\pout(u),u)\geq 0$, the new weight is also non-negative.

Therefore, by induction, every negative edge added by the procedure $\ein$ is incident to a negative vertex. An analogous argument holds for $\eout$.
\end{proof}

\begin{claim}
In the end of the algorithm, the graph $G$ is a well-behaved graph. Moreover, every negative vertex in $G$ (that is, any vertex with an outgoing negative edge) belongs to $N$.
\end{claim}
\begin{proof}
By \cref{clm:negative-incident}, every negative edge in $G$ is incident to a negative vertex in $N$. Fix a negative vertex $r\in N$, after Step~\ref{step:r'}, it is easy to see that every incoming edge of $r$ has non-negative weight and every outgoing edge of $\bar r$ has non-negative weight. Therefore, the only edges in $G$ that may have negative weight are the edges $(r,\bar r)$ for all $r\in N$, each of which is the unique outgoing edge of $r$. Moreover, $r$ has out-degree 1 and $\bar r$ has in-degree 1. Hence, $G$ is a well-behaved graph.
\end{proof}

We now turn to the properties of heavy and light vertices and begin with some simple observations.

\begin{observation} \label{obs:light-vertex}
Every light vertex is either an in-Steiner vertex, or an out-Steiner vertex.
\end{observation}
\begin{proof}
As described in the initialization process, all vertices in the original graph are heavy. Apart from the in-Steiner and out-Steiner vertices, the only other Steiner vertices we create are the N-Steiner vertices, which are also heavy as specified in Step~\ref{step:tilde-r}.
\end{proof}

\begin{observation} \label{obs:edge-regular}
In Steps~\ref{step:negative-left}, \ref{step:negative-right}, \ref{step:tilde-r-up} of $\shortcut(G,t)$, it only adds edges between regular vertices.
\end{observation}
\begin{proof}
In Steps~\ref{step:negative-left} and \ref{step:negative-right}, the endpoints of the added edges must lie in $N$ and $\bar N$, which are regular by definition. In Step~\ref{step:tilde-r-up}, the endpoints of the added edges must belong to $N$ or be N-Steiner vertices, all of which are regular by definition.
\end{proof}

Intuitively, when an in-Steiner vertex $v$ is created and classified as light, it has small out-degree at the end of that iteration. In subsequent iterations, our goal is to (effectively) avoid adding new outgoing edges from $v$. This ensures that the out-degree of $v$ remains small, allowing us to handle its out-neighbors by brute force.

In the following two lemmas, we precisely characterize the in-neighborhood and out-neighborhood of an in-Steiner vertex. The next two lemmas provide analogous characterizations for out-Steiner vertices. The proofs of \cref{lem:in-heavy-edge,lem:out-heavy-edge}, as well as \cref{clm:ein-light}, \cref{lem:in-light-edge}, \cref{clm:eout-light} and \cref{lem:out-light-edge} are mutually interdependent, and we establish them simultaneously via induction on the sequence of graph operations. Specifically, we assume that all statements hold before a given update performed by the algorithm, and then show that they continue to hold afterward.

\begin{lemma} \label{lem:in-heavy-edge}
Let $v$ be a level-$t$ in-Steiner vertex. By the end of the $t$-th iteration, $\Nout(v)$ only consists of $\pin(v)$ and some level-$t$ N-Steiner vertices. Moreover, the following properties hold throughout:
\begin{enumerate}[label=(\alph*)]
    \item \label{item:in-1} For every $u\in \Nin(v)$, there is an edge $(u,\pin(v))$ in $G$ of weight at most $w(u,v)+w(v,\pin(v))$.
    \item \label{item:in-2} For every $u\in \Nout(v)$, one of the following holds:
    \begin{itemize}
        \item $u=\pin(v)$.
        \item $u$ is an out-Steiner vertex with $\pout(u)=v$.
        \item $u\in N$.
        \item $u$ is an N-Steiner vertex.
        \item $u$ is an in-Steiner vertex with $\pin(u)\in \Nout(v)$ and $\pin(u)\in N$.
        \item $u$ is an in-Steiner vertex and $\pin(u)\in \Nout(v)$ is an N-Steiner vertex.
    \end{itemize}
    \item \label{item:in-3} For every $u\in \Nout(v)$:
    \begin{itemize}
        \item If $u\in N$, then there is an edge $(\pin(v),u)$ of weight at most $w(v,u)-w(v,\pin(v))$.
        \item If $u$ is an N-Steiner vertex, then there is an edge $(\pin(v),u)$ of weight $w(v,u)-w(v,\pin(v))$ in $\Fin$, and there is an edge $(\pin(v),\pn(u))$ of weight at most $w(v,u)-w(\pn(u),u)-w(v,\pin(v))$ in $G$.
    \end{itemize}
\end{enumerate}
\end{lemma}
\begin{proof}
By \cref{obs:edge-regular}, we only need to consider the edges added in Steps~\ref{step:in-up}, \ref{step:in-down}, \ref{step:in-r}, \ref{step:in-tilde-r}, \ref{step:out-up}, \ref{step:out-down}, \ref{step:out-r}, \ref{step:out-tilde-r}, \ref{step:r'} of $\shortcut(G,t)$, and the edges added via procedures $\ein$ and $\eout$.

\paragraph{Proof of (a).}
We first consider the incoming edges of $v$ and prove the first item. Specifically, we show that whenever an incoming edge to $v$ is added, the claim holds. That is, when an edge $(u,v)$ of weight $w(u,v)$ is added, an edge $(u,\pin(v))$ of weight at most $w(u,v)+w(v,\pin(v))$ is also added. After the edges are added, the algorithm may apply vertex potentials. However, it is easy to see that the claim continues to hold after applying any vertex potential.

We first focus on Steps~\ref{step:in-up}, \ref{step:in-down}, \ref{step:in-r}, \ref{step:in-tilde-r}, \ref{step:out-up}, \ref{step:out-down}, \ref{step:out-r}, \ref{step:out-tilde-r} of $\shortcut(G,t)$. When vertex $v$ is created in Step~\ref{step:steiner-vertex-in} of the $t$-th iteration, the algorithm adds incoming edges to $v$ in Step~\ref{step:in-down}. In particular, for some in-neighbor $u$ of $\pin(v)$, it may add an edge $(u,v)$ of weight $w(u,\pin(v))-w(v,\pin(v))$. This edge satisfies the claim trivially. If $u$ is an out-Steiner vertex, the algorithm may add an edge $(\pout(u),v)$ of weight $w(\pout(u),u)+w(u,\pin(v))-w(v,\pin(v))$. Note that since $u$ is an out-Steiner vertex and $\pin(v)\in \Nout(u)$, \cref{lem:out-heavy-edge}~\ref{item:out-1} implies that there is an edge $(\pout(u),\pin(v))$ of weight at most $w(\pout(u),u)+w(u,\pin(v))$. Therefore, the edge $(\pout(u),v)$ also satisfies the claim. Next, if $\pin(v)\in N$, the algorithm adds an edge $(\pin(v),v)$ of weight $-w(v,\pin(v))$ in Step~\ref{step:in-r}. Finally, if $\pin(v)$ is an N-Steiner vertex, it adds an edge $(\pn(\pin(v)),v)$ of weight $w(\pn(\pin(v)),\pin(v))-w(v,\pin(v))$ in Step~\ref{step:in-tilde-r}. It is straightforward to verify that all these edges trivially satisfy the claim.

In subsequent iterations, the algorithm never adds incoming edges to $v$ in Steps~\ref{step:in-down} and~\ref{step:out-up} as edges added in these steps are directed into newly created vertices. The same holds for Steps~\ref{step:in-r}, \ref{step:in-tilde-r}, \ref{step:out-r}, \ref{step:out-tilde-r}, because the endpoints of the added edges are either negative vertices or newly created vertices. In addition, it also never adds incoming edges to $v$ in Step~\ref{step:in-up}, since $v$ is an in-Steiner vertex so $\Ain_v=\emptyset$. The only remaining case is Step~\ref{step:out-down}. In this case, since $v$ is an in-Steiner vertex, \cref{obs:steiner}~\ref{obs:pin} implies that it can never be the in-parent of any vertex. Thus, the algorithm can only add an edge $(u,v)$ of weight $W$ for some out-Steiner vertex $u$, followed by an edge $(u,\pin(v))$ of weight $W+w(v,\pin(v))$, which satisfies the claim.

Next we focus on the edges added by the procedures $\ein$ and $\eout$. Since procedure $\ein$ adds edges only into a negative vertex or an N-Steiner vertex, incoming edges to $v$ can be added only via $\eout$. When an edge $(u,v)$ of weight $W$ is added for some $u\in V$, since $v$ is an in-Steiner vertex, the algorithm must have recursively invoked $\eout(u,\pin(v),W+w(v,\pin(v)))$ and added an edge $(u,\pin(v))$ of weight $W+w(v,\pin(v))$.

Finally, it is straightforward to verify that if the claim holds before Step~\ref{step:r'}, then it continues to hold after Step~\ref{step:r'}. This completes the proof of the first item.

\paragraph{Proof of (b).}
Now we consider the outgoing edges from $v$ and prove the second item. During the $t$-th iteration, the algorithm first adds the edge $w(v,\pin(v))$ in Step~\ref{step:in-up}. Then, in Step~\ref{step:scin-steiner} of $\scin(r,t)$ for some negative vertex $r\in N$, it may invoke $\ein(v,\tilde r_t,W)$ and add an edge $(v,\tilde r_t)$ of weight $W$. Note that we will never add any outgoing edge from $v$ in recursive calls to $\ein$ since $v$ is neither the in-parent nor the out-parent of any vertex. Therefore, by the end of the $t$-th iteration, $\Nout(v)$ only contains $\pin(v)$ and some level-$t$ N-Steiner vertices.

In subsequent iterations, consider Steps~\ref{step:in-up}, \ref{step:in-down}, \ref{step:in-r}, \ref{step:in-tilde-r}, \ref{step:out-up}, \ref{step:out-down}, \ref{step:out-r}, \ref{step:out-tilde-r} of $\shortcut(G,t)$. The algorithm never adds outgoing edges from $v$ in Steps~\ref{step:in-up} and~\ref{step:out-down} as edges added in these steps are directed out from newly created vertices. The same holds for Steps~\ref{step:in-r}, \ref{step:in-tilde-r}, \ref{step:out-r}, \ref{step:out-tilde-r}, because the endpoints of the added edges are either negative vertices or newly created vertices. In Step~\ref{step:out-up}, the algorithm may only add an edge $(v,u)$ for some out-Steiner vertex $u$ with $\pout(u)=v$, which satisfies the claim. The only remaining case is Step~\ref{step:in-down}. There are two cases:
\begin{itemize}
    \item The first case is the algorithm adds an edge $(v,u)$ for some in-Steiner vertex $u$, given that the edge $(v,\pin(u))$ already exists in the graph. Assume \cref{lem:in-heavy-edge}~\ref{item:in-2} holds before adding this edge. Since $\pin(u)\in \Nout(v)$, and an in-Steiner vertex cannot be the in-parent of any vertex by \cref{obs:steiner}~\ref{obs:pin}, $\pin(u)$ must fall into one of the following categories: $\pin(u)$ is an out-Steiner vertex with $\pout(\pin(u))=v$; $\pin(u)\in N$; or $\pin(u)$ is an N-Steiner vertex. We now show that the first case is not possible. This is because the algorithm adds the edge $(v,u)$ only if $v\in \Ain_{\pin(u)}$. However, if $v=\pout(\pin(u))$, then $v$ would have been removed from $\Ain_{\pin(u)}$ in Step~\ref{step:delete} of $\shortcut(G,t)$. Therefore, it must be either $\pin(u)\in N$ or $\pin(u)$ is an N-Steiner vertex, and \cref{lem:in-heavy-edge}~\ref{item:in-2} continues to hold after adding the edge $(v,u)$.
    \item The second case is the algorithm adds an edge $(v',u)$ for some in-Steiner vertex $u$ and out-Steiner vertex $v'$ with $\pout(v')=v$, given that the edge $(v',\pin(u))$ already exists in the graph. Then the algorithm adds an edge $(\pout(v'),u)=(v,u)$. In this case, since $v'$ is an out-Steiner vertex with $\pout(v')=v$ and $\pin(u)\in \Nout(v')$, by \cref{lem:out-heavy-edge}~\ref{item:out-1}, we must have $\pin(u)\in \Nout(v)$. Assume \cref{lem:in-heavy-edge}~\ref{item:in-2} holds before adding this edge. Since $\pin(u)\in \Nout(v)$, and an in-Steiner vertex cannot be the in-parent of any vertex by \cref{obs:steiner}~\ref{obs:pin}, $\pin(u)$ must fall into one of the following categories: $\pin(u)$ is an out-Steiner vertex with $\pout(\pin(u))=v$; $\pin(u)\in N$; or $\pin(u)$ is an N-Steiner vertex. We again show that the first case is not possible. This is because in that case, $\pin(u)$ is an out-Steiner vertex and $v'\in \Nin(\pin(u))$. Since $v'$ is an out-Steiner vertex with $\pout(v')=v$, but $v$ is neither in $\bar N$ nor an N-Steiner vertex, this would contradict \cref{lem:out-heavy-edge}~\ref{item:out-2} applied to the vertex $\pin(u)$. Therefore, it must be either $\pin(u)\in N$ or $\pin(u)$ is an N-Steiner vertex, and the \cref{lem:in-heavy-edge}~\ref{item:in-2} continues to hold after adding the edge $(v,u)$.
\end{itemize}

Next, consider the edges added by the procedure $\ein$ and $\eout$. Since procedure $\eout$ add edges only from a negative vertex or an N-Steiner vertex, outgoing edges from $v$ can be added only via $\ein$. In any call to $\ein(v,u,W)$, $u$ must be a negative vertex or an N-Steiner vertex, and therefore the claim always holds. Finally, the algorithm never adds outgoing edges from $v$ in Step~\ref{step:r'}. This completes the proof for the second item.

\paragraph{Proof of (c).}
For the third item, note that the algorithm adds outgoing edges from $v$ to negative vertices or to N-Steiner vertices only via the procedure $\ein$. The claim therefore follows directly from the description of $\ein$. It is easy to see that the claim continues to hold after applying any vertex potential.
\end{proof}

\begin{claim} \label{clm:ein-light}
For every call $\ein(u,v,W)$ made by the algorithm, the vertex $u$ is not a light in-Steiner vertex.
\end{claim}
\begin{proof}
We focus on the initial calls to the procedure $\ein$. For the recursive calls, the claim is immediate since a light vertex cannot be the in-parent or out-parent of any vertex. We now consider each location where $\ein$ is invoked:
\begin{itemize}
    \item Step~\ref{step:shortcut} of $\shortcut(G,t)$: Since $r\in \Nout(v)$, \cref{lem:in-light-edge} implies that $v$ is not a light in-Steiner vertex. Note that $r$ cannot be the in-parent of $v$, as in that case $v$ would have been deleted from $\Aout_r$ in Step~\ref{step:delete}.
    \item Step~\ref{step:scin-out} of $\scin$: $\pout(v)$ is heavy.
    \item Step~\ref{step:scin-steiner} of $\scin$: $v_\ell$ is a vertex created in the current iteration.
    \item Step~\ref{step:scin-bf-other} of $\scin$: $u_i$ is not a light in-Steiner vertex.
    \item Step~\ref{step:scin-light-out} of $\scin$: Since $v$ is light and $v$ is not an in-Steiner vertex, by \cref{obs:light-vertex}, $v$ must be a light out-Steiner vertex. Then by \cref{lem:out-light-edge}, since $u\in \Nin(v)$, $u$ is not a light in-Steiner vertex.
\end{itemize}
\end{proof}

\begin{lemma} \label{lem:in-light-edge}
Let $v$ be a level-$t$ in-Steiner vertex that is classified as \emph{light}. By the end of the $t$-th iteration, we have $\Nout(v)=\{\pin(v)\}\cup \Sin_v$, where $\Sin_v$ consists of only level-$t$ N-Steiner vertices and $|\Sin_v|\leq \lambda$. Moreover, the following property holds throughout: For every $u\in \Nout(v)\setminus \{\pin(v)\}$, either $u\in \Sin_v$, or $u$ is an in-Steiner vertex with $\pin(u)\in \Sin_v$.
\end{lemma}
\begin{proof}
\cref{lem:in-heavy-edge} implies that by the end of the $t$-th iteration, $\Nout(v)$ only contains $\pin(v)$ and some level-$t$ N-Steiner vertices. We denote the latter by the set $\Sin_v$. Since $v$ is classified as light in Step~\ref{step:classify-in}, we have $|\Sin_v|\leq \lambda$.

In subsequent iterations, consider Steps~\ref{step:in-up}, \ref{step:in-down}, \ref{step:in-r}, \ref{step:in-tilde-r}, \ref{step:out-up}, \ref{step:out-down}, \ref{step:out-r}, \ref{step:out-tilde-r} of $\shortcut(G,t)$. The algorithm never adds outgoing edges from $v$ in Steps~\ref{step:in-up} and~\ref{step:out-down} as edges added in these steps are directed out from newly created vertices. The same holds for Steps~\ref{step:in-r}, \ref{step:in-tilde-r}, \ref{step:out-r}, \ref{step:out-tilde-r}, because the endpoints of the added edges are either negative vertices or newly created vertices. In addition, it also never adds outgoing edges from $v$ in Step~\ref{step:out-up}, since $v$ is light. The only remaining case is Step~\ref{step:in-down}. In this case, since $v$ is light, it can never be the out-parent of any vertex. Thus, the algorithm can only add an edge $(v,u)$ for some in-Steiner vertex $u$, given that the edge $(v,\pin(u))$ already exists in the graph. Assuming \cref{lem:in-light-edge} holds before adding this edge, then we must have $\pin(u)\in \Sin_v$ since an in-Steiner vertex cannot be the in-parent of any vertex by \cref{obs:steiner}~\ref{obs:pin}. Therefore, the lemma continues to hold after adding this edge. We conclude the proof by noting that the algorithm never adds outgoing edges from $v$ via procedure $\ein$ by \cref{clm:ein-light}, and it also never adds outgoing edges from $v$ in Step~\ref{step:r'}.
\end{proof}

The proofs of \cref{lem:out-heavy-edge}, \cref{clm:eout-light} and \cref{lem:out-light-edge} are largely symmetric to those of \cref{lem:in-heavy-edge}, \cref{clm:ein-light} and \cref{lem:in-light-edge}. We therefore defer them to Appendix~\ref{sec:appendix}.

\begin{lemma} \label{lem:out-heavy-edge}
Let $v$ be a level-$t$ out-Steiner vertex. By the end of the $t$-th iteration, $\Nin(v)$ only consists of $\pout(v)$ and some level-$t$ N-Steiner vertices. Moreover, the following properties hold throughout:
\begin{enumerate}[label=(\alph*)]
    \item \label{item:out-1} For every $u\in \Nout(v)$, there is an edge $(\pout(v),u)$ in $G$ of weight at most $w(\pout(v),v)+w(v,u)$.
    \item \label{item:out-2} For every $u\in \Nin(v)$, one of the following holds:
    \begin{itemize}
        \item $u=\pout(v)$.
        \item $u$ is an in-Steiner vertex with $\pin(u)=v$.
        \item $u\in \bar N$. (Before Step~\ref{step:r'}, temporary edges with $u\in N$ may exist, and the edge $(u,v)$ will be replaced by $(\bar u,v)$ in Step~\ref{step:r'}.)
        \item $u$ is an N-Steiner vertex.
        \item $u$ is an out-Steiner vertex with $\pout(u)\in \Nin(v)$ and $\pout(u)\in \bar N$.
        \item $u$ is an out-Steiner vertex and $\pout(u)\in \Nin(v)$ is an N-Steiner vertex.
    \end{itemize}
    \item \label{item:out-3} For every $u\in \Nin(v)$:
    \begin{itemize}
        \item If $u\in \bar N$, then there exists an edge $(u,\pout(v))$ of weight at most $w(u,v)-w(\pout(v),v)$.
        \item If $u$ is an N-Steiner vertex, then there is an edge $(u,\pout(v))$ of weight $w(u,v)-w(\pout(v),v)$ in $\Fout$, and there is an edge $(\pn(u),\pout(v))$ of weight at most $w(u,v)-w(u,\pn(u))-w(\pout(v),v)$ in $G$.
    \end{itemize}
\end{enumerate}
\end{lemma}

\begin{claim} \label{clm:eout-light}
For every call $\eout(u,v,W)$ made by the algorithm, the vertex $v$ is not a light out-Steiner vertex.
\end{claim}

\begin{lemma} \label{lem:out-light-edge}
Let $v$ be a level-$t$ out-Steiner vertex. By the end of the $t$-th iteration, we have $\Nin(v)=\{\pout(v)\}\cup \Sout_v$, where $\Sout_v$ consists of only level-$t$ N-Steiner vertices. Moreover, the following property holds throughout: For every $u\in \Nin(v)\setminus \{\pout(v)\}$, either $u\in \Sout_v$, or $u$ is an out-Steiner vertex with $\pout(u)\in \Sout_v$.
\end{lemma}

\subsubsection{Analysis of $\scin(r,t)$ and $\scout(r,t)$} 
\label{sec:scin}

We establish the following lemmas, which characterize the vertices that may ever be processed by the procedure.

\begin{lemma} \label{lem:plus-one-hop}
Let $r\in N$ be a negative vertex. For any in-Steiner vertex $v$, if there exists an $h$-hop $v\rightsquigarrow r$ path $P$ in $G$ of weight $W$, then there exists an $(h+1)$-hop path $\pin(v)\rightsquigarrow r$ in $G$ of weight at most $W-w(v,\pin(v))$.
\end{lemma}
\begin{proof}
Let $u$ be the successor of $v$ on the path $P$. Such a vertex must exist since $r$ is not an in-Steiner vertex. Since $v$ is an in-Steiner vertex and $u\in \Nout(v)$, by \cref{lem:in-heavy-edge}~\ref{item:in-2}, one of the following holds:
\begin{itemize}
    \item $u$ is an out-Steiner vertex with $\pout(u)=v$: Let $u'$ be the successor of $u$ on the path $P$. Such a vertex must exist since $r$ is not an out-Steiner vertex. Since $u'\in \Nout(u)$ and $\pout(u)=v$, by \cref{lem:out-heavy-edge}~\ref{item:out-1}, there is an edge $(v,u')$ of weight at most $w(v,u)+w(u,u')=w(P[v,u'])$. We note that $u'$ is not an out-Steiner vertex. Otherwise, since $u\in \Nin(u')$ and $\pout(u)=v$ is an in-Steiner vertex, this would contradict \cref{lem:out-heavy-edge}~\ref{item:out-2} applied to $u'$. Hence, we can replace $u$ by $u'$. After that, we still have $u\in \Nout(v)$ and $w(v,u)\leq w(P[v,u])$, but $u$ must now fall into one of the later cases of our analysis.
    \item $u=\pin(v)$: The desired path is $P[u,r]$. We have
    \[
    w(P[u,r])=w(P)-w(P[v,u])\leq w(P)-w(v,u)=w(P)-w(v,\pin(v)),
    \]
    and $P[u,r]$ clearly contains at most $h$ hops.
    \item $u\in N$: By \cref{lem:in-heavy-edge}~\ref{item:in-3}, there is an edge $(\pin(v),u)$ of weight at most $w(v,u)-w(v,\pin(v))$. Therefore, the desired path is $\pin(v)\to u\rightsquigarrow r$, where the $u\rightsquigarrow r$ segment is $P[u,r]$. The weight of the path is
    \[
        w(\pin(v),u)+w(P[u,r])\leq w(v,u)-w(v,\pin(v))+w(P)-w(v,u)\leq w(P)-w(v,\pin(v)),
    \]
    and it contains at most $h$ hops.
    \item $u$ is an N-Steiner vertex: By \cref{lem:in-heavy-edge}~\ref{item:in-3}, there is an edge $(\pin(v),\pn(u))$ of weight at most $w(v,u)-w(\pn(u),u)-w(v,\pin(v))$. The desired path is $\pin(v)\to \pn(u)\to u\rightsquigarrow r$, where the $u\rightsquigarrow r$ segment is $P[u,r]$. The weight of this path is
    \begin{align*}
    & \quad\; w(\pin(v),\pn(u))+w(\pn(u),u)+w(P[u,r])\\
    & \leq w(v,u)-w(\pn(u),u)-w(v,\pin(v))+w(\pn(u),u)+w(P)-w(v,u)\\
    & = w(P)-w(v,\pin(v)).
    \end{align*}
    Note that in the $\pin(v)\to \pn(u)\to u$ segment, $(\pn(u),u)$ is the only edge that can be negative. (Strictly speaking, the edge $(\pn(u),u)$ may actually correspond to a length-two path, see \cref{rmk:r'}. Nevertheless, such a length-two path contains at most one negative edge.) Therefore, this path consists of at most $h+1$ hops.
    \item $u$ is an in-Steiner vertex and $\pin(u)\in N$: Since $v\in \Nin(u)$, by \cref{lem:in-heavy-edge}~\ref{item:in-1} applied to $u$, there is an edge $(v,\pin(u))$ of weight at most $w(v,u)+w(u,\pin(u))$. Since $\pin(u)\in\Nout(v)$ and $\pin(u)\in N$, by \cref{lem:in-heavy-edge}~\ref{item:in-3} applied to $v$, there is an edge $(\pin(v),\pin(u))$ of weight at most
    \[
    w(v,\pin(u))-w(v,\pin(v))\leq w(v,u)+w(u,\pin(u))-w(v,\pin(v))
    \]
    Moreover, in Step~\ref{step:in-r} of the procedure $\shortcut$, it adds an edge $(\pin(u),u)$ of weight $-w(u,\pin(u))$. Therefore, the desired path is $\pin(v)\to \pin(u)\to u\rightsquigarrow r$, where the $u\rightsquigarrow r$ segment is $P[u,r]$. The weight of the path is
    \begin{align*}
    & \quad\; w(\pin(v),\pin(u))+w(\pin(u),u)+w(P[u,r])\\
    & \leq w(v,u)+w(u,\pin(u))-w(v,\pin(v))-w(u,\pin(u))+w(P)-w(v,u)\\
    & = w(P)-w(v,\pin(v)).
    \end{align*}
    Note that in the $\pin(v)\to \pin(u)\to u$ segment, $(\pin(u),u)$ is the only edge that can be negative. (Strictly speaking, the edge $(\pin(u),u)$ may actually correspond to a length-two path, see \cref{rmk:r'}. Nevertheless, such a length-two path contains at most one negative edge.) Therefore, this path consists of at most $h+1$ hops.
    \item $u$ is an in-Steiner vertex and $\pin(u)$ is an N-Steiner vertex: Since $v\in \Nin(u)$, by \cref{lem:in-heavy-edge}~\ref{item:in-1} applied to $u$, there is an edge $(v,\pin(u))$ of weight at most $w(v,u)+w(u,\pin(u))$. Since $\pin(u)\in\Nout(v)$ and $\pin(u)$ is an N-Steiner vertex, by \cref{lem:in-heavy-edge}~\ref{item:in-3} applied to $v$, there is an edge $(\pin(v),\pn(\pin(u)))$ of weight at most
    \begin{align*}
    & \quad\; w(v,\pin(u))-w(\pn(\pin(u)),\pin(u))-w(v,\pin(v))\\
    &\leq w(v,u)+w(u,\pin(u))-w(\pn(\pin(u)),\pin(u))-w(v,\pin(v))
    \end{align*}
    Moreover, in Step~\ref{step:in-tilde-r} of the procedure $\shortcut$, it adds an edge $(\pn(\pin(u)),u)$ of weight $w(\pn(\pin(u)),\pin(u))-w(u,\pin(u))$. Therefore, the desired path is $\pin(v)\to \pn(\pin(u))\to u\rightsquigarrow r$, where the $u\rightsquigarrow r$ segment is $P[u,r]$. The weight of the path is
    \begin{align*}
    & \quad\; w(\pin(v),\pn(\pin(u)))+w(\pn(\pin(u)),u)+w(P[u,r])\\
    & \leq w(v,u)+w(u,\pin(u))-w(\pn(\pin(u)),\pin(u))-w(v,\pin(v))\\
    & \quad\; +w(\pn(\pin(u)),\pin(u))-w(u,\pin(u))+w(P)-w(v,u)\\
    & = w(P)-w(v,\pin(v)).
    \end{align*}
    Note that in the $\pin(v)\to \pn(\pin(u))\to u$ segment, $(\pn(\pin(u)),u)$ is the only edge that can be negative. (Strictly speaking, the edge $(\pn(\pin(u)),u)$ may actually correspond to a length-two path, see \cref{rmk:r'}. Nevertheless, such a length-two path contains at most one negative edge.) Therefore, this path consists of at most $h+1$ hops.
\end{itemize}
Since the desired path can be found in all cases, the lemma follows.
\end{proof}

\begin{lemma} \label{lem:scin-dist-valid}
Throughout the algorithm $\scin(r,t)$, for any vertex $v\in V$, if $d_v<\infty$, then there exists an $L$-hop $v\rightsquigarrow r$ path in $G$ of weight at most $d_v$.
\end{lemma}
\begin{proof}
We prove the following claim, from which the lemma follows immediately.
\begin{claim}
For any level-$\ell$ vertex $v\in V$, if $d_v<\infty$, then there exists an $(t-\ell)$-hop $v\rightsquigarrow r$ path in $G$ of weight at most $d_v$.
\end{claim}
For every vertex $v\in \Uin_r$, we set $d_v=d^0(v,r)$ initially and the claim holds trivially. We now prove the claim inductively for distances updated in Steps~\ref{step:scin-in} and~\ref{step:scin-bf-in}.

Suppose we are at Step~\ref{step:scin-in}, and suppose $v$ is a level-$\ell$ vertex. By induction hypothesis, there exists a $(t-\ell)$-hop $v\rightsquigarrow r$ path in $G$ of weight at most $d_v$. By \cref{lem:plus-one-hop}, there exists a $(t-\ell+1)$-hop $\pin(v)\rightsquigarrow r$ path of weight at most $d_v-w(v,\pin(v))$. Moreover, by \cref{obs:steiner}~\ref{obs:lower-level-in}, $\pin(v)$ is a level-$\ell'$ vertex for some $\ell' < \ell$. Therefore, if we update $d_{\pin(v)}$ to $d_v-w(v,\pin(v))$ in Step~\ref{step:scin-in}, there must exist a $(t-\ell+1)\leq (t-\ell')$-hop $\pin(v)\rightsquigarrow r$ path in $G$ with weight at most $d_{\pin(v)}$.

Suppose we are at Step~\ref{step:scin-bf-in}, and suppose $v$ is a level-$\ell$ vertex, $u_i$ is a level-$\ell'$ vertex. By induction hypothesis, there exists a $(t-\ell)$-hop $v\rightsquigarrow r$ path in $G$ of weight at most $d_v$. This implies that there is a $(t-\ell)$-hop $u_i\rightsquigarrow r$ path of weight at most $d_v+w(u_i,v)$. By \cref{lem:plus-one-hop}, there exists a $(t-\ell+1)$-hop $\pin(u_i)\rightsquigarrow r$ path of weight at most $d_v+w(u_i,v)-w(u_i,\pin(u_i))$. On the other hand, since $u_i$ is a light in-Steiner vertex, by \cref{lem:in-light-edge}, $v$ must fall into one of the following categories:
\begin{itemize}
    \item $v=\pin(u_i)$: This is not possible since we have $u_i\in \Ain_v$ at Step~\ref{step:scin-bf-in}. If $v=\pin(u_i)$, we would have deleted $u_i$ from $\Ain_v$ in Step~\ref{step:delete} of $\shortcut(G,t)$.
    \item $v$ is a level-$\ell'$ N-Steiner vertex: Then we have $\ell'=\ell$. By \cref{obs:steiner}~\ref{obs:lower-level-in}, $\pin(u_i)$ is a level-$\ell''$ vertex for some $\ell''<\ell'=\ell$.
    \item $v$ is an in-Steiner vertex where $\pin(v)$ is a level-$\ell'$ vertex: This is also not possible since at Step~\ref{step:scin-bf-in}, $v$ is not an in-Steiner vertex.
\end{itemize}
Therefore, we must have $\pin(u_i)$ is a level-$\ell''$ vertex for some $\ell'' < \ell$. Therefore, if we update $d_{\pin(u_i)}$ to $d_v+w(u_i,v)-w(u_i,\pin(u_i))$ in Step~\ref{step:scin-bf-in}, then there must be a $(t-\ell+1)\leq (t-\ell'')$-hop $\pin(u_i)\rightsquigarrow r$ path in $G$ with weight at most $d_{\pin(u_i)}$.
\end{proof}

\begin{claim} \label{clm:scin-dist}
During the algorithm $\scin(r,t)$, for any vertex $v\in V$, if $d_v<\infty$, then $d_v<\Delta_r$.
\end{claim}
\begin{proof}
For every $v\in \Uin_r$, \cref{lem:compute-U} implies that $d^0(v,r)<\Delta_r$. We inductively show that the bound holds for distances updated in Steps~\ref{step:scin-in} and~\ref{step:scin-bf-in}. In Step~\ref{step:scin-in}, since $d_v<\Delta_r$ by induction hypothesis and $w(v,\pin(v))\geq 0$, we update $d_{\pin(v)}$ only with the value $d_v - w(v,\pin(v)) < \Delta_r$. In Step~\ref{step:scin-bf-in}, we update $d_{\pin(u_i)}$ with a value $W$ only when $W < \Delta_r$. Therefore, the claim always holds.
\end{proof}

\begin{lemma} \label{lem:scin-time}
The running time of $\scin(r,t)$ is $\tilde O(|\Vin_r|^2+
\lambda|\Vin_r|)$.
\end{lemma}
\begin{proof}
Combining \cref{lem:scin-dist-valid} and \cref{clm:scin-dist}, we see that for every vertex $v\in V$ with $d_v<\infty$, $d^L(v,r)<\Delta_r$ holds. Therefore, for every vertex $v$ that has been processed during $\scin(r,t)$, $v\in \Vin_r$ must hold. We first note that each vertex is only processed once. This follows from the proof of \cref{lem:scin-dist-valid}: whenever a level-$\ell$ vertex $v$ is processed, and a level-$\ell'$ vertex $v'$ is added in Step~\ref{step:scin-in} or Step~\ref{step:scin-bf-in}, $\ell'<\ell$ must hold.

Now we show that the time required for processing a single vertex $v$ is $\tilde O(|\Vin_r|+\lambda L)$. By \cref{clm:addedge-time}, the running time of $\ein$ is $\tilde O(1)$. We focus on the non-trivial steps:
\begin{itemize}
    \item Step~\ref{step:scin-bf}: The running time of this step is $O(2^{\ell})=O(k-j+1)$. We now show that $k-j+1\leq |\Vin_r|$. This is because for every $j\leq i\leq k$, $w(u_i,v)+d_v-\Delta_r<0$ holds. Since \cref{lem:scin-dist-valid} implies $d_v\geq d^L(v,r)$, by triangle inequality, we have
    \[
    d^L(u_i,r)\leq w(u_i,v)+d^L(v,r)\leq w(u_i,v)+d_v<\Delta_r.
    \]
    Therefore, $u_i\in \Vin_r$ for every $j\leq i\leq k$, which implies that $k-j+1\leq |\Vin_r|$.
    \item Step~\ref{step:scin-light-out}: The running time of this step is $O(|\Ain_v|)=O(|\Nin(v)|)$. Since $v$ is not an in-Steiner vertex and $v$ is light, by \cref{obs:light-vertex}, $v$ must be a light out-Steiner vertex. Then by \cref{lem:out-light-edge}, for every $u\in \Nin(v)\setminus \{\pout(v)\}$, either $u\in \Sout_v$ or $u$ is an out-Steiner vertex with $\pout(u)\in \Sout_v$. Moreover, we have $|\Sout_v|\leq \lambda$. For every vertex $x\in \Sout_v$, each iteration of the $\shortcut$ algorithm may create at most $O(\log n)$ out-Steiner vertices $y$ with $\pout(y)=x$. Note that the number of vertices may increase, \cref{sec:runtime} shows that it is always bounded by $\tilde O(n)$. Since there are $L$ iterations in total, the number of such out-Steiner vertices $y$ is bounded by $O(L\log n)$. Therefore,
    \[
    \Nin(v)=O(|\Sout_v|L\log n)=O(\lambda L\log n)=\tilde O(\lambda).
    \]
\end{itemize}
\end{proof}

An analogous analysis applies to $\scout$, yielding the same running time bound. Since the argument is largely symmetric, we defer the proof to Appendix~\ref{sec:appendix}.

\begin{lemma} \label{lem:scout-time}
The running time of $\scout(r,t)$ is $\tilde O(|\Vout_r|^2+
\lambda|\Vout_r|)$.
\end{lemma}

\subsubsection{Putting it Together} \label{sec:overall}
We now combine the ingredients from \cref{sec:structural_invariants,sec:scin} to establish the key global consequences of shortcutting: controlled growth in the number of vertices, distance preservation, and hop reduction.
\begin{claim} \label{clm:heavy-increase}
In a single iteration of the algorithm $\shortcut(G,t)$, the number of heavy vertices increases by $O(\eta \log n)$.
\end{claim}
\begin{proof}
Since Steiner vertices created in earlier iterations cannot become heavy, we consider only the newly created vertices.
We first focus on the newly created in-Steiner vertices. Let $S$ be the set of all level-$t$ in-Steiner vertices. For any $v\in S$, by \cref{lem:in-heavy-edge}, $\Nout(v)$ only consists of $\pin(v)$ and some level-$t$ N-Steiner vertices by the end of $\shortcut(G,t)$. For any negative vertex $r$, if the edge $(v,\tilde r_t)$ is added, then it must be added in Step~\ref{step:scin-steiner} of $\scin(r,t)$  when its parent vertex $\pin(v)$ is processed. Since $d_{\pin(v)}<\infty$, by \cref{lem:scin-dist-valid} and \cref{clm:scin-dist}, we have that $d^L(\pin(v),r)<\Delta_r$ and thus $\pin(v)\in \Vin_r$. Therefore, we have that
\[
|\Nout(v)|\leq 1+\sum_{r\in N} \textbf{1}\left[\pin(v)\in \Vin_r\right].
\]
For each $v\in S$, we classify it as heavy if $|\Nout(v)|-1\geq \lambda$. Summing over all $v\in S$, we obtain
\begin{align*}
\sum_{v\in S}\left(|\Nout(v)|-1\right)&\leq \sum_{v\in S}\sum_{r\in N}\textbf{1}\left[\pin(v)\in \Vin_r\right]\\
& \leq \sum_{r\in N}\sum_{u\in \Vin_r} |\{v\in S:\pin(v)=u\}|.
\end{align*}
For every $u\in \Vin_r$, if $u$ is heavy, then in Step~\ref{step:steiner-vertex-in}, the algorithm may create at most $O(\log n)$ in-Steiner vertices $u'$ with $\pin(u')=u$. If $u$ is light, then it can never be the in-parent of any vertex. Therefore,
\[
\sum_{v\in S}\left(|\Nout(v)|-1\right) \leq \sum_{r\in N} |\Vin_r|\cdot O(\log n) =  O\left(\frac{\eta n\log n}{\sqrt{b}}\right)
\]
by \cref{lem:compute-delta}. This implies that the number of vertices $v\in S$ with $|\Nout(v)|-1\geq \lambda$ is at most
\[
O\left(\frac{\eta n \log n}{\lambda\sqrt{b}}\right)=O(\eta \log n),
\]
since $\lambda=n/\sqrt{b}$.

For the newly created out-Steiner vertices, let $S$ denote the set of all level-$t$ out-Steiner vertices. By the same argument, one can show that
\[
\sum_{v\in S}\left(|\Nin(v)|-1\right) \leq \sum_{r\in N} |\Vout_r|\cdot \log n =  O\left(\frac{\eta n\log n}{\sqrt{b}}\right)
\]
Therefore, the number of vertices $v\in S$ with $|\Nin(v)|-1\geq \lambda$ is also bounded by $O(\eta\log n)$.
\end{proof}

\begin{lemma}
For any pair of negative vertices $u,v\in N$, after shortcutting, their distance $d(u,v)$ does not decrease.
\end{lemma}
\begin{proof}
For every vertex $v\in V$, we define its ancestor vertex $a(v)$ and potential $\psi(v)$ as follows:
\begin{itemize}
    \item If $v\in N$, let $a(v)=v$ and $\psi(v)=0$.
    \item If $v\in \bar N$, let $(r,v)$ be the negative edge incident to $v$. Let $a(v)=r$ and $\psi(v)=-w(r,v)$. 
    \item If $v$ is an in-Steiner vertex, let $a(v)=a(\pin(v))$ and $\psi(v)=w(v,\pin(v))+\psi(\pin(v))$.
    \item If $v$ is an out-Steiner vertex, let $a(v)=a(\pout(v))$ and $\psi(v)=-w(\pout(v),v)+\psi(\pout(v))$.
    \item If $v$ is an N-Steiner vertex, let $a(v)=\pn(v)$ and $\psi(v)=\psi(\pn(v))+w(v,\pn(v))=w(v,\pn(v))$.
\end{itemize}
After applying vertex potential $\varphi$ to the whole graph, it is straightforward to verify that for every vertex $v\in V$, the potential $\psi(v)$ is updated to $\psi(v)+\varphi(v)-\varphi(a(v))$.

For any pair of vertices $u,v\in N$, let $d_0(u,v)$ denote the distance from $u$ to $v$ in $G$ when the path is restricted to edges in $E_0$ (i.e., the edges of the original graph $G_0$). We show the following invariant is maintained throughout the algorithm.

\begin{invariant} \label{inv:dist}
For every $u,v\in V$, $w(u,v)\geq d_0(a(u),a(v))+\psi(u)-\psi(v)$. In addition, for every $e=(u,v)\in \Fin\cup \Fout$, $w(e)\geq d_0(a(u),a(v))+\psi(u)-\psi(v
)$.
\end{invariant}

We start by deriving the lemma from this invariant. Consider any path $P=(v_1,\dots,v_k)$ in graph $G$. The weight of the path is
\begin{align*}
w(P)&=\sum_{i=1}^{k-1} w(v_i,v_{i+1})\\
& \geq \sum_{i=1}^{k-1}\Big(d_0(a(v_i),a(v_{i+1}))+\psi(v_i)-\psi(v_{i+1})\Big) && \text{(\cref{inv:dist})}\\
& = \psi(v_1)-\psi(v_k)+\sum_{i=1}^{k-1} d_0(a(v_i),a(v_{i+1}))\\
& \geq \psi(v_1)-\psi(v_k)+d_0(a(v_1),a(v_k)) && \text{(triangle inequality)}
\end{align*}
If $v_1,v_k\in N$, we have $\psi(v_1)=\psi(v_k)=0$ and $a(v_1)=v_1,a(v_k)=v_k$, and thus $w(P)\geq d_0(v_1,v_k)$. This implies that any path between two negative vertices has weight at least their distance when restricted to the edges of the original graph $G_0$, which establishes the lemma.

We now assume that \cref{inv:dist} holds at the beginning of the algorithm $\shortcut(G,t)$ and show that it continues to hold after the shortcut edges are added. By the argument in the preceding paragraph, at the start of $\shortcut(G,t)$ we have
\begin{equation} \label{eq:dist}
d(u,v)\geq d_0(a(u),a(v))+\psi(u)-\psi(v) \text{\quad for every } u,v\in V.
\end{equation}

Note that it suffices to show that \cref{inv:dist} is satisfied when the edge $(u,v)$ is added to the graph. It is easy to see that the invariant continues to hold after applying any vertex potential. We first focus on edges added in $\shortcut(G,t)$:
\begin{itemize}
    \item Step~\ref{step:negative-left}: By \cref{eq:dist},
    \[
    d_0(a(u),a(\bar r))+\psi(u)-\psi(\bar r)\leq d(u,\bar r)\leq w(u,\bar u)+d^0(\bar u,r)+w(r,\bar r).
    \]
    \item Step~\ref{step:negative-right}: By \cref{eq:dist},
    \[
    d_0(a(r),a(\bar v))+\psi(r)-\psi(\bar v)\leq d(r,\bar v)\leq d^1(r,v)+w(v,\bar v).
    \]
    \item Steps~\ref{step:in-up},~\ref{step:in-down}: Since $\pin(v_\ell)=v$, we have $a(v_\ell)=a(v)$ and $\psi(v_\ell)=\psi(v)+w(v_\ell,v)$. Thus,
    \[
    d_0(a(v_\ell),a(v))+\psi(v_\ell)-\psi(v)=\psi(v_\ell)-\psi(v)=w(v_\ell,v).
    \]
    For every $1\leq i\leq k-2^\ell$, the algorithm adds an edge $(u_i,v_\ell)$ of weight $w(u_i,v)-w(v_\ell,v)$. \cref{inv:dist} applied to the edge $(u_i,v)$ implies that $w(u_i,v)\geq d_0(a(u_i),a(v))+\psi(u_i)-\psi(v)$. Therefore,
    \begin{align*}
    &\quad\; d_0(a(u_i),a(v_\ell))+\psi(u_i)-\psi(v_\ell)\\
    & = d_0(a(u_i),a(v))+\psi(u_i)-\psi(v)-w(v_\ell,v)\\
    & \leq w(u_i,v)-w(v_\ell,v).
    \end{align*}
    \item Step~\ref{step:in-r}: Since $\pin(v_\ell)=v$, we have $a(v_\ell)=a(v)$ and $\psi(v_\ell)=\psi(v)+w(v_\ell,v)$. Thus,
    \[
    d_0(a(v),a(v_\ell))+\psi(v)-\psi(v_\ell)=\psi(v)-\psi(v_\ell)=-w(v_\ell,v).
    \]
    \item Step~\ref{step:in-tilde-r}: Since $\pin(v_\ell)=v$, we have $a(v_\ell)=a(v)$ and $\psi(v_\ell)=\psi(v)+w(v_\ell,v)$. Since $v$ is a N-Steiner vertex, we have $a(v)=\pn(v)$ and $\psi(v)=\psi(\pn(v))+w(v,\pn(v))=\psi(\pn(v))-w(\pn(v),v)$. Thus
    \begin{align*}
    &\quad\;d_0(a(\pn(v)),a(v_\ell))+\psi(\pn(v))-\psi(v_\ell)\\
    & = \psi(\pn(v))-\psi(v_\ell)\\
    & = \psi(\pn(v))-\psi(v)-w(v_\ell,v)\\
    & = w(\pn(v),v)-w(v_\ell,v).
    \end{align*}
    \item Steps~\ref{step:out-up},~\ref{step:out-down}: Since $\pout(v_\ell)=v$, we have $a(v_\ell)=a(v)$ and $\psi(v_\ell)=\psi(v)-w(v,v_\ell)$. Thus,
    \[
    d_0(a(v),a(v_\ell))+\psi(v)-\psi(v_\ell)=\psi(v)-\psi(v_\ell)=w(v,v_\ell).
    \]
    For every $1\leq i\leq k-2^\ell$, the algorithm adds an edge $(v_\ell,u_i)$ of weight $w(v,u_i)-w(v,v_\ell)$. \cref{inv:dist} applied to the edge $(v,u_i)$ implies that $w(v,u_i)\geq d_0(a(v),a(u_i))+\psi(v)-\psi(u_i)$. Therefore,
    \begin{align*}
    &\quad\;d_0(a(v_\ell),a(u_i))+\psi(v_\ell)-\psi(u_i)\\
    & = d_0(a(v),a(u_i))+\psi(v)-w(v,v_\ell)-\psi(u_i)\\
    & \leq w(v,u_i)-w(v,v_\ell).
    \end{align*}
    \item Step~\ref{step:out-r}: Since $\pout(v_\ell)=v$, we have $a(v_\ell)=a(v)$ and $\psi(v_\ell)=\psi(v)-w(v,v_\ell)$. Since $v\in \bar N,r\in N$, we have $a(v)=r, a(r)=r$ and $\psi(v)=\psi(r)-w(r,v)$. Thus,
    \begin{align*}
    &\quad\;d_0(a(v_\ell),a(r))+\psi(v_\ell)-\psi(r)\\
    & = \psi(v_\ell)-\psi(r)\\
    & = \psi(v)-w(v,v_\ell) -\psi(r)\\
    & = -w(r,v)-w(v,v_\ell).
    \end{align*}
    \item Step~\ref{step:out-tilde-r}: Since $\pout(v_\ell)=v$, we have $a(v_\ell)=a(v)$ and $\psi(v_\ell)=\psi(v)-w(v,v_\ell)$. Since $v$ is a N-Steiner vertex, we have $a(v)=\pn(v)$ and $\psi(v)=\psi(\pn(v))+w(v,\pn(v))$. Thus
    \begin{align*}
    &\quad\;d_0(a(v_\ell),a(\pn(v)))+\psi(v_\ell)-\psi(\pn(v))\\
    & = \psi(v_\ell)-\psi(\pn(v))\\
    & = \psi(v)-w(v,v_\ell)-\psi(\pn(v))\\
    & = w(v,\pn(v))-w(v,v_\ell).
    \end{align*}
    \item Step~\ref{step:tilde-r-up}: This case is immediate since $a(\tilde r_t)=r,\psi(r)=0$, and $\psi(\tilde r_t)=w(\tilde r_t,r)=-w(r,\tilde r_t)$.
\end{itemize}
Next, we focus on the edges added via the procedure $\ein$. We start by showing the following claim.
\begin{claim} \label{clm:ein-valid}
For any call to the procedure $\ein(u,v,W)$ with $W\geq d_0(a(u),a(v))+\psi(u)-\psi(v)$, any edge added by the procedure (including those added during recursions) satisfy \cref{inv:dist}.
\end{claim}
\begin{proof}
The algorithm first adds an edge $(u,v)$ of weight $W$, which satisfies \cref{inv:dist}. We now consider the recursive calls that the procedure may make.
\begin{itemize}
    \item If $u$ is an in-Steiner vertex, then $a(u)=a(\pin(u))$ and $\psi(u)=\psi(\pin(u))+w(u,\pin(u))$.
    \begin{itemize}
        \item If $v\in N$, the procedure recursively calls $\ein(\pin(u),v,W-w(u,\pin(u)))$. We show that $W-w(u,\pin(u))\geq d_0(a(\pin(u)),a(v))+\psi(\pin(u))-\psi(v)$. Indeed,
        \begin{align*}
            &\quad\;d_0(a(\pin(u)),a(v))+\psi(\pin(u))-\psi(v)\\
            & = d_0(a(u),a(v))+\psi(u)-w(u,\pin(u)).-\psi(v)\\
            & \leq W-w(u,\pin(u)).
        \end{align*}
        \item If $v$ is an N-Steiner vertex, then $a(v)=\pn(v)$ and $\psi(v)=\psi(\pn(v))+w(v,\pn(v))=\psi(\pn(v))-w(\pn(v),v)$. The procedure adds an edge $(\pin(u),v)$ of weight $W-w(u,\pin(u))$ to $\Fin$, and it may recursively call $\ein(\pin(u),v,W-w(u,\pin(u)))$ and \\$\ein(\pin(u),\pn(v),W-w(\pn(v),v)-w(u,\pin(u)))$. For the same reason as above, we have $W-w(u,\pin(u))\geq d_0(a(\pin(u)),a(v))+\psi(\pin(u))-\psi(v)$. Combined with the fact that $\psi(v)=\psi(\pn(v))-w(\pn(v),v)$, we obtain
        \[
        W-w(\pn(v),v)-w(u,\pin(u))\geq d_0(a(\pin(u)),a(\pn(v)))+\psi(\pin(u))-\psi(\pn(v)).
        \]
    \end{itemize}
    \item If $u$ is an out-Steiner vertex, then $a(u)=a(\pout(u))$ and $\psi(u)=\psi(\pout(u))-w(\pout(u),u)$. The procedure recursively calls $\ein(\pout(u),v,W+w(\pout(u),u))$. We show that $W+w(\pout(u),u)\geq d_0(a(\pout(u)),a(v))+\psi(\pout(u))-\psi(v)$. Indeed,
    \begin{align*}
        &\quad\;d_0(a(\pout(u)),a(v))+\psi(\pout(u))-\psi(v)\\
        & = d_0(a(u),a(v)) + \psi(u) + w(\pout(u),u) - \psi(v)\\
        & \leq W+w(\pout(u),u).
    \end{align*}
\end{itemize}
\end{proof}

By a similar argument, one can show the following analogous claim for $\eout$.
\begin{claim} \label{clm:eout-valid}
For any call to the procedure $\eout(u,v,W)$ with $W\geq d_0(a(u),a(v))+\psi(u)-\psi(v)$, any edge added by the procedure (including those added during recursions) satisfy \cref{inv:dist}. 
\end{claim}

Now we show that in Step~\ref{step:shortcut} of $\shortcut(G,t)$, initial calls to $\ein(u,v,W)$ (excluding recursive calls) satisfies $W\geq d_0(a(u),a(v))+\psi(u)-\psi(v)$.
\begin{itemize}
    \item If $\Delta_r\leq 0$, for $v\in \Ain_r$, $\ein(v,\tilde r_t,w(v,r)-\Delta_r)$ is invoked. By \cref{eq:dist}, we have
    \[
    w(v,r)\geq d(v,r)\geq d_0(a(v),a(r))+\psi(v)-\psi(r) = d_0(a(v),r)+\psi(v).
    \]
    Since $a(\tilde r_t)=r$, $\psi(\tilde r_t)=w(\tilde r_t,r)=\Delta_r$, we obtain
    \[
    w(v,r)-\Delta_r \geq d_0(a(v),r)+\psi(v)-\Delta_r=d_0(a(v),a(\tilde r_t))+\psi(v)-\psi(\tilde r_t).
    \]
    \item If $\Delta_r\geq -w(r,\bar r)$, for $v\in \Aout_{\bar r}$, $\eout(\tilde r_t,v,w(r,\bar r)+\Delta_r+w(\bar r,v))$ is invoked. By \cref{eq:dist}, we have
    \[
    w(r,\bar r)+w(\bar r,v)\geq d(r,v)\geq d_0(a(r),a(v))+\psi(r)-\psi(v) = d_0(r,a(v))-\psi(v).
    \]
    Since $a(\tilde r_t)=r$, $\psi(\tilde r_t)=w(\tilde r_t,t)=\Delta_r$, we obtain
    \[
    w(r,\bar r)+\Delta_r+w(\bar r,v) \geq d_0(r,a(v))-\psi(v)+\Delta_r = d_0(a(\tilde r_t),a(v))+\psi(\tilde r_t)-\psi(v).
    \]
\end{itemize}

Next we show that during $\scin(r,t)$, initial calls to $\ein(u,v,W)$ (excluding recursive calls) satisfies $W\geq d_0(a(u),a(v))+\psi(u)-\psi(v)$. Combined with \cref{clm:ein-valid}, this implies that all edges added during $\scin(r,t)$ satisfies \cref{inv:dist}. Note that we have $a(\tilde r_t)=r$ and $\psi(\tilde r_t)=w(\tilde r_t,r)=\Delta_r$. When a vertex $v$ is processed during $\scin(r,t)$, \cref{lem:scin-dist-valid} implies that $d(v,r)\leq d_v$. By \cref{eq:dist}, we have
\[
d_v\geq d(v,r)\geq d_0(a(v),a(r))+\psi(v)-\psi(r)= d_0(a(v),a(r))+\psi(v).
\]
The procedure $\ein$ is invoked in the following steps.
\begin{itemize}
    \item Step~\ref{step:scin-out}: $\ein(\pout(v),\tilde r_t,w(\pout(v),v)+d_v-\Delta_r)$ may be invoked. We have $a(v)=a(\pout(v))$ and $\psi(v)=\psi(\pout(v))-w(\pout(v),v)$. Thus,
    \begin{align*}
        &\quad\;w(\pout(v),v)+d_v-\Delta_r\\
        &\geq w(\pout(v),v)+d_0(a(v),a(r))+\psi(v)-\Delta_r\\
        & = w(\pout(v),v)+d_0(a(\pout(v)),a(\tilde r_t))+\psi(\pout(v))-w(\pout(v),v)-\psi(\tilde r_t)\\
        & = d_0(a(\pout(v)),a(\tilde r_t))+\psi(\pout(v))-\psi(\tilde r_t)
    \end{align*}
    \item Step~\ref{step:scin-steiner}: $\ein(v_\ell,\tilde r_t,w(v_\ell,v)+d_v-\Delta_r)$ may be invoked. We have $a(v_\ell)=a(v)$ and $\psi(v_\ell)=\psi(v)+w(v_\ell,v)$. Thus,
    \begin{align*}
        &\quad\;w(v_\ell,v)+d_v-\Delta_r\\
        &\geq w(v_\ell,v)+d_0(a(v),a(r))+\psi(v)-\Delta_r\\
        & = w(v_\ell,v)+d_0(a(v_\ell),a(\tilde r_t))+\psi(v_\ell)-w(v_\ell,v)-\psi(\tilde r_t)\\
        & = d_0(a(v_\ell),a(\tilde r_t))+\psi(v_\ell)-\psi(\tilde r_t)
    \end{align*}
    \item Step~\ref{step:scin-bf-other}: $\ein(u_i,\tilde r_t,w(u_i,v)+d_v-\Delta_r)$ may be invoked. \cref{inv:dist} applied to the edge $(u_i,v)$ implies that $w(u_i,v)\geq d_0(a(u_i),a(v))+\psi(u_i)-\psi(v)$. Thus,
    \begin{align*}
        &\quad\;w(u_i,v)+d_v-\Delta_r\\
        &\geq d_0(a(u_i),a(v))+\psi(u_i)-\psi(v)+d_0(a(v),a(r))+\psi(v)-\Delta_r\\
        & = d_0(a(u_i),a(v))+d_0(a(v),a(\tilde r_t))+\psi(u_i)-\psi(\tilde r_t)\\
        & \geq d_0(a(u_i),a(\tilde r_t))+\psi(u_i)-\psi(\tilde r_t).
    \end{align*}
    \item Step~\ref{step:scin-light-out}: $\ein(u,\tilde r_t,w(u,v)+d_v-\Delta_r)$ may be invoked. \cref{inv:dist} applied to the edge $(u,v)$ implies that $w(u,v)\geq d_0(a(u),a(v))+\psi(u)-\psi(v)$. Thus,
    \begin{align*}
        &\quad\;w(u,v)+d_v-\Delta_r\\
        &\geq d_0(a(u),a(v))+\psi(u)-\psi(v)+d_0(a(v),a(r))+\psi(v)-\Delta_r\\
        & = d_0(a(u),a(v))+d_0(a(v),a(\tilde r_t))+\psi(u)-\psi(\tilde r_t)\\
        & \geq d_0(a(u),a(\tilde r_t))+\psi(u)-\psi(\tilde r_t).
    \end{align*}
\end{itemize}

The analysis for $\scout$ is analogous.

Finally, we see that Step~\ref{step:r'} of $\shortcut(G,t)$ preserves \cref{inv:dist}. This step can be viewed as follows: for each negative edge $(r,\bar r)$, we first redirect every outgoing edge from $r$ so that it instead originates from $\bar r$, then apply a vertex potential to $r$ and $\bar r$ ($\varphi(r)=\win_r,\varphi(\bar r)=W-\wout_r$). It is straightforward to verify that both these steps preserves \cref{inv:dist}. This completes the proof.
\end{proof}

\begin{lemma} \label{lem:shortcut-main}
For any pair of vertices $\alpha,\beta\in V_0$, if there exists an $h$-hop $\alpha\rightsquigarrow\beta$ shortest path $P$ in $G$, then after shortcutting, there is an $\alpha\rightsquigarrow \beta$ path $P'$ with at most $h-\lfloor h/3\rfloor$ hops and weight $w(P')\leq w(P)$.
\end{lemma}
\begin{proof}
Note that since $P$ is a shortest path, every subpath of $P$ is also a shortest path between its endpoints. Consider three successive negative edges $(u,\bar u),(r,\bar r),(v,\bar v)$ on $P$. There are three cases as discussed below. We show that in each case, the number of hops is reduced by one for every three hops, thereby establishing the lemma.

\paragraph{Case 1: $d^0(\bar u,r)<\Delta_r$.} Then we have $\bar u\in \Uin_r$ by \cref{lem:compute-U}. In Step \ref{step:negative-left}, we have added an edge $(u,\bar r)$ of weight $w(u,\bar u)+d^0(\bar u,r)+w(r,\bar r)$, which is at most $w(P[u,\bar r])$. Since $P[u,\bar r]$ consists of 2 hops, whereas the shortcut edge $(u,\bar r)$ only uses 1 hop, the number of hops is reduced by 1.

\paragraph{Case 2: $d^1(r,v)<-\Delta_r$.} Then we have $v\in \Uout_r$ by \cref{lem:compute-U}. In Step \ref{step:negative-right}, we have added an edge $(r,\bar v)$ of weight $d^1(r,v)+w(v,\bar v)$, which is at most $w(P[r,\bar v])$. Since $P[r,\bar v]$ consists of 2 hops, whereas the shortcut edge $(r,\bar v)$ only uses 1 hop, the number of hops is reduced by 1.

\paragraph{Case 3: $d^0(\bar u,r)\geq \Delta_r$ and $d^1(r,v)\geq -\Delta_r$.}

We will prove the following claims.

\begin{claim} \label{clm:shortcut-in}
After shortcutting, there is a $0$-hop $\bar u\rightsquigarrow\tilde r_t$ path of weight $d(\bar u,r)-\Delta_r$.
\end{claim}

\begin{claim} \label{clm:shortcut-out}
After shortcutting, there is a $0$-hop $\tilde r_t\rightsquigarrow v$ path of weight $d(r,v)+\Delta_r$.
\end{claim}

The lemma follows immediately from the two claims: there exists a $\bar u\rightsquigarrow\tilde r_t\rightsquigarrow v$ path whose total weight is $d(\bar u,r)-\Delta_r+d(r,v)+\Delta_r=d(\bar u,v)$. Moreover, since $P[\bar u,v]$ consists of 1 hop and this path contains 0 hops, the number of hops is reduced by 1.

We prove \cref{clm:shortcut-in} first. First consider the case $\Delta_r\leq 0$. Let $x$ be the predecessor of $r$ on the path $P[\bar u,r]$. There are two cases:
\begin{itemize}
    \item $x\notin \Ain_r$: By \cref{obs:list}~\ref{obs:in-list}, it must be the case that $x$ is an in-Steiner vertex with $\pin(x)=r$. In this case, let $y$ be the predecessor of $x$ on the path $P[\bar u,r]$. Such a vertex must exist since $\bar u$ is not an in-Steiner vertex. Then by \cref{lem:in-heavy-edge}~\ref{item:in-1}, there is an edge $(y,r)$ of weight at most $w(y,x)+w(x,r)$. Moreover, by \cref{obs:steiner}~\ref{obs:pin}, $y$ is not an in-Steiner vertex. Since $y\in \Nin(r)$ and $y$ is not an in-Steiner vertex, we must have $y\in\Ain_r$. Thus, we have added the edge $(y,\tilde r_t)$ of weight $w(y,r)-\Delta_r\leq w(P[y,r])-\Delta_r$ in Step~\ref{step:shortcut} of $\shortcut(G,t)$. Consider the path $\bar u \rightsquigarrow y \to \tilde r_t$, where the segment $\bar u \rightsquigarrow y$ is $P[\bar u,y]$. The weight of this path is at most $d(\bar u,r) - \Delta_r$.
    \item $x\in \Ain_r$: We have added the edge $(x,\tilde r_t)$ of weight $w(x,r)-\Delta_r$ in Step~\ref{step:shortcut} of $\shortcut(G,t)$. Consider the path $\bar u \rightsquigarrow x\to  \tilde r_t$, where the segment $\bar u \rightsquigarrow x$ is $P[\bar u,x]$. The weight of this path is exactly $d(\bar u,r) - \Delta_r$.
\end{itemize}

This completes the proof for the case $\Delta_r\leq 0$. In the remainder, we assume $\Delta_r>0$.

Consider the following recursive algorithm $\findpathin(x,y,Q)$. The input consists vertices $x,y\in V$ and a $\bar u\rightsquigarrow r$ path $Q$ that satisfy the following assumptions:
\begin{assumption} \label{asmp:scin-x}
$x\in \Nin(y)$ and $d(x,r)=w(x,y)+d(y,r)\geq \Delta_r$. In addition, $x$ lies on the path $Q$.
\end{assumption}
\begin{assumption} \label{asmp:scin-y}
During the execution of $\scin(r,t)$, vertex $y$ is processed, and at the moment it is processed, we have $d_y=d(y,r)<\Delta_r$.
\end{assumption}
\begin{assumption} \label{asmp:scin-Q}
$w(Q)=d(\bar u,r)$, and $Q[x,r]$ contains $0$ hops.
\end{assumption}
The output of the algorithm is a $0$-hop $\bar u\rightsquigarrow \tilde r_t$ path of weight $d(\bar u,r)-\Delta_r$.

Now we describe the algorithm $\findpathin(x,y,Q)$.
\begin{itemize}
    \item {\bf(Case R1)} If $y$ is an in-Steiner vertex: Since $y$ is an in-Steiner vertex and $x\in \Nin(y)$, by \cref{lem:in-heavy-edge}~\ref{item:in-1}, there is an edge $(x,\pin(y))$ of weight at most $w(x,y)+w(y,\pin(y))$. Then we can recursively invoke $\findpathin(x,\pin(y),Q)$ and return the resulting path. \cref{asmp:scin-y} is satisfied since when vertex $y$ is processed in $\scin(r,t)$, in Step~\ref{step:scin-in}, the algorithm must have updated $d_{\pin(y)}$ with $d_y-w(y,\pin(y))$ and thus $\pin(y)$ is also processed. In addition, by triangle inequality,
    \[
        d(\pin(y),r) \geq d(y,r)-d(y,\pin(y))=d_y-d(y,\pin(y)) \geq d_{\pin(y)}.
    \]
    By \cref{lem:scin-dist-valid}, we have $d_{\pin(y)}\leq d(\pin(y),r)$. Therefore, $d(\pin(y),r)=d_{\pin(y)}$ holds. Since $w(y,\pin(y))\geq 0$, $d_{\pin(y)}\leq d_y<\Delta_r$ also holds.
    \item We assume $y$ is not an in-Steiner vertex henceforth.
    \item If $x\notin \Ain_y$: By \cref{obs:list}~\ref{obs:in-list}, one of the following holds:
    \begin{itemize}
        \item {\bf(Case R2)} $x$ is an in-Steiner vertex with $\pin(x)=y$: Let $x'$ be the predecessor of $x$ on the path $Q$. Such a vertex must exist since $\bar u$ is not an in-Steiner vertex. Since $x$ is an in-Steiner vertex and $x'\in \Nin(x)$, by \cref{lem:in-heavy-edge}~\ref{item:in-1}, there is an edge $(x',y)$ of weight at most $w(x',x)+w(x,y)$. Then we can recursively invoke $\findpathin(x',y,Q)$ and return the resulting path. \cref{asmp:scin-x} is satisfied since
        \[
            d(x',r)=w(x',x)+w(x,y)+d(y,r)\geq w(x',y)+d(y,r).
        \]
        Combined with the triangle inequality $d(x',r)\leq w(x',y)+d(y,r)$, we obtain $d(x',r)=w(x',y)+d(y,r)$. By \cref{asmp:scin-Q}, we have $w(x',x)\geq 0$. Thus,
        \[
        d(x',r)=w(x',x)+d(x,r)\geq d(x,r)\geq \Delta_r.
        \]
        \item $y$ is an in-Steiner vertex: This case is already handled.
        \item $y$ is an out-Steiner vertex with $\pout(y)=x$: We can return the path $\bar u\rightsquigarrow x\to \tilde r_t$ where the $\bar u\rightsquigarrow x$ segment is $Q[\bar u,x]$. This is because when $y$ is processed in $\scin(r,t)$, in Step~\ref{step:scin-out}, for $W=w(x,y)+d_y-\Delta_r$, the algorithm adds an edge $(x,\tilde r_t)$ of weight $W$ if $W\geq 0$. Note that
        \[
            d(x,r)=w(x,y)+d(y,r)=w(x,y)+d_y.
        \]
        Since $d(x,r)\geq \Delta_r$, $W\geq 0$ must hold. Therefore, the edge $(x,\tilde r_t)$ of weight $W=d(x,r)-\Delta_r$ is added. The output path has weight
        \[
        w(Q[\bar u,x])+w(x,\tilde r_t)=d(\bar u,x)+d(x,r)-\Delta_r=d(\bar u,r)-\Delta_r,
        \]
        and it contains $0$ hops.
    \end{itemize}
    \item If $x\in \Ain_y$: When vertex $y$ is processed in $\scin(r,t)$, in Step~\ref{step:scin-sort}, $\Ain_y=\{u_1,\dots,u_k\}$ is sorted and we computed $j$ and $\ell=\lceil \log(k-j+1)\rceil$. Since $x\in \Ain_y$, we may assume $x=u_i$. Now we consider three cases:
    \begin{itemize}
        \item $i\leq k-2^{\ell}$: In Step~\ref{step:in-down} of the procedure $\shortcut$, the algorithm adds an edge $(x,y_{\ell})$ of weight $w(x,y)-w(y_\ell,y)$. When vertex $y$ is processed in $\scin(r,t)$, in Step~\ref{step:scin-steiner}, the algorithm adds an edge $(y_{\ell},\tilde r_t)$ of weight $w(y_\ell,y)+d_y-\Delta_r$. We can therefore return the path $\bar u\rightsquigarrow x\to y_{\ell}\to \tilde r_t$ where the $\bar u\rightsquigarrow x$ segment is $Q[\bar u,x]$. It has weight
        \begin{align*}
            & \quad \; w(Q[\bar u,x])+w(x,y_\ell)+w(y_{\ell},\tilde r_t)\\
            & = d(\bar u,x)+w(x,y)-w(y_\ell,y)+w(y_\ell,y)+d_y-\Delta_r\\
            & = d(\bar u,x)+d(x,r)-\Delta_r\\
            & = d(\bar u,r)-\Delta_r,
        \end{align*}
        and it contains $0$ hops.
        \item $i>k-2^{\ell}$ and $x$ is a light in-Steiner vertex: When vertex $y$ is processed in $\scin(r,t)$, in Step~\ref{step:scin-bf-in}, we have $W=w(x,y)+d_y-w(x,\pin(x))$. Consider two cases:
        \begin{itemize}
            \item {\bf(Case R3)} $W<\Delta_r$: Then we can recursively invoke $\findpathin(x,\pin(x),Q)$ and return the resulting path. \cref{asmp:scin-y} is satisfied since in Step~\ref{step:scin-bf-in}, the algorithm must have updated $d_{\pin(x)}$ with $W$ and thus $\pin(x)$ is processed. In addition, by triangle inequality,
            \[
             d(\pin(x),r)\geq d(x,r)-d(x,\pin(x))=w(x,y)+d(y,r)-d(x,\pin(x))=W\geq d_{\pin(x)}.
            \]
            By \cref{lem:scin-dist-valid}, we have $d_{\pin(x)}\leq d(\pin(x),r)$. Therefore, $d(\pin(x),r)=d_{\pin(x)}$ holds. Since $W<\Delta_r$, $d_{\pin(x)}\leq W<\Delta_r$ also holds.
            \item {\bf(Case R4)} $W\geq \Delta_r$: Since $x$ is a light in-Steiner vertex and $y\in \Nout(x)$ is not an in-Steiner vertex, by \cref{lem:in-light-edge}, $y$ must be an N-Steiner vertex ($y$ cannot be $\pin(x)$ since $x\in \Ain_y$). By \cref{lem:in-heavy-edge}~\ref{item:in-3}, there is an edge $(\pin(x),y)$ of weight $w(x,y)-w(x,\pin(x))$ in $\Fin$. Since $W=w(x,y)+d_y-w(x,\pin(x))\geq \Delta_r$, we have
            \[
            w(x,y)-w(x,\pin(x))\geq \Delta_r-d_y=\Delta_r-d(y,r).
            \]
            Since $d(y,r)<\Delta_r$, we see that $w(x,y)-w(x,\pin(x))\geq 0$. Therefore, by \cref{obs:f}~\ref{obs:fin}, the edge $(\pin(x),y)$ of weight $w(x,y)-w(x,\pin(x))$ is added into $G$. Consider the path $Q'=\bar u\rightsquigarrow x\to \pin(x)\to y\rightsquigarrow r$, where the segment $\bar u\rightsquigarrow x$ is $Q[\bar u,x]$, and the segment $y\rightsquigarrow r$ is an arbitrary $(y,r)$ shortest path. We can recursively invoke $\findpathin(\pin(x),y,Q')$ and return the resulting path. We first show that \cref{asmp:scin-x} is satisfied. By triangle inequality,
            \begin{align*}
            d(\pin(x),r)& \geq d(x,r)-w(x,\pin(x))\\
            & = w(x,y)+d(y,r)-w(x,\pin(x))\\
            & \geq w(\pin(x),y)+d(y,r).
            \end{align*}
            Combined with the triangle inequality $d(\pin(x),r)\leq w(\pin(x),y)+d(y,r)$, we obtain $d(\pin(x),r)=w(\pin(x),y)+d(y,r)$. Since
            \[
            W=w(x,y)+d_y-w(x,\pin(x))=d(x,r)-w(x,\pin(x))\geq \Delta_r,
            \]
            we have
            \[
            d(\pin(x),r) \geq d(x,r)-w(x,\pin(x))\geq \Delta_r.
            \]
            
            Now we show that \cref{asmp:scin-Q} is also satisfied. Indeed,
            \begin{align*}
                w(Q') & =w(Q[\bar u,x])+w(x,\pin(x))+w(\pin(x),y)+d(y,r)\\
                & \leq  d(\bar u,x)+w(x,\pin(x))+ w(x,y)-w(x,\pin(x))+d(y,r)\\
                & = d(\bar u,x)+w(x,y)+d(y,r) \\
                & = d(\bar u,r).
            \end{align*}
            In addition, $Q'[\bar u,x]$ contains $0$ hops since $Q[\bar u,x]$ contains $0$ hops.
        \end{itemize}
        \item $i>k-2^{\ell}$ and $x$ is not a light in-Steiner vertex: We can return the path $\bar u\rightsquigarrow x\to \tilde r_t$, where the $\bar u\rightsquigarrow x$ segment is $Q[\bar u,x]$. This is because when vertex $y$ is processed in $\scin(r,t)$, in Step~\ref{step:scin-bf-other}, for $W=w(x,y)+d_y-\Delta_r$, the algorithm adds an edge $(x,\tilde r_t)$ of weight $W$ if $W\geq 0$. Note that
        \[
            d(x,r)=w(x,y)+d(y,r)=w(x,y)+d_y.
        \]
        Since $d(x,r)\geq \Delta_r$, $W\geq 0$ must hold. Therefore, the edge $(x,\tilde r_t)$ of weight $W=d(x,r)-\Delta_r$ is added. The output path has weight
        \[
        w(Q[\bar u,x])+w(x,\tilde r_t)=d(\bar u,x)+d(x,r)-\Delta_r=d(\bar u,r)-\Delta_r,
        \]
        and it contains $0$ hops.
    \end{itemize}
\end{itemize}

Now we show that this algorithm terminates. The algorithm $\findpathin(x,y,Q)$ makes a recursive call in four cases, denoted by Case R1--R4. We now consider each of them:
\begin{itemize}
    \item Case R1: This case occurs when the second parameter $y$ is an in-Steiner vertex. In the recursive call, $y$ is replaced by $\pin(y)$. By \cref{obs:steiner}~\ref{obs:pin}, $\pin(y)$ is not an in-Steiner vertex. We note that once $y$ becomes non-in-Steiner vertex, Case R1 cannot occur again. This is because in Cases R2 and R4, $y$ remains unchanged, and in Case R3, $y$ is replaced by $\pin(x)$, which is also not an in-Steiner vertex by \cref{obs:steiner}~\ref{obs:pin}. Therefore, Case R1 can occur at most once at the beginning, after which only Cases R2–R4 may occur.
    \item Case R4: $\findpathin(\pin(x),y,Q')$ is called. We claim that this call terminates without further recursion. If a recursive call is made within it, only Cases R2--R4 is possible, and in each of these cases the first parameter $x$ must be an in-Steiner vertex. However, by \cref{obs:steiner}~\ref{obs:pin}, $\pin(x)$ is not an in-Steiner vertex. Hence, no recursive call is possible, and the call terminates immediately.
    \item Case R2: $\findpathin(x',y,Q)$ is called. In this case, $x'$ is replaced by the predecessor of $x$ in the path $Q$.
    \item Case R3: $\findpathin(x,\pin(x),Q)$ is called. In this case, $y$ is replaced by $\pin(x)$. As shown in the proof of \cref{lem:scin-dist-valid}, if $y$ is a level-$\ell$ vertex, then $\pin(x)$ must be a level-$\ell'$ vertex for some $\ell'<\ell$. 
\end{itemize}
Therefore, the algorithm terminates after a finite number of recursive calls.

Consider the segment $P[\bar u,r]$, let $x$ be the last vertex on this segment with $d(x,r)=d^0(x,r)\geq \Delta_r$ and let $y$ be the successor of $x$ with $d(y,r)=d^0(y,r)<\Delta_r$. Such a vertex must exist since $\Delta_r>0$. By \cref{lem:compute-U}, we have $x\not\in \Uin_r,y\in \Uin_r$. We can then invoke $\findpathin(x,y,P[\bar u,r])$, it is straightforward to verify that this is a valid input. The returned path is the desired path. This completes the proof of \cref{clm:shortcut-in}.

By a symmetric argument, one can prove \cref{clm:shortcut-out}. We provide its proof in Appendix~\ref{sec:appendix} for completeness.
\end{proof}

\subsubsection{Overall Runtime} \label{sec:runtime}
We analyze the runtime of \cref{alg:sssp} in this section. Note that in Step~\ref{step:betweenness-reduction} of $\shortcut(G,t)$, we make recursive calls to the single-source shortest path algorithm.

We first analyze the runtime of \cref{alg:sssp} excluding these recursive calls. Recall that the input graph $G_0$ to \cref{alg:sssp} consists of $n_0$ vertices. We assume $n_0$ is large enough. In each iteration of $\shortcut(G,t)$, additional Steiner vertices may be created. In Step~\ref{step:steiner-vertex-in}, for every heavy vertex $v$, we may create at most $O(\log n)$ in-Steiner vertices $u$ with $\pin(u)=v$. In Step~\ref{step:steiner-vertex-out}, for every heavy vertex $v$, we may create at most $O(\log n)$ out-Steiner vertices $u$ with $\pout(u)=v$. In Step~\ref{step:tilde-r}, for every negative vertex $r\in N$, we create an N-Steiner vertex $\tilde r_t$. Therefore, in each iteration of $\shortcut(G,t)$, the number of vertices added to the graph is bounded by
\[
\text{[\# of heavy vertices]}\cdot O(\log n)+\eta.
\]

Let $n^\star$ denote the number of vertices in $G$ at the end of \cref{alg:sssp}. We prove that $n^\star = O(n_0\log^2 n^\star\log^2 n_0)$. Initially, all vertices in the graph $G_0$ are heavy, so there are $n_0$ heavy vertices. Consider a single iteration of $\shortcut(G,t)$, let $n$ denote the current number of vertices in $G$. By \cref{clm:heavy-increase}, the number of heavy vertices increases by $O(\eta\log n)=O(n_0 \log n^\star)$. Since there are $L=O(\log n_0)$ iterations in total, the number of heavy vertex is always bounded by $O(n_0\log n^\star\log n_0)$. Then, by the discussion above, the number of new vertices added in each iteration is bounded by
\[
O(n_0\log n^\star\log n_0\cdot \log n+\eta) = O(n_0\log ^2n^\star\log n_0).
\]
Therefore, after $L=O(\log n_0)$ iterations, the number of vertices in $G$ is bounded by $O(n_0\log^2 n^\star\log^2 n_0)$, as desired. Finally, since $n^\star = O(n_0\log^2 n^\star\log^2 n_0)$, we have $n^\star = O(n_0\log^4 n_0)$. We conclude that throughout \cref{alg:sssp}, the number of vertices in $G$ is always bounded by $O(n_0\log^4n_0)$.

Now we analyze the running time of a single iteration of the procedure $\shortcut(G,t)$. By the above argument, the number of vertices satisfies $n=\tilde O(n_0)$. We show that, excluding the recursive call in Step~\ref{step:betweenness-reduction}, the running time of $\shortcut(G,t)$ is bounded by $\tilde O(n^2\cdot \gamma)=\tilde O(n_0^2\cdot \gamma)$. From the description of the algorithm, it is clear that the running time of all steps other than Steps~\ref{step:pre} and~\ref{step:shortcut} is bounded by $\tilde O(n^2)$. In Step~\ref{step:pre}, the running time of the algorithm from \cref{lem:compute-delta} is $\tilde O(hn^2)=\tilde O(n^2)$ since $h=O(\log n_0)$. For every negative vertex $r\in N$, the algorithm from \cref{lem:compute-U} requires $\tilde O((|\Uin_r|+|\Uout_r|)^2)$ time. Since $\Uin_r\subseteq \Vin_r, \Uout_r\subseteq \Vout_r$, and by \cref{lem:compute-delta}, 
\[
\sum_{r\in N} (|\Vin_r|+|\Vout_r|)^2=O\left(\frac{\eta n^2}{b}\right) = O(n^2\cdot \gamma),
\]
the total running time across all negative vertices $r\in N$ is bounded by $\tilde O(n^2\cdot \gamma)$.

In Step~\ref{step:shortcut}, by \cref{lem:scin-time} and \cref{lem:scout-time}, the total running time for $\scin(r,t)$ and $\scout(r,t)$ over all negative vertex $r\in N$ is bounded by
\begin{align*}
&\quad\; \sum_{r\in N}\tilde O\left(|\Vin_r|^2+
\lambda|\Vin_r|\right)+\tilde O\left(|\Vout_r|^2+
\lambda|\Vout_r|\right)\\
& = \sum_{r\in N} \tilde O\left(\left(|\Vin_r|+|\Vout_r|\right)^2\right) + \tilde O\left(\lambda \cdot \left(|\Vin_r|+|\Vout_r|\right)\right)\\
& = \tilde O\left(\frac{\eta n^2}{b}\right)+O\left(\lambda \cdot \frac{\eta n}{\sqrt{b}}\right)& & \text{(by \cref{lem:compute-delta})}\\
& = \tilde O\left(\frac{\eta n^2}{b} + \frac{n}{\sqrt{b}}\cdot \frac{\eta n}{\sqrt{b}}\right) & & \left(\text{since }\lambda=n/\sqrt{b}\right)\\
& = \tilde O\left(\frac{\eta n^2}{b}\right)= \tilde O(n^2\cdot \gamma) && (\text{since }b=\eta/\gamma)\\
\end{align*}
Since there are $L=O(\log n_0)$ iterations of $\shortcut(G,t)$, we conclude that the running time of \cref{alg:sssp} excluding the recursive calls is bounded by $\tilde O(n_0^2\cdot \gamma)$.

Now we consider the recursive calls to \cref{alg:sssp}. In each iterations of $\shortcut(G,t)$, Step~\ref{step:betweenness-reduction} makes a single recursive call on a graph with $O(hn)=O(hn_0\log^4 n_0)=O(n_0\log^5 n_0)$ vertices and $O(b\log n)=O(\eta\log n/\gamma)$ negative vertices. Let $C$ be a large enough constants such that $L\leq C\log n_0$; and the above bounds can be rewritten as at most $Cn_0\log^5 n_0$ vertices and at most $C\eta\log n_0/\gamma$ negative vertices.

Define
\[
\zeta = \log _{\frac{\gamma}{C\log n_0}} \eta.
\]
Note that after at most $\zeta$ levels of recursion, the number of negative vertex drops below $100$, and the instance becomes the trivial case.

Let $n_{\zeta}$ denote the initial number of vertices for a level-$\zeta$ recursion instance. Since the number of vertices grows from $n$ to at most $Cn\log^5 n$ in each recursion level, we see that
\[
n_{\zeta} \leq n_0 \cdot \left(C\log^5 n_{\zeta}\right)^\zeta.
\]

At recursion level $1\leq i\leq \zeta$, there are at most $L^i=(C\log n_{\zeta})^i\leq (C\log n_{\zeta})^\zeta$ recursion instances. Therefore the total running time over all recursion levels is bounded by
\[
\zeta \cdot (C\log n_{\zeta})^\zeta \cdot \tilde O\left((n_{\zeta})^2\cdot \gamma\right) = \tilde O\left(\zeta \cdot \gamma \cdot \left(C\log^5 n_{\zeta}\right)^{3\zeta} \cdot n_0^2\right).
\]
Choosing $\gamma = C\log n_0\cdot 2^{\sqrt{\log n_0}}$, the above time is bounded by $(n_0)^{2+o(1)}$.

In particular, we have that $\gamma=(n_0)^{o(1)}$ and
\[
\zeta = \log _{\frac{\gamma}{C\log n_0}} \eta \leq \frac{\log \eta}{\sqrt{\log n_0}} \leq \sqrt{\log n_0}.
\]
On the other hand, since $n_{\zeta} \leq n_0 \cdot \left(C\log^5 n_{\zeta}\right)^\zeta$, we have
\[
\log n_{\zeta} \leq \log n_0+\sqrt{\log n_0}\cdot  (\log C+5\log\log n_{\zeta}).
\]
Since $n_{\zeta}\geq n_0$ is large enough, we may assume $\log C+5\log\log n_{\zeta}\leq \sqrt{\log n_{\zeta}}$. Then, we have
\[
\log n_{\zeta} \leq \log n_0 + \sqrt{\log n_0} \cdot \sqrt{\log n_{\zeta}},
\]
which implies that $\log n_{\zeta} = O(\log n_0)$. Then the running time can be simplified as
\[
\tilde O\left(\left(C'\log^5 n_0\right)^{3\sqrt{\log n_0}} \cdot (n_0)^{2+o(1)}\right)
\]
for some constant $C'$, and it is clear that this is bounded by $(n_0)^{2+o(1)}$.

\bibliographystyle{alpha}
\bibliography{ref}

\appendix
\section{Omitted Proofs} \label{sec:appendix}
\subsection{Proof of \cref{lem:out-heavy-edge}}
By \cref{obs:edge-regular}, we only need to consider the edges added in Steps~\ref{step:in-up}, \ref{step:in-down}, \ref{step:in-r}, \ref{step:in-tilde-r}, \ref{step:out-up}, \ref{step:out-down}, \ref{step:out-r}, \ref{step:out-tilde-r} of $\shortcut(G,t)$, and the edges added via procedures $\ein$ and $\eout$.

\paragraph{Proof of (a).}
We first consider the outgoing edges of $v$ and prove the first item. Specifically, we show that whenever an outgoing edge from $v$ is added, the claim holds. That is, when an edge $(v,u)$ of weight $w(v,u)$ is added, an edge $(\pout(v),u)$ of weight at most $w(\pout(v),v)+w(v,u)$ is also added. After the edges are added, the algorithm may apply vertex potentials. However, it is easy to see that the claim continues to hold after applying any vertex potential.

We first focus on Steps~\ref{step:in-up}, \ref{step:in-down}, \ref{step:in-r}, \ref{step:in-tilde-r}, \ref{step:out-up}, \ref{step:out-down}, \ref{step:out-r}, \ref{step:out-tilde-r} of $\shortcut(G,t)$. When vertex $v$ is created in Step~\ref{step:steiner-vertex-out} of the $t$-th iteration, the algorithm adds outgoing edges from $v$ in Step~\ref{step:out-down}. In particular, for some out-neighbor $u$ of $\pout(v)$, it may add an edge $(v,u)$ of weight $w(\pout(v),u)-w(\pout(v),v)$. This edge satisfies the claim trivially. If $u$ is an in-Steiner vertex, the algorithm may add an edge $(v,\pin(u))$ of weight $w(\pout(v),u)+w(u,\pin(u))-w(\pout(v),v)$. Note that since $u$ is an in-Steiner vertex and $\pout(v)\in \Nin(u)$, \cref{lem:in-heavy-edge}~\ref{item:in-1} implies that there is an edge $(\pout(v),\pin(u))$ of weight at most $w(\pout(v),u)+w(u,\pin(u))$. Therefore, the edge $(v,\pin(u))$ also satisfies the claim. Next, if $\pout(v)\in \bar N$ and $r$ is the unique in-neighbor of $\pout(v)$, the algorithm adds an edge $(v,r)$ of weight $-w(r,\pout(v))-w(\pout(v),v)$ in Step~\ref{step:out-r}. Finally, if $\pout(v)$ is an N-Steiner vertex, it adds an edge $(v,\pn(\pout(v)))$ of weight $w(\pout(v),\pn(\pout(v)))-w(\pout(v),v)$ in Step~\ref{step:out-tilde-r}. It is straightforward to verify that all these edges trivially satisfy the claim.

In subsequent iterations, the algorithm never adds outgoing edges from $v$ in Steps~\ref{step:in-up} and~\ref{step:out-down} as edges added in these steps are directed into newly created vertices. The same holds for Steps~\ref{step:in-r}, \ref{step:in-tilde-r}, \ref{step:out-r}, \ref{step:out-tilde-r}, because the endpoints of the added edges are either negative vertices or newly created vertices. In addition, it also never adds outgoing edges from $v$ in Step~\ref{step:out-up}, since $v$ is an out-Steiner vertex so $\Aout_v=\emptyset$. The only remaining case is Step~\ref{step:in-down}. In this case, since $v$ is an out-Steiner vertex, \cref{obs:steiner}~\ref{obs:pout} implies that it can never be the out-parent of any vertex. Thus, the algorithm can only add an edge $(v,u)$ of weight $W$ for some in-Steiner vertex $u$, followed by an edge $(\pout(v),u)$ of weight $W+w(\pout(v),v)$, which satisfies the claim.

Next we focus on the edges added by the procedures $\ein$ and $\eout$. Since procedure $\eout$ only adds edges from a negative vertex or an N-Steiner vertex, outgoing edges from $v$ can be added only via $\ein$. When an edge $(v,u)$ of weight $W$ is added for some $u\in V$, since $v$ is an out-Steiner vertex, the algorithm must have recursively invoked $\eout(\pout(v),u,W+w(\pout(v),v))$ and added an edge $(\pout(v),u)$ of weight $W+w(\pout(v),v)$.

Finally, it is straightforward to verify that if the claim holds before Step~\ref{step:r'}, then it continues to hold after Step~\ref{step:r'}. This completes the proof of the first item.

\paragraph{Proof of (b).}
Now we consider the incoming edges to $v$ and prove the second item. During the $t$-th iteration, the algorithm first adds the edge $w(\pout(v),v)$ in Step~\ref{step:out-up}. Then, in Step~\ref{step:scout-steiner} of $\scout(r,t)$ for some negative vertex $r\in N$, it may invoke $\eout(\tilde r_t,v,W)$ and add an edge $(\tilde r_t,v)$ of weight $W$. Note that we will never add any incoming edge to $v$ in recursive calls to $\eout$ since $v$ is neither the in-parent nor the out-parent of any vertex. Therefore, by the end of the $t$-th iteration, $\Nin(v)$ only contains $\pout(v)$ and some level-$t$ N-Steiner vertices.

In subsequent iterations, consider Steps~\ref{step:in-up}, \ref{step:in-down}, \ref{step:in-r}, \ref{step:in-tilde-r}, \ref{step:out-up}, \ref{step:out-down}, \ref{step:out-r}, \ref{step:out-tilde-r} of $\shortcut(G,t)$. The algorithm never adds incoming edges to $v$ in Steps~\ref{step:in-down} and~\ref{step:out-up} as edges added in these steps are directed into newly created vertices. The same holds for Steps~\ref{step:in-r}, \ref{step:in-tilde-r}, \ref{step:out-r}, \ref{step:out-tilde-r}, because the endpoints of the added edges are either negative vertices or newly created vertices. In Step~\ref{step:in-up}, the algorithm may only add an edge $(u,v)$ for some in-Steiner vertex $u$ with $\pin(u)=v$, which satisfies the claim. The only remaining case is Step~\ref{step:out-down}. There are two cases:
\begin{itemize}
    \item The first case is the algorithm adds an edge $(u,v)$ for some out-Steiner vertex $u$, given that the edge $(\pout(u),v)$ already exists in the graph. Assume \cref{lem:out-heavy-edge}~\ref{item:out-2} holds before adding this edge. Since $\pout(u)\in \Nin(v)$, and an out-Steiner vertex cannot be the out-parent of any vertex by \cref{obs:steiner}~\ref{obs:pout}, $\pout(u)$ must fall into one of the following categories: $\pout(u)$ is an in-Steiner vertex with $\pin(\pout(u))=v$; $\pout(u)\in \bar N$; or $\pout(u)$ is an N-Steiner vertex. We now show that the first case is not possible. This is because the algorithm adds the edge $(u,v)$ only if $v\in \Aout_{\pout(u)}$. However, if $v=\pin(\pout(u))$, then $v$ would have been removed from $\Aout_{\pout(u)}$ in Step~\ref{step:delete} of $\shortcut(G,t)$. Therefore, it must be either $\pout(u)\in N$ or $\pout(u)$ is an N-Steiner vertex, and \cref{lem:out-heavy-edge}~\ref{item:out-2} continues to hold after adding the edge $(u,v)$.
    \item The second case is the algorithm adds an edge $(u,v')$ for some out-Steiner vertex $u$ and in-Steiner vertex $v'$ with $\pin(v')=v$, given that the edge $(\pout(u),v')$ already exists in the graph. Then the algorithm adds an edge $(u,\pin(v'))=(u,v)$. In this case, since $v'$ is an in-Steiner vertex with $\pin(v')=v$ and $\pout(u)\in \Nin(v')$, by \cref{lem:in-heavy-edge}~\ref{item:in-1}, we must have $\pout(u)\in \Nin(v)$. Assume \cref{lem:out-heavy-edge}~\ref{item:out-2} holds before adding this edge. Since $\pout(u)\in \Nin(v)$, and an out-Steiner vertex cannot be the out-parent of any vertex by \cref{obs:steiner}~\ref{obs:pout}, $\pout(u)$ must fall into one of the following categories: $\pout(u)$ is an in-Steiner vertex with $\pin(\pout(u))=v$; $\pout(u)\in \bar N$; or $\pout(u)$ is an N-Steiner vertex. We again show that the first case is not possible. This is because in that case, $\pout(u)$ is an in-Steiner vertex and $v'\in \Nout(\pout(u))$. Since $v'$ is an in-Steiner vertex with $\pin(v')=v$, but $v$ is neither a negative nor an N-Steiner vertex, this would contradict \cref{lem:in-heavy-edge}~\ref{item:in-2} applied to the vertex $\pout(u)$. Therefore, it must be either $\pout(u)\in \bar N$ or $\pout(u)$ is an N-Steiner vertex, and \cref{lem:out-heavy-edge}~\ref{item:out-2} continues to hold after adding the edge $(u,v)$.
\end{itemize}

Next, consider the edges added by the procedure $\ein$ and $\eout$. Since procedure $\ein$ add edges only to a negative vertex or an N-Steiner vertex, incoming edges to $v$ can be added only via $\eout$. In any call to $\eout(u,v,W)$, $u$ must be a negative vertex or an N-Steiner vertex, and therefore the claim always holds. Finally, in Step~\ref{step:r'}, the algorithm may replace edges $(r,v)$ by $(\bar r,v)$ for some negative vertices $r\in N$, and the claim continues to hold. This completes the proof for the second item.

\paragraph{Proof of (c).}
For the third item, note that the algorithm adds incoming edges from negative vertices or to N-Steiner vertices to $v$ only via the procedure $\eout$. The claim therefore follows directly from the description of $\eout$. (For the edges from negative vertices, they will be replaced by edges from vertices in $\bar N$ after Step~\ref{step:r'}, and it is straightforward to verify that the claim holds). It is easy to see that the claim continues to hold after applying any vertex potential.

\subsection{Proof of \cref{clm:eout-light}}
We focus on the initial calls to the procedure $\eout$. For the recursive calls, the claim is immediate since a light vertex cannot be the in-parent or out-parent of any vertex. We now consider each location where $\eout$ is invoked:
\begin{itemize}
    \item Step~\ref{step:shortcut} of $\shortcut(G,t)$: Since $\bar r\in \Nin(v)$, \cref{lem:out-light-edge} implies that $v$ is not a light out-Steiner vertex. Note that $\bar r$ cannot be the out-parent of $v$, as in that case $v$ would have been deleted from $\Aout_{\bar r}$.
    \item Step~\ref{step:scout-in} of $\scout$: $\pin(v)$ is heavy.
    \item Step~\ref{step:scout-steiner} of $\scout$: $v_\ell$ is a vertex created in the current iteration.
    \item Step~\ref{step:scout-bf-other} of $\scout$: $u_i$ is not a light out-Steiner vertex.
    \item Step~\ref{step:scout-light-in} of $\scout$: Since $v$ is light and $v$ is not an out-Steiner vertex, by \cref{obs:light-vertex}, $v$ must be a light in-Steiner vertex. Then by \cref{lem:in-light-edge}, since $u\in \Nout(v)$, $u$ is not a light out-Steiner vertex.
\end{itemize}

\subsection{Proof of \cref{lem:out-light-edge}}
\cref{lem:out-heavy-edge} implies that by the end of the $t$-th iteration, $\Nin(v)$ only contains $\pout(v)$ and some level-$t$ N-Steiner vertices. We denote the latter by the set $\Sout_v$. Since $v$ is classified as light in Step~\ref{step:classify-out}, we have $|\Sout_v|\leq \lambda$.

In subsequent iterations, consider Steps~\ref{step:in-up}, \ref{step:in-down}, \ref{step:in-r}, \ref{step:in-tilde-r}, \ref{step:out-up}, \ref{step:out-down}, \ref{step:out-r}, \ref{step:out-tilde-r} of $\shortcut(G,t)$. The algorithm never adds incoming edges to $v$ in Steps~\ref{step:in-down} and~\ref{step:out-up} as edges added in these steps are directed into newly created vertices. The same holds for Steps~\ref{step:in-r}, \ref{step:in-tilde-r}, \ref{step:out-r}, \ref{step:out-tilde-r}, because the endpoints of the added edges are either negative vertices or newly created vertices. In addition, it also never adds incoming edges to $v$ in Step~\ref{step:in-up}, since $v$ is light. The only remaining case is Step~\ref{step:out-down}. In this case, since $v$ is light, it can never be the out-parent of any vertex. Thus, the algorithm can only add an edge $(u,v)$ for some out-Steiner vertex $u$, given that the edge $(\pout(u),v)$ already exists in the graph. Assuming \cref{lem:out-light-edge} holds before adding this edge, then we must have $\pout(u)\in \Sout_v$ since an out-Steiner vertex cannot be the out-parent of any vertex by \cref{obs:steiner}~\ref{obs:pout}. Therefore, the lemma continues to hold after adding this edge. We conclude the proof by noting that the algorithm never adds incoming edges to $v$ via procedure $\eout$ by \cref{clm:eout-light}, and it also never adds incoming edges from $v$ in Step~\ref{step:r'} since $\Nin(v)$ does not contain negative vertices.

\subsection{Analysis of $\scout$: Proof of \cref{lem:scout-time}}

\begin{lemma} \label{lem:plus-one-hop-out}
Let $r\in N$ be a negative vertex. For any out-Steiner vertex $v$, if there exists an $h$-hop $r\rightsquigarrow v$ path $P$ in $G$ of weight $W$, then there exists an $(h+1)$-hop path $r\rightsquigarrow \pout(v)$ in $G$ of weight at most $W-w(\pout(v),v)$.
\end{lemma}
\begin{proof}
Let $u$ be the predecessor of $v$ on the path $P$. Such a vertex must exist since $r$ is not an out-Steiner vertex. Since $v$ is an out-Steiner vertex and $u\in \Nin(v)$, by \cref{lem:out-heavy-edge}~\ref{item:out-2}, one of the following holds:
\begin{itemize}
    \item $u$ is an in-Steiner vertex with $\pin(u)=v$: Let $u'$ be the predecessor of $u$ on the path $P$. Such a vertex must exist since $r$ is not an in-Steiner vertex. Since $u'\in \Nin(u)$ and $\pin(u)=v$, by \cref{lem:in-heavy-edge}~\ref{item:in-1}, there is an edge $(u',v)$ of weight at most $w(u',u)+w(u,v)=w(P[u',v])$. We note that $u'$ is not an in-Steiner vertex. Otherwise, since $u\in \Nout(u')$ and $\pin(u)=v$ is an out-Steiner vertex, this would contradict \cref{lem:in-heavy-edge}~\ref{item:in-2} applied to $u'$. Hence, we can replace $u$ by $u'$. After that, we still have $u\in \Nin(v)$ and $w(u,v)\leq w(P[u,v])$, but $u$ must now fall into one of the later cases of our analysis.
    \item $u=\pout(v)$: The desired path is $P[r,u]$. We have
    \[
    w(P[r,u])=w(P)-w(P[u,v])\leq w(P)-w(u,v)=w(P)-w(\pout(v),v),
    \]
    and $w(P)$ clearly contains at most $h$ hops.
    \item $u\in \bar N$: By the \cref{lem:out-heavy-edge}~\ref{item:out-3}, there is an edge $(u,\pout(v))$ of weight at most $w(u,v)-w(\pout(v),v)$. Therefore, the desired path is $r\rightsquigarrow u\to \pout(v)$, where the $r\rightsquigarrow u$ segment is $P[r,u]$. The weight of the path is
    \[
        w(P[r,u])+w(u,\pout(v))\leq w(P)-w(u,v)+w(u,v)-w(\pout(v),v)\leq w(P)-w(\pout(v),v),
    \]
    and it contains at most $h$ hops.
    \item $u$ is an N-Steiner vertex: By the \cref{lem:out-heavy-edge}~\ref{item:out-3}, there is an edge $(\pn(u),\pout(v))$ of weight at most $w(u,v)-w(u,\pn(u))-w(\pout(v),v)$. The desired path is $r\rightsquigarrow u\to\pn(u)\to\pout(v)$, where the $r\rightsquigarrow u$ segment is $P[r,u]$. The weight of this path is
    \begin{align*}
    & \quad\; w(P[r,u])+w(u,\pn(u))+w(\pn(u),\pout(v))\\
    & \leq w(P)-w(u,v)+w(u,\pn(u))+w(u,v)-w(u,\pn(u))-w(\pout(v),v)\\
    & =w(P)-w(\pout(v),v).
    \end{align*}
    Note that in the $u\to\pn(u)\to\pout(v)$ segment, $(\pn(u),\pout(v))$ is the only edge that can be negative. (Strictly speaking, the edge $(\pn(u),\pout(v))$ may actually correspond to a length-two path, see \cref{rmk:r'}. Nevertheless, such a length-two path contains at most one negative edge.) Therefore, this path consists of at most $h+1$ hops.
    \item $u$ is an out-Steiner vertex and $\pout(u)\in \bar N$: Since $v\in \Nout(u)$, by \cref{lem:out-heavy-edge}~\ref{item:out-1} applied to $u$, there is an edge $(\pout(u),v)$ of weight at most $w(\pout(u),u)+w(u,v)$. Since $\pout(u)\in\Nin(v)$ and $\pout(u)\in \bar N$, by \cref{lem:out-heavy-edge}~\ref{item:out-3} applied to $v$, there is an edge $(\pout(u),\pout(v))$ of weight at most
    \[
    w(\pout(u),v)-w(\pout(v),v)\leq w(\pout(u),u)+w(u,v)-w(\pout(v),v)
    \]
    Moreover, since $\pout(u)\in \bar N$, let $s\in N$ be the unique in-neighbor of $\pout(u)$. In Step~\ref{step:out-r} of the procedure $\shortcut$, it adds an edge $(u,s)$ of weight $-w(s,\pout(u))-w(\pout(u),u)$. Therefore, the desired path is $r\rightsquigarrow u\to s\to \pout(u)\to \pout(v)$, where the $r\rightsquigarrow u$ segment is $P[r,u]$. The weight of the path is
    \begin{align*}
    & \quad\; w(P[r,u])+w(u,s)+w(s,\pout(u))+w(\pout(u),\pout(v))\\
    & \leq w(P)-w(u,v)-w(s,\pout(u))-w(\pout(u),u)+w(s,\pout(u))\\
    & \quad\; + w(\pout(u),u)+w(u,v)-w(\pout(v),v)\\
    & = w(P)-w(\pout(v),v).
    \end{align*}
    Note that in the $u\to s\to \pout(u)\to \pout(v)$ segment, $(s,\pout(u))$ is the only edge that can be negative. Therefore, this path consists of at most $h+1$ hops.
    \item $u$ is an out-Steiner vertex and $\pout(u)$ is an N-Steiner vertex: Since $v\in \Nout(u)$, by \cref{lem:out-heavy-edge}~\ref{item:out-1} applied to $u$, there is an edge $(\pout(u),v)$ of weight at most $w(\pout(u),u)+w(u,v)$. Since $\pout(u)\in \Nin(v)$ and $\pout(u)$ is an N-Steiner vertex, by \cref{lem:out-heavy-edge}~\ref{item:out-3} applied to $v$, there is an edge $(\pn(\pout(u)),\pout(v))$ of weight at most
    \begin{align*}
    & \quad\; w(\pout(u),v) - w(\pout(u),\pn(\pout(u))) - w(\pout(v),v)\\
    &\leq w(\pout(u),u)+w(u,v)- w(\pout(u),\pn(\pout(u))) - w(\pout(v),v)
    \end{align*}
    Moreover, in Step~\ref{step:out-tilde-r} of the procedure $\shortcut$, it adds an edge $(u,\pn(\pout(u)))$ of weight $w(\pout(u),\pn(\pout(u)))-w(\pout(u),u)$. Therefore, the desired path is $r\rightsquigarrow u\to\pn(\pout(u))\to\pout(v)$, where the $r\rightsquigarrow u$ segment is $P[r,u]$. The weight of the path is
    \begin{align*}
    & \quad\; w(P[r,u])+w(u,\pn(\pout(u)))+w(\pn(\pout(u)),\pout(v))\\
    & \leq w(P)-w(u,v)+w(\pout(u),\pn(\pout(u)))-w(\pout(u),u)\\
    & \quad\; +w(\pout(u),u)+w(u,v)- w(\pout(u),\pn(\pout(u))) - w(\pout(v),v)\\
    & = w(P)-w(\pout(v),v).
    \end{align*}
    Note that in the $u\to\pn(\pout(u))\to\pout(v)$ segment, $(\pn(\pout(u)),\pout(v))$ is the only edge that can be negative. (Strictly speaking, the edge $(\pn(\pout(u)),\pout(v))$ may actually correspond to a length-two path, see \cref{rmk:r'}. Nevertheless, such a length-two path contains at most one negative edge.) Therefore, this path consists of at most $h+1$ hops.
\end{itemize}
Since the desired path can be found in all cases, the lemma follows.
\end{proof}

\begin{lemma} \label{lem:scout-dist-valid}
Throughout the algorithm $\scout(r,t)$, for any vertex $v\in V$, if $d_v<\infty$, then there exists an $(L+1)$-hop $r\rightsquigarrow v$ path in $G$ of weight at most $d_v$.
\end{lemma}
\begin{proof}
We prove the following claim, from which the lemma follows immediately.
\begin{claim}
For any level-$\ell$ vertex $v\in V$, if $d_v<\infty$, then there exists an $(t-\ell+1)$-hop $r\rightsquigarrow v$ path in $G$ of weight at most $d_v$.
\end{claim}
For every vertex $v\in \Uout_r$, we set $d_v=d^1(r,v)$ initially and the claim holds trivially. We now prove the claim inductively for distances updated in Steps~\ref{step:scout-out} and~\ref{step:scout-bf-out}.

Suppose we are at Step~\ref{step:scout-out}, and suppose $v$ is a level-$\ell$ vertex. By induction hypothesis, there exists a $(t-\ell+1)$-hop $r\rightsquigarrow v$ path in $G$ of weight at most $d_v$. By \cref{lem:plus-one-hop-out}, there exists a $(t-\ell+2)$-hop $r\rightsquigarrow \pout(v)$ path of weight at most $d_v-w(\pout(v),v)$. Moreover, by \cref{obs:steiner}~\ref{obs:lower-level-out}, $\pout(v)$ is a level-$\ell'$ vertex for some $\ell' < \ell$. Therefore, if we update $d_{\pout(v)}$ to $d_v-w(\pout(v),v)$ in Step~\ref{step:scout-out}, there must exist a $(t-\ell+2)\leq (t-\ell'+1)$-hop $r\rightsquigarrow \pout(v)$ path in $G$ with weight at most $d_{\pout(v)}$.

Suppose we are at Step~\ref{step:scout-bf-out}, and suppose $v$ is a level-$\ell$ vertex, $u_i$ is a level-$\ell'$ vertex. By induction hypothesis, there exists a $(t-\ell+1)$-hop $r\rightsquigarrow v$ path in $G$ of weight at most $d_v$. This implies that there is a $(t-\ell+1)$-hop $r\rightsquigarrow u_i$ path of weight at most $d_v+w(v,u_i)$. By \cref{lem:plus-one-hop-out}, there exists a $(t-\ell+2)$-hop $r\rightsquigarrow \pout(u_i)$ path of weight at most $d_v+w(v,u_i)-w(\pout(u_i),u_i)$. On the other hand, since $u_i$ is a light out-Steiner vertex, by \cref{lem:out-light-edge}, $v$ must fall into one of the following categories:
\begin{itemize}
    \item $v=\pout(u_i)$: This is not possible since we have $u_i\in \Aout_v$ at Step~\ref{step:scout-bf-out}. If $v=\pout(u_i)$, we would have deleted $u_i$ from $\Aout_v$ in Step~\ref{step:delete} of $\shortcut(G,t)$.
    \item $v$ is a level-$\ell'$ N-Steiner vertex: We have $\ell'=\ell$. By \cref{obs:steiner}~\ref{obs:lower-level-out}, $\pout(u_i)$ is a level-$\ell''$ vertex for some $\ell''<\ell'=\ell$.
    \item $v$ is an out-Steiner vertex where $\pout(v)$ is a level-$\ell'$ vertex: This is also not possible since at Step~\ref{step:scout-bf-out}, $v$ is not an out-Steiner vertex.
\end{itemize}
Therefore, we must have $\pout(v)$ is a level-$\ell''$ vertex for some $\ell'' < \ell$. Therefore, if we update $d_{\pout(u_i)}$ to $d_v+w(v,u_i)-w(\pout(u_i),u_i)$ in Step~\ref{step:scout-bf-out}, then there must be a $(t-\ell+2)\leq (t-\ell''+1)$-hop $r\rightsquigarrow \pout(u_i)$ path in $G$ with weight at most $d_{\pout(u_i)}$.
\end{proof}

\begin{claim} \label{clm:scout-dist}
During the algorithm $\scout(r,t)$, for any vertex $v\in V$, if $d_v<\infty$, then $d_v<-\Delta_r$.
\end{claim}
\begin{proof}
For every $v\in \Uout_r$, \cref{lem:compute-U} implies that $d^1(r,v)<-\Delta_r$. We inductively show that the bound holds for distances updated in Steps~\ref{step:scout-out} and~\ref{step:scout-bf-out}. In Step~\ref{step:scout-out}, since $d_v<-\Delta_r$ by induction hypothesis and $w(\pout(v),v)\geq 0$, we update $d_{\pout(v)}$ only with the value $d_v - w(\pout(v),v) < -\Delta_r$. In Step~\ref{step:scout-bf-out}, we update $d_{\pout(u_i)}$ with a value $W$ only when $W < -\Delta_r$. Therefore, the claim always holds.
\end{proof}

\begin{lemma}[\cref{lem:scout-time} Restated]
The running time of $\scout(r,t)$ is $\tilde O(|\Vout_r|^2+
\lambda|\Vout_r|)$.
\end{lemma}
\begin{proof}
Combining \cref{lem:scout-dist-valid} and \cref{clm:scout-dist}, we see that for every vertex $v\in V$ with $d_v<\infty$, $d^{L+1}(r,v)<-\Delta_r$ holds. Therefore, for every vertex $v$ that has been processed during $\scout(r,t)$, $v\in \Vout_r$ must hold. We first note that each vertex is only processed once. This follows from the proof of \cref{lem:scout-dist-valid}: whenever a level-$\ell$ vertex $v$ is processed, and a level-$\ell'$ vertex $v'$ is added in Step~\ref{step:scout-out} or Step~\ref{step:scout-bf-out}, $\ell'<\ell$ must hold.

Now we show that the time required for processing a single vertex $v$ is $\tilde O(|\Vout_r|+\lambda L)$. By \cref{clm:addedge-time}, the running time of $\eout$ is $\tilde O(1)$. We focus on the non-trivial steps:
\begin{itemize}
    \item Step~\ref{step:scout-bf}: The running time of this step is $O(2^{\ell})=O(k-j+1)$. We now show that $k-j+1\leq |\Vout_r|$. This is because for every $j\leq i\leq k$, $d_v+w(v,u_i)+\Delta_r<0$ holds. Since \cref{lem:scout-dist-valid} implies $d_v\geq d^{L+1}(r,v)$, by triangle inequality, we have
    \[
    d^{L+1}(r,u_i)\leq d^{L+1}(r,v)+w(v,u_i)\leq d_v+w(v,u_i)<-\Delta_r.
    \]
    Therefore, $u_i\in \Vout_r$ for every $j\leq i\leq k$, which implies that $k-j+1\leq |\Vout_r|$.
    \item Step~\ref{step:scout-light-in}: The running time of this step is $O(|\Aout_v|)=O(|\Nout(v)|)$. Since $v$ is not an out-Steiner vertex and $v$ is light, by \cref{obs:light-vertex}, $v$ must be a light in-Steiner vertex. Then by \cref{lem:in-light-edge}, for every $u\in \Nout(v)\setminus \{\pin(v)\}$, either $u\in \Sin_v$ or $u$ is an in-Steiner vertex with $\pin(u)\in \Sin_v$. Moreover, we have $|\Sin_v|\leq \lambda$. For every vertex $x\in \Sin_v$, each iteration of the $\shortcut$ algorithm may create at most $O(\log n)$ in-Steiner vertices $y$ with $\pin(y)=x$. Note that the number of vertices may increase, \cref{sec:runtime} shows that it is always bounded by $\tilde O(n)$. Since there are $L$ iterations in total, the number of such in-Steiner vertices $y$ is bounded by $O(L\log n)$. Therefore,
    \[
    \Nout(v)=O(|\Sin_v|L\log n)=O(\lambda L\log n)=\tilde O(\lambda).
    \]
\end{itemize}
\end{proof}

\subsection{Proof of \cref{clm:shortcut-out}}
First consider the case $\Delta_r\geq -w(r,\bar r)$.Let $x$ be the successor of $\bar r$ on the path $P[\bar r,v]$. There are two cases:
\begin{itemize}
    \item $x\notin \Aout_{\bar r}$: By \cref{obs:list}~\ref{obs:out-list}, it must be the case that $x$ is an out-Steiner vertex with $\pout(x)=\bar r$. In this case, let $y$ be the successor of $x$ on the path $P[\bar r,v]$. Such a vertex must exist since $v$ is not an out-Steiner vertex. Then by \cref{lem:out-heavy-edge}~\ref{item:out-1}, there is an edge $(\bar r,y)$ of weight at most $w(\bar r,x)+w(x,y)$. Moreover, by \cref{obs:steiner}~\ref{obs:pout}, $y$ is not an out-Steiner vertex. Since $y\in \Nout(\bar r)$ and $y$ is not an out-Steiner vertex, we must have $y\in\Aout_{\bar r}$. Thus, we have added the edge $(\tilde r_t,y)$ of weight $w(r,\bar r)+\Delta_r+w(\bar r,y)\leq w(P[r,y])+\Delta_r$ in Step~\ref{step:shortcut} of $\shortcut(G,t)$. Consider the path $\tilde r_t \to y \rightsquigarrow v$, where the segment $y \rightsquigarrow v$ is $P[y,v]$. The weight of this path is at most $d(r,v) + \Delta_r$.
    \item $x\in \Aout_{\bar r}$: We have added the edge $(\tilde r_t,x)$ of weight $w(r,\bar r)+\Delta_r+w(\bar r,x)$ in Step~\ref{step:shortcut} of $\shortcut(G,t)$. Consider the path $\tilde r_t \to x \rightsquigarrow v$, where the segment $x \rightsquigarrow v$ is $P[x,v]$. The weight of this path is exactly $d(r,v)+\Delta_r$.
\end{itemize}

This completes the proof for the case $\Delta_r\geq -w(r,\bar r)$. In the remainder, we assume $\Delta_r<-w(r,\bar r)$.

Consider the following recursive algorithm $\findpathout(x,y,Q)$. The input consists vertices $x,y\in V$ and a $r\rightsquigarrow v$ path $Q$ that satisfy the following assumptions:
\begin{assumption} \label{asmp:scout-x}
$x\in \Nout(y)$ and $d(r,x)=d(r,y)+w(y,x)\geq -\Delta_r$. In addition, $x$ lies on the path $Q$.
\end{assumption}
\begin{assumption} \label{asmp:scout-y}
During the execution of $\scout(r,t)$, vertex $y$ is processed, and at the moment it is processed, we have $d_y=d(r,y)<-\Delta_r$.
\end{assumption}
\begin{assumption} \label{asmp:scout-Q}
$w(Q)=d(r,v)$, and $Q[x,v]$ contains $0$ hops.
\end{assumption}
The output of the algorithm is a $0$-hop $\tilde r_t\rightsquigarrow v$ path of weight $d(r,v)+\Delta_r$.

Now we describe the algorithm $\findpathout(x,y,Q)$.
\begin{itemize}
    \item {\bf(Case R1)} If $y$ is an out-Steiner vertex: Since $y$ is an out-Steiner vertex and $x\in \Nout(y)$, by \cref{lem:out-heavy-edge}~\ref{item:out-1}, there is an edge $(\pout(y),x)$ of weight at most $w(\pout(y),y)+w(y,x)$. Then we can recursively invoke $\findpathout(x,\pout(y),Q)$ and return the resulting path. \cref{asmp:scout-y} is satisfied since when vertex $y$ is processed in $\scout(r,t)$, in Step~\ref{step:scout-out}, the algorithm must have updated $d_{\pout(y)}$ with $d_y-w(\pout(y),y)$ and thus $\pout(y)$ is also processed. In addition, by triangle inequality,
    \[
        d(r,\pout(y)) \geq d(r,y)-d(\pout(y),y)=d_y-d(\pout(y),y) \geq d_{\pout(y)}.
    \]
    By \cref{lem:scout-dist-valid}, we have $d_{\pout(y)}\leq d(r,\pout(y))$. Therefore, $d(r,\pout(y))=d_{\pout(y)}$ holds. Since $w(\pout(y),y)\geq 0$, $d_{\pout(y)}\leq d_y<-\Delta_r$ also holds.
    \item We assume $y$ is not an out-Steiner vertex henceforth.
    \item If $x\notin \Aout_y$: By \cref{obs:list}~\ref{obs:out-list}, one of the following holds:
    \begin{itemize}
        \item {\bf(Case R2)} $x$ is an out-Steiner vertex with $\pout(x)=y$: Let $x'$ be the successor of $x$ on the path $Q$. Such a vertex must exist since $v$ is not an out-Steiner vertex. Since $x$ is an out-Steiner vertex and $x'\in \Nout(x)$, by \cref{lem:out-heavy-edge}~\ref{item:out-1}, there is an edge $(y,x')$ of weight at most $w(y,x)+w(x,x')$. Then we can recursively invoke $\findpathout(x',y,Q)$ and return the resulting path. \cref{asmp:scout-x} is satisfied since
        \[
            d(r,x')=d(r,y)+w(y,x)+w(x,x')\geq d(r,y)+w(y,x').
        \]
        Combined with the triangle inequality $d(r,x')\leq d(r,y)+w(y,x')$, we obtain $d(r,x')=d(r,y)+w(y,x')$. By \cref{asmp:scout-Q}, we have $w(x,x')\geq 0$. Thus,
        \[
        d(r,x')=d(r,x)+w(x,x')\geq d(r,x)\geq -\Delta_r.
        \]
        \item $y$ is an out-Steiner vertex: This case is already handled.
        \item $y$ is an in-Steiner vertex with $\pin(y)=x$: We can return the path $\tilde r_t\to x\rightsquigarrow v$ where the $x\rightsquigarrow v$ segment is $Q[x,v]$. This is because when $y$ is processed in $\scout(r,t)$, in Step~\ref{step:scout-in}, for $W=d_y+w(y,x)+\Delta_r$, the algorithm adds an edge $(\tilde r_t,x)$ of weight $W$ if $W\geq 0$. Note that
        \[
            d(r,x)=d(r,y)+w(y,x)=d_y+w(y,x).
        \]
        Since $d(r,x)\geq -\Delta_r$, $W\geq 0$ must hold. Therefore, the edge $(\tilde r_t,x)$ of weight $W=d(r,x)+\Delta_r$ is added. The output path has weight
        \[
        w(\tilde r_t,x)+w(Q[x,v])=d(r,x)+\Delta_r+d(x,v)=d(r,v)+\Delta_r,
        \]
        and it contains $0$ hops.
    \end{itemize}
    \item If $x\in \Aout_y$: When vertex $y$ is processed in $\scout(r,t)$, in Step~\ref{step:scin-sort}, $\Aout_y=\{u_1,\dots,u_k\}$ is sorted and we computed $j$ and $\ell=\lceil \log(k-j+1)\rceil$. Since $x\in \Aout_y$, we may assume $x=u_i$. Now we consider three cases:
    \begin{itemize}
        \item $i\leq k-2^{\ell}$: In Step~\ref{step:out-down} of the procedure $\shortcut$, the algorithm adds an edge $(y_{\ell},x)$ of weight $w(y,x)-w(y,y_\ell)$. When vertex $y$ is processed in $\scout(r,t)$, in Step~\ref{step:scin-steiner}, the algorithm adds an edge $(\tilde r_t,y_{\ell})$ of weight $d_y+w(y,y_\ell)+\Delta_r$. We can therefore return the path $\tilde r_t\to y_\ell\to x\rightsquigarrow v$ where the $x\rightsquigarrow v$ segment is $Q[x,v]$. It has weight
        \begin{align*}
            & \quad \; w(\tilde r_t,y_{\ell})+w(y_\ell,y)+w(Q[x,v])\\
            & = d_y+w(y,y_\ell)+\Delta_r+w(y,x)-w(y,y_\ell)+d(x,v)\\
            & = d(r,x)+d(x,v)+\Delta_r\\
            & = d(r,v)+\Delta_r,
        \end{align*}
        and it contains $0$ hops.
        \item $i>k-2^{\ell}$ and $x$ is a light out-Steiner vertex: When vertex $y$ is processed in $\scout(r,t)$, in Step~\ref{step:scout-bf-out}, we have $W=d_y+w(y,x)-w(\pout(x),x)$. Consider two cases:
        \begin{itemize}
            \item {\bf(Case R3)} $W<-\Delta_r$: Then we can recursively invoke $\findpathout(x,\pout(x),Q)$ and return the resulting path. \cref{asmp:scout-y} is satisfied since in Step~\ref{step:scout-bf-out}, the algorithm must have updated $d_{\pout(x)}$ with $W$ and thus $\pout(x)$ is processed. In addition, by triangle inequality,
            \[
             d(r,\pout(x))\geq d(r,x)-d(\pout(x),x)=d(r,y)+w(y,x)-d(\pout(x),x)=W\geq d_{\pout(x)}.
            \]
            By \cref{lem:scout-dist-valid}, we have $d_{\pout(x)}\leq d(r,\pout(x))$. Therefore, $d(r,\pout(x))=d_{\pout(x)}$ holds. Since $W<-\Delta_r$, $d_{\pout(x)}\leq W<-\Delta_r$ also holds.
            \item {\bf(Case R4)} $W\geq -\Delta_r$: Since $x$ is a light out-Steiner vertex and $y\in \Nout(x)$ is not an out-Steiner vertex, by \cref{lem:out-light-edge}, $y$ must be an N-Steiner vertex ($y$ cannot be $\pout(x)$ since $x\in \Aout_y$). By \cref{lem:out-heavy-edge}~\ref{item:out-3}, there is an edge $(y,\pout(x))$ of weight $w(y,x)-w(\pout(x),x)$ in $\Fout$. Since $W=d_y+w(y,x)-w(\pout(x),x)\geq -\Delta_r$, we have
            \[
            w(y,x)-w(\pout(x),x)\geq -\Delta_r-d_y=-\Delta_r-d(r,y).
            \]
            Since $d(r,y)<-\Delta_r$, we see that $w(y,x)-w(\pout(x),x)\geq 0$. Therefore, by \cref{obs:f}~\ref{obs:fout}, the edge $(y,\pout(x))$ of weight $w(y,x)-w(\pout(x),x)$ is added into $G$. Consider the path $Q'=r\rightsquigarrow y\to \pout(x)\to x\rightsquigarrow v$, where the segment $x\rightsquigarrow v$ is $Q[x,v]$, and the segment $r\rightsquigarrow y$ is an arbitrary $(r,y)$ shortest path. We can recursively invoke $\findpathout(\pout(x),y,Q')$ and return the resulting path. We first show that \cref{asmp:scout-x} is satisfied. By triangle inequality,
            \begin{align*}
            d(r,\pout(x))& \geq d(r,x)-d(\pout(x),x)\\
            & = d(r,y)+w(y,x)-d(\pout(x),x)\\
            & \geq  d(r,y)+w(y,\pout(x)).
            \end{align*}
            Combined with the triangle inequality $d(r,\pout(x))\leq d(r,y)+w(y,\pout(x))$, we obtain $d(r,\pout(x))= d(r,y)+w(y,\pout(x))$. Since
            \[
            W=d_y+w(y,x)-w(\pout(x),x)=d(r,x)-w(\pout(x),x)\geq \Delta_r,
            \]
            we have
            \[
            d(r,\pout(x)) \geq d(r,x)-d(\pout(x),x)\geq \Delta_r.
            \]
            
            Now we show that \cref{asmp:scout-Q} is also satisfied. Indeed,
            \begin{align*}
                w(Q') & = d(r,y)+w(y,\pout(x))+w(\pout(x),x)+w(Q[x,v])\\
                & = d(r,y)+w(y,x)-w(\pout(x),x)+w(\pout(x),x)+d(x,v)\\
                & = d(r,y)+w(y,x)+d(x,v) \\
                & = d(r,v).
            \end{align*}
            In addition, $Q'[x,v]$ contains $0$ hops since $Q[x,v]$ contains $0$ hops.
        \end{itemize}
        \item $i>k-2^{\ell}$ and $x$ is not a light out-Steiner vertex: We can return the path $\tilde r_t\to x\rightsquigarrow v$, where the $x\rightsquigarrow v$ segment is $Q[x,v]$. This is because when vertex $y$ is processed in $\scout(r,t)$, in Step~\ref{step:scout-bf-other}, for $W=d_y+w(y,x)+\Delta_r$, the algorithm adds an edge $(\tilde r_t,x)$ of weight $W$ if $W\geq 0$. Note that
        \[
            d(r,x)=d(r,y)+w(y,x)=d_y+w(y,x).
        \]
        Since $d(r,x)\geq -\Delta_r$, $W\geq 0$ must hold. Therefore, the edge $(\tilde r_t,x)$ of weight $W=d(r,x)+\Delta_r$ is added. The output path has weight
        \[
        w(\tilde r_t,x)+w(Q[x,v])=d(r,x)+\Delta_r+d(x,v)=d(r,v)+\Delta_r,
        \]
        and it contains $0$ hops.
    \end{itemize}
\end{itemize}

Now we show that this algorithm terminates. The algorithm $\findpathout(x,y,Q)$ makes a recursive call in four cases, denoted by Case R1--R4. We now consider each of them:
\begin{itemize}
    \item Case R1: This case occurs when the second parameter $y$ is an out-Steiner vertex. In the recursive call, $y$ is replaced by $\pout(y)$. By \cref{obs:steiner}~\ref{obs:pout}, $\pout(y)$ is not an out-Steiner vertex. We note that once $y$ becomes non-out-Steiner vertex, Case R1 cannot occur again. This is because in Cases R2 and R4, $y$ remains unchanged, and in Case R3, $y$ is replaced by $\pout(x)$, which is also not an out-Steiner vertex by \cref{obs:steiner}~\ref{obs:pout}. Therefore, Case R1 can occur at most once at the beginning, after which only Cases R2–R4 may occur.
    \item Case R4: $\findpathout(\pout(x),y,Q')$ is called. We claim that this call terminates without further recursion. If a recursive call is made within it, only Cases R2--R4 is possible, and in each of these cases the first parameter $x$ must be an out-Steiner vertex. However, by \cref{obs:steiner}~\ref{obs:pout}, $\pout(x)$ is not an out-Steiner vertex. Hence, no recursive call is possible, and the call terminates immediately.
    \item Case R2: $\findpathout(x',y,Q)$ is called. In this case, $x'$ is replaced by the successor of $x$ in the path $Q$.
    \item Case R3: $\findpathout(x,\pout(x),Q)$ is called. In this case, $y$ is replaced by $\pout(x)$. As shown in the proof of \cref{lem:scout-dist-valid}, if $y$ is a level-$\ell$ vertex, then $\pout(x)$ must be a level-$\ell'$ vertex for some $\ell'<\ell$. 
\end{itemize}
Therefore, the algorithm terminates after a finite number of recursive calls.

Consider the segment $P[r,v]$, let $x$ be the first vertex on this segment with $d(r,x)=d^1(r,x)\geq-\Delta_r$ and let $y$ be the predecessor of $x$ with $d(r,y)=d^1(r,y)<-\Delta_r$. By \cref{lem:compute-U}, we have $x\not\in \Uout_r,y\in \Uout_r$. We can then invoke $\findpathout(x,y,P[r,v])$, it is straightforward to verify that this is a valid input. The returned path is the desired path. This completes the proof of \cref{clm:shortcut-out}.
\end{document}